\newif\ifshowproofs
\showproofsfalse

\newif\iffixlater
\fixlaterfalse

\newif\iftr
\trtrue
\iftr\showproofstrue\fi

\newcommand{\commentout}[1]{} 

\newcommand{\numof}{\#}

\iftr
\documentclass{article} 
\else
\documentclass{ieee}
\fi

\iftr
\newenvironment{proof}{{\bf Proof:} }{\qed}
\fi

\usepackage{amsmath,amssymb}
\usepackage{graphicx}
\usepackage{xcolor}
\usepackage{ifthen}
\usepackage[square]{natbib}
\usepackage{enumerate} 
\usepackage{centernot} 
\usepackage[bookmarks=true,pdfpagelabels=true]{hyperref}
\hypersetup{
    bookmarksnumbered=true,         
    plainpages=false,
    final=true,
    unicode=false,          %
    breaklinks=true,
    colorlinks=false,       %
    pdfborder={0 0 0},      %
}

\newlength{\diagsize}
\newcommand{\flowsto}{\ensuremath{\rightarrowtail}}
\newcommand{\obsflowsto}{\ensuremath{\raisebox{-.35ex}[0pt][0pt]{$\stackrel{o}{\rightarrowtail}$}}}
\newcommand{\obsflowstoin}[2]{\ensuremath{#1\obsflowsto{#2}}}
\newcommand{\flowstoin}[3][]{\flowstoinf[#1]{#2}{#3}{\flowsto}}
\newcommand{\flowstoinf}[4][]{\ensuremath{#2 %
\ifthenelse{\equal{#1}{}}{#4}{\raisebox{-.35ex}{$\stackrel{\mbox{\raisebox{-.5ex}[1ex][0pt]{\scriptsize $#1$}}}{#4}$}}%
 #3}}
\newcommand{\notflowstoin}[3][\flowsto]{\ensuremath{#2{\centernot{#1}}#3}}
\newcommand{\machinespec}[1]{\ensuremath{\mathtt{#1}}}
\newcommand{\obs}{\machinespec{obs}}
\newcommand{\step}{\machinespec{step}}
\newcommand{\dom}{\machinespec{dom}}

\newcommand{\run}[2]{\ensuremath{#1 \cdot #2}}
\newcommand{\ta}[1]{\ensuremath{\mathtt{ta}_{#1}}}

\newcommand{\TFF}{FTA}
\newcommand{\tff}[1]{\ensuremath{\mathtt{fta}_{#1}}}

\newcommand{\view}[1]{\ensuremath{\mathtt{view}_{#1}}}

\newcommand{\sysvar}{\ensuremath{M}}

\newcommand{\archvar}{\ensuremath{\mathcal{A}}}

\newcommand{\concat}{\ensuremath{\,\hat{~}}\,}
\newcommand{\aconcat}{\ensuremath{\circ}}

\newcommand{\compose}{\ensuremath{\circ}}

\newcommand{\restrict}{\upharpoonright}

\newcommand{\Nat}{\ensuremath{\mathbb{N}}}

\newcommand{\note}[1]{{\color{blue}{\bf [[ #1 ]]}}}

\newcommand{\rimp}{\Rightarrow} 
\newcommand{\knows}[2][K]{\ensuremath{{#1}_{#2}}}

\newcommand{\sat}{\ensuremath{\vDash}}
\newcommand{\interpsysf}[3]{#1, #2, #3}
\newcommand{\interpsys}{\interpsysf{\sysvar}{\pi}{\alpha}}
\newcommand{\interpsysvf}[2]{#1, #2}
\newcommand{\interpsysv}{\interpsysvf{\sysvar}{\pi}}

\newcommand{\refines}{\leq}
\newcommand{\srefines}{\preceq}
\newcommand{\trefines}{\sqsubseteq}

\newcommand{\source}{{\cal S}}
\newcommand{\intn}{{\tt  I}}

\newcommand{\intarch}{{\cal AI}}

\newcommand{\accesscontrolname}[1]{\ensuremath{\mathtt{#1}}}
\newcommand{\contents}{\accesscontrolname{contents}}
\newcommand{\alter}{\accesscontrolname{alter}}
\newcommand{\observe}{\accesscontrolname{observe}}

\newtheorem{definition}{Definition}
\newtheorem{lemma}{Lemma}
\newtheorem{theorem}{Theorem}
\newtheorem{proposition}{Proposition}

\newcommand{\fdefault}{f^\mathit{def}}

\newenvironment{noshow}{\setbox0=\vbox\bgroup}{\egroup}

\newenvironment{example}{\medskip \noindent {\bf Example:}}{\hfill$\Box$\medskip}

\renewenvironment{proof}%
{\ifshowproofs\noindent\emph{Proof:}~\else\begin{noshow}\fi}%
{\ifshowproofs\hfill$\Box$\medskip\else\end{noshow}\fi}
\newenvironment{simpleproof}{\medskip \noindent\emph{Proof:}~}{\hfill$\Box$\medskip}
\newenvironment{proofof}[1]{\medskip\noindent\textbf{Proof of {#1}:}~}{\hfill$\Box$\medskip}

\newcommand{\inff}[1]{\mathtt{T}_{#1}}

\title{Using Architecture to Reason About Information Security%
\footnote{Version of \today. An early version of this paper circulated 
in 2009 and has been cited in the literature. The present version 
includes a significant reworking of results on architectural refinement from that version. 
A version of the paper introducing the access control implementation results 
was presented at the Layered Assurance Workshop at ACSAC 2012. The present 
version adds to these earlier versions full proofs of all results and an extended set of 
examples.}}
\author{Stephen Chong \\ 
  {Harvard University} 
\and Ron van der Meyden \\
{UNSW Australia}}

\date{}

\begin{document}

\maketitle

\begin{abstract}
We demonstrate, 
by a number of examples, that 
information-flow security properties can be proved 
from abstract architectural descriptions, that describe 
only
the causal structure of a system and local properties of 
trusted components. 
We specify these architectural descriptions of systems 
by generalizing intransitive noninterference policies to 
admit the ability to filter information passed between communicating domains.
A notion of refinement of such 
system
architectures 
is developed that 
supports top-down development of architectural specifications 
and proofs by abstraction of information security properties.    
We also show that, in a  concrete setting where the causal structure is enforced
by access control, a static check of the access control setting plus local verification of the 
trusted components is sufficient to prove that a generalized intransitive noninterference 
policy is satisfied.
\end{abstract}

\iftr
\tableofcontents
\fi

\section{Introduction}
 
System architectures are high-level designs that describe the overall
structure of a system in terms of its components and their
interactions. 
Proposals for architectural modeling languages 
(e.g., 
\citet{aadlv2} and Acme \citep{Garlan00AcmeChapter}) vary with
respect to their level of detail and contents, but at the most
abstract level, architectures 
specify the {\em 
causal 
structure
of a system}.

The MILS (Multiple Independent Levels of Security and Safety)
initiative \citep{AHOT05,vBCLTU,BDRF08} of the US Air Force proposes to
use architecture as a key part of the assurance case for
high-assurance systems.  The details of the MILS vision are still
under development but, as articulated by \citet{BDRF08}, it encompasses
a 2-level design process, consisting of a policy level and a resource
sharing level. 

At the policy level, the system is described by an architecture in the
form of a graph, in which vertices correspond to components and the
edges specify 
permitted communication between 
components.  In this respect, the architecture is like an {\em
intransitive noninterference} security policy
\citep{HaighY87,Rushby_92,wiini}. At the policy level, one might also specify
which components are trusted, and the local policies that these
components are trusted to enforce. 
According to the MILS vision, building a system according to the
architecture, by composing components that satisfy their local
policies, should result in 
the system satisfying {\em global} security and safety properties.

At the resource sharing level, MILS envisages the use of a range of
infrastructural mechanisms to ensure that the architectural
information-flow policy is enforced despite components sharing
resources such as processors, file systems, and network links.  
These mechanisms might include physical isolation, separation
kernels, periods processing, cryptography and separating network
infrastructure.  
It is intended that this infrastructure will be developed to a high level of
assurance, so that a systems assurance case can be obtained by the
composition of the assurance cases for trusted components and systems
infrastructure. It is hoped this will enable a COTS-like market for
infrastructural mechanisms and trusted components.

The key contribution of this paper is to demonstrate, through several
examples, that it is in fact 
possible, as envisaged in the MILS literature, to derive interesting
information security properties compositionally from a high-level
specification of trusted components and their
architectural structure. 
We
focus
 on compositional reasoning about in\-for\-ma\-tion-fl\-ow security
properties.

We present a framework that allows the specification of a system
architecture with local constraints on some system components. 
To give a precise meaning to the architectural structure, we 
extend 
the semantics for intransitive noninterference developed by
\citet{wiini}.  An architectural interpretation of this semantics has
previously been given 
\citep{arch-refinement}.  
In order to express constraints on trusted components, we 
extend architectures by labeling edges between components with functions
that further restrict the information permitted to flow along edges.
One of the contributions of the paper is to give a formal semantics to the
enriched 
architectures that include these new types of edges.
We also develop a theory of refinement for these enriched architectures,
which enables top-down, correctness-preserving development of 
architectural specifications.
It also
enables simple proofs of information security properties on complex 
architectures to be obtained 
using
an abstraction of that 
architecture. 

We demonstrate the use of the framework through examples motivated by
systems with interesting security requirements.
These include 
multi-level secure databases,
the Starlight Interactive Link
\citep{starlight-interactive},
a trusted downgrader, 
and a simple electronic election system.

In each example, we identify an architectural structure and a
mathematically precise set of local 
constraints
 on the trusted
components.  We then show that information-flow properties expressed
in a logic of knowledge 
arise
as a consequence of the interaction of the local 
constraints and the architectural structure.
Our results show that for {\em any} system that is compliant with the
architecture, if the trusted components satisfy their local 
constraints
then the system 
 satisfies the global information-flow
properties.

The information security properties presented in the examples provide information-theoretic and application-specific guarantees. Thus, there are no covert channels that can violate the information security properties, and  the negation of each information security property would constitute an application-specific attack.

Only a few 
examples have been presented to date to formally justify the MILS 
approach to high-assurance secure systems development. One example 
is developed in \citet{GreveWV03}, but with respect to 
a more concrete model (based on a separation kernel formal security 
policy that deals with access control on memory segments) than the abstract, ``noninterference''
style semantics we consider. 
Our policy level model is more abstract, and allows greater flexibility for implementations. 
However, we also consider a more concrete model, 
systems with structured state subject to ``reference monitor conditions"  \citep{Rushby_92}. 
We show that in this setting, to prove that a system complies with one of our extended architectures, 
it suffices to check a simple condition on the access control setting and to  prove local 
properties of the trusted components. 

By developing an abstract semantics for architectures and specifications of trusted
components, and by developing additional examples, our work 
advances the case that 
global information-flow security properties can be  derived from a high-level systems 
architecture and local constraints on trusted components within this architecture, 
in the style of reasoning envisaged by \citet{BDRF08}.

The structure of the paper is as follows. In
Section~\ref{sec:archandsem} we review architectures 
and 
their semantics.
In Section~\ref{sec:infosec}, we introduce the epistemic logic we use
to express information security properties.  In
Section~\ref{sec:extended} we extend architectures with \emph{filter
  functions} that allow fine-grained specification of what information
flows between components. The extended architectures enable the proof
of additional information security properties.
Section~\ref{sec:refine} extends 
previous work on 
architectural refinement to account
for filter functions. 
The concrete model based on access control is developed in
Section~\ref{sec:access-control} and 
we consider possible platforms and techniques that might be used to
show the access control model holds in Section~\ref{sec:enforce}.
We discuss related work in Section~\ref{sec:related}.
Section~\ref{sec:conclusion} concludes.
We use examples throughout, to illustrate and
motivate the definitions and results.

\section{Architectures and semantics}\label{sec:archandsem}

Architectures give a policy level description of the 
structure of a system. 
We begin with a simple notion of architecture, following \citet{arch-refinement}. 
A richer notion will be introduced later. 

An architecture is a pair $\archvar=(D, \flowsto)$, where $D$ is a set
of security domains, and the binary relation $\flowsto \subseteq D
\times D$ is an \emph{information-flow policy}. The relation
$\flowsto$ is reflexive but not necessarily transitive. Intuitively,
information is allowed to flow from domain $u$ to domain $v$ only if
$\flowstoin{u}{v}$. The relation is reflexive as it is assumed that
information flow within a domain
cannot be prevented, so 
 is always allowed.

\iftr 
In the literature on information flow policies, 
domains are generally understood to correspond to security levels. 
We use a more general interpretation, in which 
domains can may also correspond to system components or agents in the system. 
In an implementation of an
architecture, separate domains do not necessarily utilize separate
resources. Hardware, code, and data may be shared between
domains. Indeed, a key challenge is ensuring that the information-flow
policy is respected despite the shared use of resources.
\fi 

\newcommand{\hlarch}{\ensuremath{\mathcal{HL}}}

\subsection{Example: {\hlarch} architecture}
The architecture $\hlarch = (\{{H,L\}}, \{(L,L),$ $(H,H),$ $(L,H)\})$ consists
of two security domains $H$ and $L$, and the in\-for\-ma\-tion-flow policy
indicates that information is allowed to flow from $L$ to $H$, in
addition to the reflexive information flows. We can depict $\hlarch$
graphically, indicating security domains with rectangles, and the
information-flow policy with arrows. We omit arrows for reflexive
information flows.

\begin{center}
  \includegraphics[height=\diagsize]{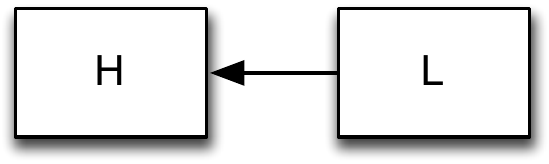}
\end{center}

\subsection{Example: Hinke-Schaefer} \label{sec:hinke-schaefer}

A variety of architectures have been proposed for 
multi-level secure database management systems (MLS/DBMS) \citep{Thuraisingham05}. 
In the Hinke-Schaefer architecture \citep{HinkeS75}, 
several (untrusted) single-level DBMSs are composed together in
a trusted operating system. Each user interacts with a single-level
DBMS. The operating system enforces access control between the
single-level DBMSs, allowing more restrictive DBMSs to read the
storage files of less restrictive DBMSs, but not vice versa.

\newcommand{\hsarch}{\ensuremath{\mathcal{HS}}}
\newcommand{\Hdbms}{\ensuremath{H_{\mathit{DBMS}}}}
\newcommand{\Ldbms}{\ensuremath{L_{\mathit{DBMS}}}}
\newcommand{\Huser}{\ensuremath{H_{\mathit{user}}}}
\newcommand{\Luser}{\ensuremath{L_{\mathit{user}}}}

The following diagram shows architecture $\hsarch$, 
which represents the Hinke-Schaefer architecture for two security levels
at the  MILS policy level.
\begin{center}
  \includegraphics[height=3.5\diagsize]{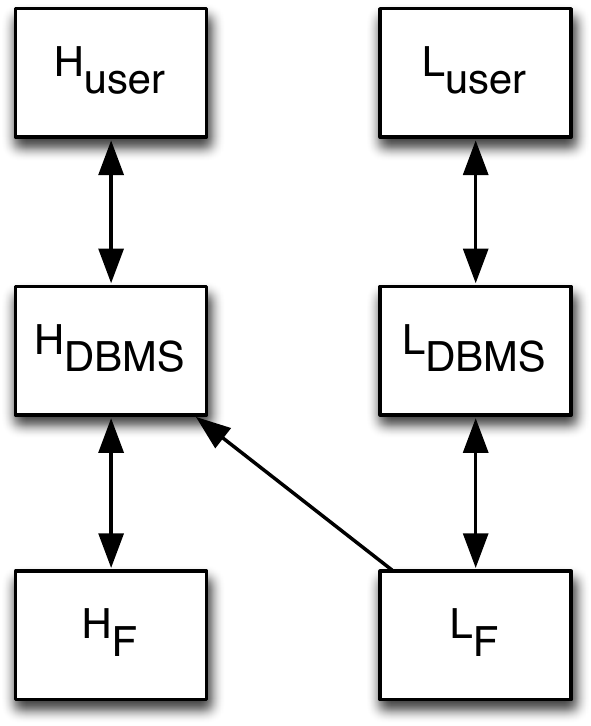}
\end{center}
Domains $\Huser$ and $\Luser$ represent users of a
high-security 
and
low-security DBMS respectively; they interact
with the single-level DBMSs $\Hdbms$ and $\Ldbms$
respectively. The single-level DBMSs store their data in database
files denoted $H_{F}$ and $L_{F}$. Note that information is allowed to
flow to $\Hdbms$ from both $H_{F}$ and $L_{F}$, as the high-security
DBMS is allowed to read the storage files of both the high-security
and low-security DBMSs.

The Hinke-Schaefer architecture is also known as the ``operating
system providing mandatory access control'' architecture 
\citep{Thuraisingham05}, as the operating system is trusted to
enforce the information flows specified in the architecture.
This amounts to a decision to implement the policy level 
architecture $\hsarch$ at the resource sharing level 
by means of a trusted separation kernel.

\subsection{Machine model}

To specify what it means for an implementation to satisfy an
architecture, we must first define what an implementation is.  We use
the \emph{state-observed} machine model~\citep{Rushby_92}, which
defines deterministic state-based machines. 
A machine has a set of actions $A$, and each
action is associated with a security domain. Intuitively, if action
$a$ is associated with domain $u$, then $a$ represents a decision,
choice, or action taken by the system component represented by $u$.
Actions deterministically alter the state of the machine, and we
assume that the observations of each security domain are determined by
the current machine state.

A machine is a tuple
$\sysvar =\langle S, s_0, A,D, \step, \obs, \dom\rangle$ 
where $S$ is a set of states, $s_0\in S$ is the
\emph{initial state}, $A$ is a set of \emph{actions}, 
$D$ is a set of domains, 
$\step:S \times
A \rightarrow S$ is a deterministic transition 
function, 
$\dom: A
\rightarrow D$ associates a domain with each action, and
observation function $\obs: D \times S \rightarrow O$ describes for
each state what observations can be made by each domain, for some set
of observations $O$.

We assume that it is possible to execute any action in any state: the
function $\step$ is total.  Given sequence of actions $\alpha \in
A^*$, we write $\run{s}{\alpha}$ for the state reached by performing
each action in turn, starting in state $s$. We define
$\run{s}{\alpha}$ inductively defined using the transition function
$\step$, by 
\begin{align*}
 \run{s}{\epsilon} &= s \\
\run{s}{\alpha a} &= \step(\run{s}{\alpha}, a)~
\end{align*}
for $\alpha \in A^*$ and $a\in A$. 
 (Here $\epsilon$ denotes the empty sequence.)
For notational convenience, we write $\obs_u$ for the function
$\obs(u, \cdot)$, and $\obs_u(\alpha)$ for
$\obs_u(\run{s_0}{\alpha})$, where $\alpha \in A^*$.

Given a sequence $\alpha \in A^*$, the \emph{view} of 
a group of domains $G$ 
of
$\alpha$ is the sequence of 
the group's 
observations and the actions that
belong to 
members of the group. 
Intuitively, 
$G$'s 
view is the history of its
observations and the actions it has performed.
The function $\view{G}$ defines the view of domain $G$.
We first define the observation of group $G$ at state $s$ by 
$\obs_G(s) = \langle \obs_u(s)\rangle_{u\in G}$, i.e. the tuple of 
observations of individuals $u\in G$. 
The view function is then 
defined inductively by 
\begin{align*}
  \view{G}(\epsilon) &= \obs_G(s_0) \\ 
  \view{G}(\alpha a) &= 
\begin{cases}
\view{G}(\alpha) \,a\, 
\obs_G(
{\alpha a})  & \text{if }\dom(a) \in G\\   
\view{G}(\alpha) \aconcat \obs_G(
{\alpha a})  & \text{otherwise}  ~
\end{cases}
\end{align*}
for $\alpha \in A^*$ and $a\in A$. 
To capture that the semantics is asynchronous and insensitive to stuttering of observations, the 
definition uses the 
absorptive 
concatenation operator
$\aconcat$: for set $X$, sequence $\alpha \in X^*$, and element $x\in
X$, $\alpha \aconcat x = \alpha$ if $x$ is equal to the last element
of $\alpha$, and $\alpha \aconcat x = \alpha x$ otherwise.
When $G = \{u\}$ is a singleton, 
we write $\view{u}$ for $\view{G}$.

Finally, for any sequence of actions $\alpha \in A^*$, we write $\alpha
\restrict G$ for  the subsequence of $\alpha$ of actions whose domain is
in the set $G$.

\subsection{Semantics}

A machine satisfies an architecture if, in all possible executions of
the machine, information flow 
is in accordance with the
architecture's information-flow policy.  
We formalize this using an approach proposed by \citet{wiini}, 
which involves the use of a concrete operational model to  define an upper bound
on the information that a domain is permitted to learn.

The operational model is captured using a function $\ta{u}$, 
which maps a sequence of actions $\alpha \in A^*$ to a representation of 
the maximal information that domain $u$ is permitted to have after $\alpha$, 
according to the policy $\flowsto$. 
(Term ``ta'' is derived from \emph{transmission} of information about \emph{actions};
the definition corrects problems identified by \citet{wiini} with earlier
``intransitive purge'' based semantics \citep{Rushby_92}.) 

An action of $v$ should convey information to $u$
only if $\flowstoin{v}{u}$. Moreover, the information conveyed should
be no more than the information that $v$ is permitted to have. Given
machine 
$\sysvar =\langle S, s_0, A, D, \step, \obs, \dom\rangle$,
function $\ta{u}$ is defined inductively by $\ta{u}(\epsilon) =
\epsilon$, and, for $\alpha \in A^*$ and $a \in A$,
\begin{align*}
  \ta{u}(\alpha a) &= 
  \begin{cases}
    \ta{u}(\alpha) & \text{if } \notflowstoin{\dom(a)}{u} \\
    (\ta{u}(\alpha), \ta{\dom(a)}(\alpha), a) & \text{otherwise}~. 
  \end{cases}
\end{align*}

Note that if information is not allowed to flow from $\dom(a)$ to $u$,
then $\ta{u}(\alpha a) = \ta{u}(\alpha)$, 
i.e., the maximal information permitted to $u$ does not change. 
If information is allowed to
flow from $\dom(a)$ to $u$, then the information conveyed is at most
the information that domain $\dom(a)$ is permitted to have
(i.e., $\ta{\dom(a)}(\alpha)$), 
and the action $a$ that was performed.
Thus, in this case we add the tuple $(\ta{\dom(a)}(\alpha),a)$ to the maximal information 
$\ta{u}(\alpha)$ that $u$ was permitted to have before the action $a$ was performed.

A machine is 
\emph{TA-compliant} with an architecture if it has an appropriate set of domains, and 
for each domain $u$, what $u$ observes in
state $\run{s_0}{\alpha}$ is determined by $\ta{u}(\alpha)$. That is,
$\ta{u}$ describes the maximal information that $u$ may learn:
if in two runs $\alpha$ and $\alpha'$ the maximal information that $u$ may learn is identical ($\ta{u}(\alpha) = \ta{u}(\alpha')$), then $u$'s observations in each run must be identical ($\obs_{u}(\alpha) = \obs_{u}(\alpha')$).

\begin{definition}[TA-compliance] 
A system 
$\sysvar$ 
  is TA-com\-pli\-ant with 
  architecture $(D,
  \flowsto)$ if 
  it has domains $D$ and for all $u\in D$ 
  and all sequences $\alpha, \alpha'
  \in A^*$ such that $\ta{u}(\alpha) = \ta{u}(\alpha')$, we have
  $\obs_{u}(\alpha) = \obs_{u}(\alpha')$.
\end{definition}

\iftr

TA-compliance requires that if $\ta{u}(\alpha) = \ta{u}(\alpha')$ then
the observations of $u$ in state $\run{s_0}{\alpha}$ and in state
$\run{s_0}{\alpha'}$ are equal. The following lemma shows that in fact
TA-compliance implies that if $\ta{u}(\alpha) = \ta{u}(\alpha')$ then
$\view{u}(\alpha) = \view{u}(\alpha')$.

\begin{lemma}\label{lem:ta-view}
  If $\sysvar$ is TA-compliant with respect to architecture
  $(D, \flowsto)$, then for all agents $u$ and all $\alpha,
  \alpha' \in A^*$ such that $\ta{u}(\alpha) = \ta{u}(\alpha')$ we
  have $\view{u}(\alpha) = \view{u}(\alpha')$.
\end{lemma}

\begin{proof}
  By induction on $|\alpha| + |\alpha'|$. The base case is
  trivial. Suppose that the result holds for all sequences of shorter
  combined length, and consider strings $\alpha a$ and $\alpha'$ such that 
  $\ta{u}(\alpha a) = \ta{u}(\alpha')$.

We consider several cases. 

\begin{itemize}
\item If $\notflowstoin{\dom(a)}{u}$ then $\ta{u}(\alpha) =
  \ta{u}(\alpha a) = \ta{u}(\alpha')$, so $\view{u}(\alpha) =
  \view{u}(\alpha')$, by the induction hypothesis, and
  $\obs_{u}(\alpha a) = \obs_u(\alpha)$, by TA-compliance.  By
  reflexivity of $\flowsto$, we have $\dom(a) \neq u$, so it follows
  that
  \begin{align*}
    \view{u}(\alpha a) &= \view{u}(\alpha) \aconcat \obs_{u}(\alpha a) \\
 &= \view{u}(\alpha)\\
&= \view{u}(\alpha')
  \end{align*}

\item If $\flowstoin{\dom(a)}{u}$ 
  then $\ta{u}(\alpha a) = (\ta{u}(\alpha), \ta{\dom(a)}(\alpha), a) =
  \ta{u}(\alpha')$. This means that $\alpha' \ne \epsilon$. Assume
  $\alpha' = \beta b$.  Without loss of generality, we may assume
  $\flowstoin{\dom(b)}{u}$, since otherwise we may apply the previous
  case with the roles of $\alpha a$ and $\beta b$ reversed.  Let
  $\ta{u}(\beta b) = (\ta{u}(\beta), \ta{\dom(b)}(\beta), b)$.  Since
  $\ta{u}(\alpha a) = \ta{u}(\beta b)$, we have $a = b$, and by TA-compliance we have
  $\obs_u(\alpha a) = \obs_u(\beta b)$. Also, $\ta{u}(\alpha) =
  \ta{u}(\beta)$, and so by the inductive hypothesis,
  $\view{u}(\alpha) = \view{u}(\beta)$.
  \begin{itemize}
  \item If $\dom(a) \ne u$, then
  \begin{align*}
    \view{u}(\alpha a) &= \view{u}(\alpha) \aconcat \obs_u(\alpha a) \\ 
       &= \view{u}(\beta) \aconcat \obs_u(\beta b) \\
       &= \view{u}(\beta b).
 \end{align*}

 \item If $\dom(a) = u$ then 
  \begin{align*}
    \view{u}(\alpha a) &= \view{u}(\alpha)\, a\, \obs_u(\alpha a) \\ 
       &= \view{u}(\beta) \,b\, \obs_u(\beta b) \\
       &= \view{u}(\beta b).
 \end{align*}
\end{itemize}
\end{itemize}
\end{proof}
\fi %

\section{Information security properties}\label{sec:infosec}

We 
use 
a 
(fairly standard) 
propositional epistemic logic \citep{fhmv95} 
 to express information
security properties. 
The syntax 
is defined as follows:
\begin{align*}
\phi, \psi &::= \top \mid p \mid \neg \phi \mid \phi \wedge \psi \mid
\knows{G} \phi \\
G & \text{ ranges over groups of domains.} 
\end{align*}
In case $G= \{u\}$ is a singleton, we write simply $\knows{u}\phi$ for $\knows{G}\phi$. 

Formulas $\top$, $p$, $\neg \phi$, and $\phi \wedge \psi$ are standard
from propositional logic: $\top$ is always satisfied, and $p$ is a
propositional constant. 
Epistemic formula 
$K_G\phi$ says that the group of domains $G$, considered as a single domain, knows $\phi$.

Formulas are interpreted using a possible worlds semantics, where a
world is a sequence of actions $\alpha \in A^*$. 
A {\em proposition} is a set $X\subseteq A^*$.  We say proposition $X$
is \emph{non-trivial} if $X\ne \emptyset$ and $X\ne A^*$.  An
\emph{interpretation function} $\pi$ is a function from propositional
constants to propositions. 

We define the semantics of the logic using satisfaction relation
$\interpsys \sat \phi$, which intuitively means that formula $\phi$ is
true given interpretation function $\pi$, and machine $\sysvar$ that
has executed sequence $\alpha \in
A^*$. Figure~\ref{fig:satisfaction-relation} defines 
relation 
$\interpsys \sat \phi$.  We write $\interpsysv \sat \phi$ if
for all $\alpha \in A^*$ we have $\interpsys \sat \phi$.
We say that $\phi$ is {\em valid} if $\interpsysv \sat \phi$ for all 
systems $M$ and interpretations $\pi$.

\begin{figure}
\[
\begin{array}{ll}
\interpsys \sat \top \\
\interpsys \sat p
    &\text{iff } \alpha \in \pi(p) \\
\interpsys  \sat \neg \phi
    &\text{iff } \interpsys \centernot\sat \phi\\
\interpsys  \sat \phi \wedge \psi
    &\text{iff } \interpsys \sat \phi \text{ and } \interpsys \sat \psi \\
\interpsys  \sat 
\knows{G} \phi    &\text{iff } \interpsysf{\sysvar}{\pi}{\alpha'} \sat \phi \text{ for all } 
 \alpha' \in  A^* \text{ s.t. } \alpha \approx_G\alpha' 
\end{array}\]
  \caption{Epistemic Logic Semantics}
  \label{fig:satisfaction-relation}
\end{figure}

To interpret 
epistemic formulas
$\knows{G}\phi$, 
we 
use 
an indistinguishability relation for
each group of domains $G$, that describes what sequences of actions $G$ considers
possible given its view of the actual sequence of actions. Two
sequences of actions $\alpha \in A^*$ and $\alpha' \in A^*$ are
indistinguishable to group of domains $G$, written $\alpha \approx_{G} \alpha'$,
if $G$'s 
views
of the two sequences are identical:
$\alpha \approx_{G}\alpha' \iff 
\view{G}(\alpha) = \view{G}(\alpha').$
\subsection{Group knowledge vs. Distributed knowledge}

We note that 
the 
notion of group knowledge 
$\knows{G}$ 
differs from 
{\em distributed knowledge} \citep{fhmvbook},
the notion most commonly used in the literature on 
epistemic logic for the knowledge that a group would have if they pooled
their local information. 
The distributed knowledge operator $\knows[D]{G}$ for group $G$ is
given semantics  by 
\begin{align*}
\interpsys  \sat \knows[D]{G} \phi
    &~ \text{iff } \interpsysf{\sysvar}{\pi}{\alpha'} \sat \phi \text{ for all } \alpha' \in  A^* \text{ s.t. } \alpha \approx^{D}_G \alpha' 
\end{align*}
using a different indistinguishability relation
$ \approx^{D}_G$, defined as the intersection of  $\approx_u$  for $u \in G$. 
The reason we use group knowledge is that it proves to have a stronger relationship to 
a type of architectural abstraction  that we consider below. 
The two notions are related by the following 
result.

\begin{lemma} 
For $u\in G$, the formulas $\knows{u}\phi \rimp \knows[D]{G}\phi$ and 
$\knows[D]{G}\phi \rimp \knows{G}\phi$ are valid. 
\end{lemma} 

\begin{proof}
That $\knows{u}\phi \rimp \knows[D]{G}\phi$ is valid is a well-known fact 
of epistemic logic. It follows simply from the definitions 
by noting that if $M,\pi, \alpha \models  \neg  \knows[D]{G}\phi$ 
then there exists $\alpha' \approx^D_G \alpha$ such that
$M,\pi, \alpha' \models  \neg \phi$. Since $\alpha' \approx^D_G \alpha$
implies $\alpha' \approx_u \alpha$, we conclude $M,\pi, \alpha \models  \neg  \knows{u}\phi$.  

For $\knows[D]{G}\phi \rimp \knows{G}\phi$, we similarly consider the contrapositive. 
If $M,\pi, \alpha \models  \neg  \knows{G}\phi$
then there exists $\alpha' \approx_G \alpha$ such that 
$M,\pi, \alpha' \models  \neg \phi$. We claim that 
$\alpha' \approx_G \alpha$ implies $\alpha' \approx_u \alpha$ for all 
$u \in G$. It then follows that $\alpha' \approx^D_G \alpha$, and hence 
$M,\pi, \alpha \models  \neg  \knows[D]{G}\phi$.  

To prove the claim, we show that there exists a function $h$ such that for any $\beta \in A^*$,
$h(\view{G}(\beta)) = \view{u}(\beta)$. The domain of this function is the set of possible views of $G$ 
in $M$. These are nonempty sequences of actions $a$ with $\dom (a) \in G$ and tuples 
$\langle \obs_u(s)\rangle_{u\in G}$ of observations. The first and last element of 
such a sequence is a tuple of observations, and the sequence does not contain  
any adjacent actions.   We define the function $h$ on such sequences as follows,
where $\delta$ ranges over sequences produced by $\view{G}$, $a \in A$, 
and $t$ is a tuple of observations indexed by $G$. We write $t_u$ for the component 
of $t$ corresponding to $u\in G$. (In particular, note that  if $t = \obs_G(s)$, then 
$t_u = \obs_u(s)$.)
  \begin{align*}
    h(t) &= t_u \\
     h(\delta \, a \, t) &= 
     \begin{cases}
        h(\delta) \, a\, t_u  & \text{if }\dom(a) = u\\
        h(\delta) \aconcat t_u  & \text{otherwise}
     \end{cases}\\
     h(\delta \, t) &= h(\delta) \aconcat t_u
  \end{align*}

  We prove that $h(\view{G}(\beta)) = \view{u}(\beta)$ by induction
  on $|\beta|$. The base case is trivial. Suppose that
  $h(\view{G}(\beta)) = \view{u}(\beta)$ and consider $\beta b$. 
  
  Suppose first that $\dom(b) \not \in G$, so also $\dom(b) \neq u$. 
  Then, $\view{G}(\beta b) = \view{G}(\beta) \aconcat \obs_G(\beta b)$. 
  \begin{itemize} 
  \item If $\obs_G(\beta b) = \obs_G(\beta)$, then 
  $\view{G}(\beta b)) = \view{G}(\beta)$, and 
  also $\obs_u(\beta b) = \obs_u(\beta)$. 
  Moreover, $\view{u}(\beta b) = \view{u}(\beta) \aconcat \obs_u(\beta b) = \view{u}(\beta)$.  
  Thus,   using the induction hypothesis, we have 
   $h(\view{G}(\beta b)) = h(\view{G}(\beta)) = \view{u}(\beta) = \view{u}(\beta b)$. 
  \item 
  Otherwise, if  $\obs_G(\beta b) \neq  \obs_G(\beta)$, then 
  $\view{G}(\beta b) = \view{G}(\beta)\, \obs_G(\beta b)$, 
  so  $h(\view{G}(\beta b)) = h(\view{G}(\beta))\aconcat \obs_u(\beta b) = 
  \view{u}(\beta)\aconcat \obs_u(\beta b) = \view{u}(\beta b)$. 
\end{itemize} 
  
 Alternately, suppose $\dom(b)\in  G$. Then $\view{G}(\beta b) = \view{G}(\beta)
  \, b \, \obs_{G}(\beta b)$. 
  If $\dom(b) = u$ then
  $h(\view{G}(\beta b)) = h(\view{G}(\beta)) \, b \,
  \obs_u(\beta b)$. If $\dom(b) \ne u$ then $h(\view{G}(\beta b)) =
  h(\view{G}(\beta)) \aconcat \obs_u(\beta b)$. Either way, by the
  inductive hypothesis, we have $h(\view{G}(\beta b)) =
  \view{u}(\beta b)$.

  Thus we have $h(\view{G}(\beta)) = \view{u}(\beta)$ for any
  sequence $\beta$. Let $\alpha,\alpha' \in A^*$, and assume
  $\view{G}(\alpha) = \view{G}(\alpha')$. Then
  \begin{align*}
 \view{u}(\alpha) &=
 h(\view{G}(\alpha))\\
 &=  h(\view{G}(\alpha'))\\
 &= \view{u}(\alpha').
  \end{align*}
\end{proof}

The converse relationship 
$\knows{G}\phi \rimp \knows[D]{G}\phi$ is not valid. 
For example, consider a system $M$ with exactly 
two domains $u,v$, which both make observation $\bot$ at all states. 
Let $G = \{u,v\}$. Consider the proposition $p$ with $\pi(p)$ consisting of all sequences in which 
there is an action of domain $u$ that precedes any action of domain $v$. 
Let $\alpha = a_u a_v$ and $\alpha' = a_v a_u$ 
where $a_u$ is an action of $u$ and $a_v$ is an action of $v$. 
Then $M,\pi,\alpha \models \knows{G}p$, since 
$\view{G}(\alpha) = \view{G}(\beta)$ implies that $\beta = \alpha$. 
However, we have $\view{u}(\alpha) = \bot a_u \bot = \view{u}(\alpha')$, 
and similarly for domain $v$, so $\alpha \approx^D_G\alpha'$. Since $M, \pi, \alpha' \not  \models  p$, 
we obtain that $M,\pi,\alpha \not \models \knows[D]{G}p$.%
\footnote{We remark that the example relies upon the assumption of asynchrony: it 
can be shown that in synchronous  systems we have $\knows{G}\phi \equiv \knows[D]{G}\phi$.}

\subsection{Knowledge and architectural refinement}

The value of 
using 
the less common notion $\knows{G}\phi$ of group knowledge 
rather than 
distributed knowledge $\knows[D]{G}\phi$ 
is 
that it 
captures the way that knowledge properties are 
preserved under  a particular type of architectural abstraction. 
Given a system $\sysvar= \langle S, s_0, A,D_1, \step, \obs, \dom\rangle$
and a surjective mapping $r: D_1 \rightarrow D_2$, define
$r(\sysvar)=$ $\langle S, s_0, A,D_2, \step,$ $\obs',$ $\dom'\rangle$ 
to be the system that is identical to $\sysvar$, except that it has domains $D_2$, and 
the functions $\dom'$ and $\obs'$ are defined 
by $\dom' = r \compose \dom$, and, for $u\in D_2$, 
$\obs'_{u}(s) = \obs_{G}(s)$, where $G = r^{-1}(u)$. 
Intuitively, each domain $u$ in $r(\sysvar)$ corresponds to the group of 
domains $r^{-1}(u)$ in $M$, with every action of a domain in $r^{-1}(u)$
treated as an action of $u$. 
Similarly, for a formula $\phi$ of the epistemic logic, we write $r^{-1}(\phi)$
for the formula obtained by replacing each occurrence of a group $G$ in a modal 
operator in $\phi$ by the group $r^{-1}(G)$. 

The existence of a surjective mapping $r: D_1 \rightarrow
D_2$ is one requirement for architectural refinement, which we discuss
in greater detail in later sections. 
The following result shows that abstracting a system $M$ to $r(M)$ 
(or, conversely, refining $r(M)$ to $M$) preserves satisfaction of formulas, 
subject to a corresponding abstraction on groups being applied in the formulas.

\begin{theorem} \label{thm:Kpullback}
Let $r: D_1 \rightarrow D_2$ be surjective and let $M$ be a system with domains $D_1$. 
Then for all interpretations $\pi$, sequences of actions $\alpha$ of $M$, and formulas $\phi$
(not including distributed knowledge operators $D_G$) 
 for agents $D_2$ 
we have $r(M), \pi,\alpha\models  \phi$ iff $M, \pi,\alpha\models  r^{-1}(\phi)$. 
\end{theorem}  

  \begin{proof} 
   We  first  claim that for all $\alpha,\beta\in A^*$ and groups $G\subseteq D_2$,
   we have   $\view{r^{-1}(G)}^M(\alpha) = \view{r^{-1}(G)}^M(\beta)$ iff 
  $\view{G}^{r(M)}(\alpha) = \view{G}^{r(M)}(\beta)$, where the 
  superscripts indicate the system within which views are computed. 
Note that, by definition, $$\obs_G^{r(M)}(s) =  \langle \obs^{r(M)}_u(s)\rangle_{u\in G} = 
\langle \langle \obs^M_v(s)\rangle_{v\in {r^{-1}(u)}} \rangle_{u\in G}$$ 
and 
$$   \obs_{r^{-1}(G)}^{M}(s) = \langle \obs^M_v(s)\rangle_{v\in r^{-1}(G)} = \langle \obs^M_v(s)\rangle_{v\in r^{-1}(u), ~ u\in G}
\mathpunct.$$
Since these expressions simply group the same collection of values indexed by $v$ in two different ways, 
 there exists functions $f$ and $f^{-1}$ such that 
for all states $s$, we have $f(\obs_{r^{-1}(G)}^{M}(s) )= \obs_G^{r(M)}(s) $ and 
$f^{-1}(\obs_G^{r(M)}(s) ) = \obs_{r^{-1}(G)}^{M}(s)$. Moreover, 
we have for states $s$ and $t$ that 
$\obs_G^{r(M)}(s) = \obs_G^{r(M)}(t)$ iff 
$\obs_{r^{-1}(G)}^{M}(s) = \obs_{r^{-1}(G)}^{M}(t)$. 

Given  a function $h$ on (group) observations, we generalize it to a function $h_+$ on view-like 
sequences of actions and observations, inductively  by 
$h_+(o) = h(o)$, for $o$ an observation,  $h_+(\sigma a o) =  h_+(\sigma) a h(o)$ 
for $\sigma $ a sequence, $a$ an action and $o$ an observation, and 
$h_+(\sigma o_1 o_2) =  h_+(\sigma o_1) \aconcat h(o_2)$ for $\sigma $ a sequence
and $o_1,o_2$ observations. 
   
We now prove the claim by showing that  for all sequences of actions $\alpha$, we have 
$f_+(\view{{r^{-1}(G)}}^M(\alpha)) = \view{G}^{r(M)}(\alpha))$ and $f^{-1}_+(\view{G}^{r(M)}(\alpha)) = \view{{r^{-1}(G)}}^M(\alpha)$. 
The proof is by induction on  $\alpha$. 
We consider just the case of $f_+$.      
   In case $\alpha = \epsilon$, we have 
  $f_+(\view{{r^{-1}(G)}}^M(\alpha)) = f_+(\obs^M_{r^{-1}(G)}(s_0)) = f(\obs^M_{r^{-1}(G)}(s_0)) =
  \obs_G^{r(M)}(s_0) =  \view{G}^{r(M)}(\alpha)$.
  For sequences $\alpha a$, there are two cases. 
  If $\dom^M(a) \in r^{-1}(G)$ then $\dom^{r(M)}(a) = r(\dom^M(a)) \in G$. 
  Thus, in this case, using the inductive hypothesis, 
  \begin{align*} 
  f_+(\view{{r^{-1}(G)}}^M(\alpha a)) 
   & = f_+(\view{{r^{-1}(G)}}^M(\alpha) \, a \, \obs^M_{r^{-1}(G)}(s_0\cdot \alpha a)) \\ 
 &   = f_+(\view{{r^{-1}(G)}}^M(\alpha)) \, a \, f(\obs^M_{r^{-1}(G)}(s_0\cdot \alpha a))\\ 
&   =  \view{G}^{r(M)}(\alpha) \, a \, \obs^{r(M)}_G(s_0\cdot \alpha a) \\ 
& =  \view{G}^{r(M)}(\alpha a)\mathpunct.
\end{align*}
Alternately, if 
$\dom^M(a) \not \in r^{-1}(G)$ then $\dom^{r(M)}(a) = r(\dom^M(a)) \not \in G$. 
In this case, there are two further possibilities. 
If $\obs^M_{r^{-1}(G)}(s_0\cdot \alpha) = \obs^M_{r^{-1}(G)}(s_0\cdot \alpha a)$, 
then also $\obs^{r(M)}_{G}(s_0\cdot \alpha) = \obs^{r(M)}_{G}(s_0\cdot \alpha a)$. Thus  
 \begin{align*} 
  f_+(\view{{r^{-1}(G)}}^M(\alpha a)) 
   & = f_+(\view{{r^{-1}(G)}}^M(\alpha) \circ \obs^M_{r^{-1}(G)}(s_0\cdot \alpha a)) \\ 
 &   = f_+(\view{{r^{-1}(G)}}^M(\alpha))\\ 
 &   =  \view{G}^{r(M)}(\alpha) \\ 
&   =  \view{G}^{r(M)}(\alpha) \circ \obs^{r(M)}_G(s_0\cdot \alpha a) \\ 
& =  \view{G}^{r(M)}(\alpha a)\mathpunct.
\end{align*}
The other possibility is that $\obs^M_{r^{-1}(G)}(s_0\cdot \alpha) \neq \obs^M_{r^{-1}(G)}(s_0\cdot \alpha a)$, 
where we have
 \begin{align*} 
  f_+(\view{{r^{-1}(G)}}^M(\alpha a)) 
   & = f_+(\view{{r^{-1}(G)}}^M(\alpha) \, \obs^M_{r^{-1}(G)}(s_0\cdot \alpha a)) \\ 
 &   = f_+(\view{{r^{-1}(G)}}^M(\alpha)) \circ  f(\obs^M_{r^{-1}(G)}(s_0\cdot \alpha a))\\ 
&   =  \view{G}^{r(M)}(\alpha) \circ \obs^{r(M)}_G(s_0\cdot \alpha a) \\ 
& =  \view{G}^{r(M)}(\alpha a)\mathpunct.
\end{align*}
Thus, in any case we have $ f_+(\view{{r^{-1}(G)}}^M(\alpha a)) = \view{G}^{r(M)}(\alpha a)$, 
completing the induction.  The argument for $f_+^{-1}$ is symmetric.

The result now follows by induction on the construction of $\phi$. The cases of 
atomic propositions and boolean operators are trivial. For formulas of the form $\knows{G}\psi$, 
we have $r(M),\pi, \alpha \models \knows{G}(\psi)$ iff $r(M), \pi, \beta\models \psi$ for all sequences of actions $\beta$ with 
$\view{G}^{r(M)}(\alpha)= \view{G}^{r(M)}(\beta)$. 
 By the induction hypothesis and the claim proved above, 
this is equivalent to $M, \pi, \beta\models r^{-1}(\psi)$ for all sequences of actions $\beta$ with 
$\view{r^{-1}(G)}^M(\alpha) = \view{r^{-1}(G)}^M(\beta)$. This is equivalent to $M, \pi, \alpha\models \knows{r^{-1}(G)} r^{-1}(\psi)$, i.e., 
$M, \pi, \alpha\models r^{-1}( \knows{G} \psi)$. 
\end{proof} 

We note that the example given above to show the difference between group and distributed knowledge also shows 
that Theorem~\ref{thm:Kpullback} would not hold if we were to include distributed knowledge in the language
and  analogously define $r^{-1}(\knows[D]{G} \phi) = \knows[D]{r^{-1}(G)}\phi$. 
For example, consider the function $r$ with $r(u) = r(v) = w$. Since $r^{-1}(\knows{w}p) = \knows{\{u,v\}}p$
and $M,\pi,\alpha \models  \knows{\{u,v\}}p$, we have $r(M), \pi, \alpha \models \knows{w} p$, hence 
$r(M), \pi, \alpha \models \knows[D]{w} p$. However, as shown above, we do not have 
$M,\pi, \alpha \models \knows[D]{\{u,v\}} p$, i.e., we do not have $M,\pi, \alpha \models r^{-1}(\knows[D]{w} p)$. 

\paragraph{$G$-dependent propositions}

\commentout{ 
A proposition 
in a system $M$
is $G$-action-local if the actions of domains in the group $G$ suffice
to decide the proposition.
$G$-action-local propositions are useful to specify 
confidential information in some of our examples. 

 Formally, for a group $G$, a set
$X\subseteq A^*$ is a \emph{$G$-action-local proposition} if for all
$\alpha,\alpha'\in A^*$ if $\alpha \restrict G = \alpha' \restrict G$,
then $\alpha\in X \iff \alpha' \in X$. We write ``$u$-action-local'' when $G=\{u\}$.
Note that if $G\subseteq G'$ and $X$ is $G$-action local then $X$ is $G'$-action local. 
}%

A proposition in a system $M$ depends on the actions of a group $G$ if its truth value 
can be affected by making changes only to the actions of domains in the group $G$. 
The notion of dependence of a proposition on a group $G$ is 
useful to specify confidential information in some of our examples. 

Formally,  for a group $G$ of domains and $\alpha \in A^*$, say that a proposition $X\subseteq A^*$  {\em depends on $G$ actions at $\alpha$} 
if there exists $\beta\in A^*$ such that $\alpha\restrict \overline{G} = \beta\restrict\overline{G}$ 
but $\alpha \in X$ iff $\beta \not \in X$. 
(Notation $\overline{G}$ is shorthand for the set $D \setminus G$, and denotes the set of all domains excluding those in $G$.)
Intuitively, this says that which $G$ actions have occurred, 
and their placement with respect to the actions of other domains, can affect whether or not
the proposition holds. We say that $X$ {\em depends everywhere on $G$ actions}  if 
$X$ depends on $G$ actions at $\alpha$ for all $\alpha \in A^*$.%
\footnote{In earlier versions of this work, 
we used nontrivial  $G$-action local propositions.  
A proposition $X$ is {\em $G$-action local} if  for all
$\alpha,\beta\in A^*$,  if $\alpha \restrict G = \beta \restrict G$,
then $\alpha\in X \iff \beta \in X$. 
It can easily be seen that a non-trivial $G$-action local proposition is everywhere 
dependent on $G$ actions. Consequently, a formulation of our results using 
propositions that depend on $G$ actions is more general.} 

We can also reason about how 
architectural abstraction affects 
$G$-dependent
propositions.

\begin{lemma}\label{lem:prop-refine}
Let $r:D_1\rightarrow D_2$ be surjective and let $\sysvar$ be a  system with 
domains $D_1$. 
Then  
proposition $X$ depends on $G$ actions at $\alpha$
in $r(\sysvar)$ iff $X$ 
depends on $r^{-1}(G)$ actions at $\alpha$
 in $\sysvar$.
\end{lemma}

\begin{proof}
Note that the restriction operation  is relative to a system $M$; to  emphasize 
this we write $\alpha \restrict^M G$ to indicate that we use the domain function from $M$. 
Note also that $$r^{-1}(\overline{G}) = \overline{r^{-1}(G)}~.$$
By definition,   $X$ depends on $G$ actions at $\alpha$
in $r(\sysvar)$ if there exists $\beta\in A^*$ such that 
$\alpha\restrict^{r(\sysvar)} \overline{G} = \beta\restrict^{r(\sysvar)} \overline{G}$ and 
$\alpha \in X$ iff $ \beta \not \in X$. 
Similarly, $X$ depends on $r^{-1}(G)$ actions at $\alpha$
in $\sysvar$ if there exists $\beta\in A^*$ such that 
$\alpha\restrict^{\sysvar} \overline{r^{-1}(G)} = \beta\restrict^{r(\sysvar)} \overline{r^{-1}(G)}$ and 
$\alpha \in X$ iff $ \beta \not \in X$.

Since $\dom^\sysvar(a) \in \overline{r^{-1}(G)}$ iff
$\dom^\sysvar(a) \in r^{-1}(\overline{G})$ iff
 $\dom^{r(\sysvar)}(a) \in \overline{G}$,
 we have 
 $$\alpha\restrict^{r(\sysvar)} \overline{G} = \beta\restrict^{r(\sysvar)} \overline{G}
\text{  iff } 
 \alpha\restrict^{\sysvar} \overline{r^{-1}(G)} = \beta\restrict^{\sysvar} \overline{r^{-1}(G)}~.$$
 It follows that $X$ depends on $G$ actions at $\alpha$
in $r(\sysvar)$ iff $X$ 
depends on $r^{-1}(G)$ actions at $\alpha$
 in $\sysvar$.
\end{proof}

\subsection{Example: {\hlarch} information security}

The logic allows us to state information security properties about
machines, in terms of the knowledge of domains. The architecture can
provide sufficient structure to prove that a given information
security property holds in all machines that 
comply with 
the architecture.

For example, using the $\hlarch$ architecture, we are able to show
that in any execution of any machine that 
complies with 
$\hlarch$, the
domain $L$ does not 
know 
any 
proposition that depends on $H$ actions. 
\iftr\else
\footnote{%
Proofs of all results are in the full version of this paper, available at \url{http://people.seas.harvard.edu/~chong/arch-filter-full.pdf}.
}
\fi

\begin{theorem}\label{thm:hl-h-action-local}
  If $\sysvar$ is 
  TA-compliant with 
   {\hlarch} and $\pi(p)$ 
   depends on $H$ actions at $\alpha$  then $\interpsysv, \alpha \sat \neg \knows {L} p$.
\end{theorem}
\begin{proof}
  We first show, for any sequence $\gamma\in A^*$,  that $\ta{L}(\gamma) = \ta{L}(\gamma \restrict \{L\})$. 
  The base case, $\gamma=\epsilon$
  is trivial. Consider $\gamma a$. If $\dom(a) \ne L$ then 
  \begin{align*}
    \ta{L}(\gamma a) &= \ta{L}(\gamma) & \text{by defn $\ta{L}$} \\
  &= \ta{L}(\gamma \restrict \{L\}) & \text{by IH} \\
  &= \ta{L}(\gamma a \restrict \{L\}) & \text{since $\dom(a) \ne L$}.
  \end{align*}

  If $\dom(a) = L$ then
  \begin{align*}
    \ta{L}(\gamma a) &= (\ta{L}(\gamma), \ta{L}(\gamma), a) & \text{by defn $\ta{L}$} \\
  &= (\ta{L}(\gamma \restrict \{L\}), \ta{L}(\gamma \restrict \{L\}), a) & \text{by IH} \\
  &= \ta{L}(\gamma a \restrict \{L\}) & \text{since $\dom(a) = L$}.
  \end{align*}

Now, since $\pi(p)$ depends on $H$ actions at $\alpha$ there exists $\beta\in A^*$ such that 
$\alpha \restrict \{ L\}  = \beta \restrict  \{L\}$ and $\alpha \in \pi(p)$ iff $\beta \not \in \pi(p)$. 
By what was shown above,  $\ta{L}(\alpha) = \ta{L}(\alpha \restrict \{L\}) =\ta{L}(\beta \restrict \{L\}) =  \ta{L}(\beta)$,  so 
  since $\sysvar$ is TA-compliant, by
  Lemma~\ref{lem:ta-view}, we have $\alpha \approx_L \beta$ and the
  result follows immediately.
\end{proof}

\subsection{Example: Hinke-Schaefer }

In the Hinke-Schaefer database architecture $\hsarch$, none of the
domains $\Luser$, $\Ldbms$, or $L_{F}$ 
know 
anything about the
domains $\Huser$, $\Hdbms$ or $H_{F}$. This is true even if we
consider the 
group knowledge of $\Luser$, $\Ldbms$, and
$L_{F}$.

\begin{theorem} \label{thm:hinke-schaefer} Let 
system $\sysvar$ be
  TA-compliant with 
  $\hsarch$, and let $G = \{\Luser,
  \Ldbms, L_{F}\}$. 
  If $\pi(p)$ 
depends on  $\{\Huser,$ $\Hdbms,$ $H_{F}\}$ actions at $\alpha$,  then $ \interpsysvf{\sysvar}{\pi}, \alpha \sat \neg \knows{G} p. $
\end{theorem}
\iftr

Although we could prove Theorem~\ref{thm:hinke-schaefer}
directly, we defer the proof to
Section~\ref{sec:hink-schaefer-refinement}, where the result follows
easily from Theorem~\ref{thm:hl-h-action-local} and the relationship
between architectures $\hlarch$ and $\hsarch$.
(Specifically, architecture $\hsarch$ is a refinement of architecture $\hlarch$.)

\fi %

Since 
$  \knows[K]{u}p\rimp\knows{G}p$ is valid for $u\in G$, 
it follows that also $\interpsysvf{\sysvar}{\pi} \sat \neg \knows[K]{u} p$
for $u \in \{\Luser, \Ldbms, L_{F}\}$. In particular,  
$\Luser$ does not have any information about the High side of the system.

\section{Extended architecture}\label{sec:extended}

The architectures used so far impose coarse, global constraints on the
causal structure of systems. If $\flowstoin{u}{v}$ then TA-compliance
permits domain $u$ to send to domain $v$ any and all data it
has. However, in many systems, key security properties depend on
the fact that trusted components allow 
only certain information to flow from one domain to another.
Finer specification of information flows in the architecture allow us
to prove stronger information security properties.

In this section we extend the notion of architecture by introducing
\emph{filter functions} to allow fine-grained specification
of what information flows between domains.  We define semantics for
these extended architectures, and present examples where
the extended architectures allow us to prove strong information
security properties.

\subsection{Filter functions}

\newcommand{\labels}{{\cal L}}
\newcommand{\archint}{{\cal I}}
\newcommand{\archspec}{{\cal C}}

An extended architecture is a pair $\archvar=(D, \flowsto)$, where
$D$ is a set of security domains, and $\flowsto \subseteq D \times D \times (\labels\cup \{\top\})$, where $\labels $ is 
a set of
function names. 
We write  ${\flowstoin[f]{u}{v}}$ when $(u,v,f) \in \flowsto$, 
write $\flowstoin{u}{v}$ as shorthand for
$\exists f.~(u,v,f) \in \flowsto$,  and 
$\notflowstoin{u}{v}$ as shorthand for
$\neg \exists f.~(u,v,f) \in \flowsto$. 

Intuitively, ${\flowstoin[f]{u}{v}}$ represents that information flow from 
$u$ to $v$ is permitted, but may be subject to 
constraints. 
In case $f= \top$, there are no 
constraints on information flow from $u$ to $v$: 
any information that may be possessed by $u$ is permitted to be passed
to $v$ when  $u$ acts, just as in the definition of 
TA-compliance. 
If $f\in \labels$ then information is allowed to
flow from domain $u$ to domain $v$, but it needs to be {\em filtered} through the 
function denoted by $f$: 
only information output by this function may be 
transmitted from $u$ to $v$. If $\notflowstoin{u}{v}$ then no 
direct 
flow of information from $u$ to $v$ is permitted. 

In some cases, it may be possible for the operating system or network
infrastructure to enforce a given filter function. However, in general, a
filter function is a local constraint on a trusted component of the
system. That is, if ${\flowstoin[f]{u}{v}}$ for $f \ne \top$, then
component $u$ is trusted to enforce that information sent to $v$ is
filtered appropriately.

We require that extended architectures have the following
properties:
 \begin{enumerate}  
 \item For all $u,v\in D$, there exists at most one 
 $f\in \labels \cup \{\top\}$
 such that
   \flowstoin[f]{u}{v}.
 \item 
   The relation $\flowsto$ is reflexive in that for all $u \in D$
we have  $(u,u,\top)\in \flowsto$. 
\end{enumerate} 
The first condition requires that all permitted flows of information
from $u$ to $v$ are represented using a single 
labeled edge. 
Intuitively, any policy with multiple such 
edges
can always be transformed into one satisfying this condition, 
by combining the pieces of information flowing across these edges  
into a tuple that flows across a single edge. 
The second condition is motivated from the fact, already noted above, 
that information flow from a domain to itself cannot be prevented. 

\iftr
For example, the following diagram shows an extended architecture with domains $H$, $D$, and $L$, intended to represent a high-security domain, a trusted declassifier, and a low-security domain. The arrows indicate permitted flow between domains. The label on the edge from $D$ to $L$ indicates that information going from $D$ to $L$ should be filtered by function $\mathit{rel}$. 
When drawing extended architectures, we annotate arrows
between domains with the filter function names. For arrows drawn without
a label,  
and elided reflexive arrows, the implied 
label is $\top$. 

\begin{center}
  \includegraphics{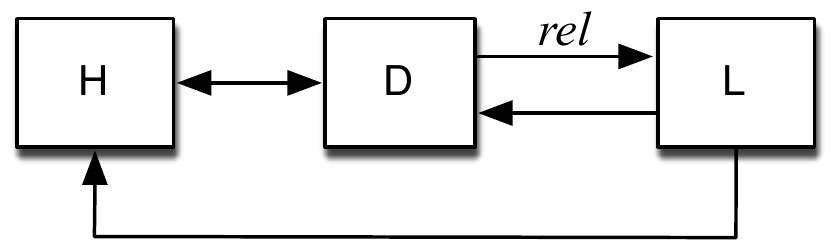}
\end{center}

\fi 

Extended architectures do not define the interpretations of 
the 
function names 
$\labels$. If $\archvar=(D, \flowsto)$ is an extended 
architecture, an {\em interpretation}  for $\archvar$ is a tuple
$\archint = (A, \dom, \intn)$, where $A$ is a set of actions, $\dom:A\rightarrow D$ assigns these
actions to domains of $\archvar$, and $\intn$ is a function mapping each $f\in \labels$
to a function with domain $A^* \times A$ (and arbitrary codomain). 
We call the pair $(\archvar, \archint)$ an {\em interpreted extended
  architecture}, 
or simply an {\em interpreted architecture}.

Intuitively,  if ${\flowstoin[f]{u}{v}}$ and  $\alpha\in A^*$ and $a\in A$ is an action with $\dom(a) = u$, then 
$\intn(f)(\alpha, a)$ is the information that is permitted  to flow from $u$ to $v$ 
when the action $a$ is performed after occurrence of the sequence 
of actions 
$\alpha \in A^*$. 

Given extended architecture $(D, \flowsto)$ and 
an architectural interpretation $\archint = (A,\dom,\intn)$, 
we define a function 
$\tff{u}$ with domain $A^*$ that,  like $\ta{u}$, captures 
the
maximal information 
that domain $u$ is permitted to have after a sequence of actions has been executed. 
The definition is recursive with a function $\inff{v,u}$ for $u,v\in D$, 
mapping a sequence $\alpha\in A^*$ and an action $a\in A$ with $\dom(a) = u$ 
to 
\begin{align*}  
\inff{v,u}(\alpha, a) &= 
\begin{cases}
  \epsilon & \text{if  $\notflowstoin{v}{u}$} \\
  (\tff{v}(\alpha), a)  & \text{if  $\flowstoin[\top]{v}{u}$} \\
  \intn(f)(\alpha, a)  & \text{if $\flowstoin[f]{v}{u}$.} 
\end{cases}
\end{align*}
Intuitively, $\inff{v,u}(\alpha, a)$ represents the new information permitted to be known by $u$ 
when action $a$ is performed after sequence $\alpha$. 

The function $\tff{u}$ is defined 
by
$\tff{u}(\epsilon) = \epsilon$, and, for $\alpha \in A^*$ and $a \in
A$,
$\tff{u}(\alpha a)  = \tff{u}(\alpha) \concat \inff{\dom(a),u}(\alpha,a)$ 
where $\concat$ is the operation of appending an element to the end of a sequence.
Some important technical points concerning the append operation are that 
for any sequence $\sigma$, we define $\sigma \concat \epsilon = \sigma$
(i.e., appending the empty sequence $\epsilon$ has no effect), 
and if $\delta$ happens to be a nonempty sequence, then $\sigma\concat \delta$ 
is the sequence that extends the sequence $\sigma$ by the single additional element 
$\delta$. For example if $\delta$ is the sequence $ab$, then $\sigma\concat \delta$ has 
final element equal to the sequence $ab$ rather than $b$.

Unfolding the definition, we obtain 
\begin{align*}  
\tff{u}(\alpha a) &= 
\begin{cases}
  \tff{u}(\alpha) & \text{if  $\notflowstoin{\dom(a)}{u}$} \\
  \tff{u}(\alpha) \concat (\tff{\dom(a)}(\alpha), a)  & \text{if  $\flowstoin[\top]{\dom(a)}{u}$} \\
  \tff{u}(\alpha) \concat \intn(f)(\alpha, a)  & \text{if $\flowstoin[f]{\dom(a)}{u}$.} 
\end{cases}
\end{align*}
The first two  clauses resemble the definition of $\ta{u}$; 
the third adds to this that the information flowing along
an edge labeled by a function name $f$ is filtered by the 
interpretation $\intn(f)$. Note that if $\intn(f)(\alpha, a) = \epsilon$, where
$\flowstoin[f]{\dom(a)}{u}$, then $\tff{u}(\alpha a) =
\tff{u}(\alpha)$. That is, filter function 
$\intn(f)$ 
can specify that no
information should flow under certain conditions.
Note also that $\tff{u}$ 
and $\inff{v,u}$ have 
implicit parameters, viz., 
an information flow policy $\flowsto$
and an  architectural interpretation  $\archint = (A,\dom,\intn)$. 
When we need to make some of these parameters explicit, we write 
expressions such as 
$\tff{u}^{(\flowsto,\archint)}$, or $\tff{u}^{\flowsto}$.

The function $\tff{u}$ is used analogously
to $\ta{u}$ to define the maximal information that a domain is
permitted to observe for a given sequence of actions. However,
$\tff{u}$ is a more precise bound than $\ta{u}$, as it uses filter
functions to bound the information sent between domains. 

It is reasonable to assume that  information sent from $u$ to $v$ is information that
$u$ is permitted to have. We say that a function is
$\tff{u}$-compatible when the information it conveys is determined by
information that $u$ is permitted to have.
\begin{definition}
  Function $h$ with domain $A^* \times A$ 
is  $\tff{u}$-compatible when  for all sequences $\alpha,\beta\in A^*$, 
$\tff{u}(\alpha)=\tff{u}(\beta)$ implies 
that for all $a\in A$ with $\dom(a) =u$ we have $h(\alpha, a) = h(\beta, a)$.
\end{definition}
We say that the interpretation $\archint= (A,\dom,\intn)$ is 
{\em compatible} with 
\mbox{$\archvar = (D, \flowsto)$} if 
for all $u\in D$ and edges $\flowstoin[f]{u}{v}$ with $f\in \labels$, 
the function $\intn(f)$ is $\tff{u}$-compatible.  In what follows, we 
require that interpretations be compatible with their architectures.

A machine complies with an interpreted extended architecture 
if it has appropriate domains and actions, and 
for each domain $u$, what $u$ observes in state
$\run{s_0}{\alpha}$ is determined by $\tff{u}(\alpha)$. We call such a
machine 
\emph{\TFF-compliant}. 
(``\TFF'' is derived from \emph{filtered transmission} of information about \emph{actions}.)

\begin{definition}[\TFF-compliant]
  A machine $\sysvar$ $=$ $\langle S,$ $s_0,$ $A,$ $D,$ $\step,$ $\obs, \dom\rangle $ is \TFF-compliant 
  with  an interpreted architecture
  $(\archvar, \archint)$, 
  with  $\archvar = (D',\flowsto)$ and $\archint =$ $(A',$ $\dom',$ $\intn)$, 
   if $A = A'$, $D= D'$, $\dom = \dom'$ and   
   for all agents $u\in D$ and all $\alpha, \alpha' \in
  A^*$, if $\tff{u}(\alpha) = \tff{u}(\alpha')$ then
  $\obs_u(\alpha) = \obs_u(\alpha')$.
\end{definition}

For an  interpreted architecture $\intarch$, 
with actions $A$ and domains $D$, if $\pi$ is an interpretation with $\pi(p) \subseteq A^*$, 
we write $\intarch, \pi, \alpha \models \phi$ if 
$\sysvar, \pi, \alpha \models \phi$ for all systems $\sysvar$ that are \TFF-compliant with 
$\intarch$. Similarly, we write $\intarch, \pi \models \phi$ if $\intarch, \pi, \alpha \models \phi$
for all $\alpha \in A^*$. 

Separating extended architectures from their interpretations
ensures that extended architectures can be completely 
represented by graphical diagrams with labeled edges. It also 
allows us to deal with examples where an extended architecture 
can be implemented in a variety of ways, and weak constraints on the 
set of actions and the set of filter functions  suffice to enforce 
the security properties of interest. We will present a number of 
examples of this in what follows. To capture the constraints on the 
architectural interpretations at the semantic level, we 
use the notion of an {\em architectural specification}, which 
is a pair $(\archvar, \archspec)$ where $\archvar$ is an extended architecture
and $\archspec$ is a set of architectural interpretations for $\archvar$. 
(We will not attempt in this paper to develop any syntactic 
notation for architectural specifications.) 

\begin{definition}
  A machine $\sysvar$ is \TFF-compliant with an 
 architectural specification $(\archvar, \archspec)$ if there 
  exists an interpretation $\archint\in \archspec$ such that 
  $\sysvar$ is \TFF-compliant with the interpreted architecture
  $(\archvar,\archint)$. 
\end{definition}

The following theorem shows that \TFF-compliance generalizes TA-compliance. 
Thus, we are free to interpret a given
architecture as an extended architecture.

\begin{theorem}\label{thm:archtoextarch}
  Let $\archvar_1=(D, \flowsto_1)$ be an architecture, and let
  $\archvar_2=(D,\flowsto_2)$ be the extended architecture such that
  $(u,v,f) \in \flowsto_2$ if and only if $f = \top$ and $(u,v)
  \in \flowsto_1$.  
  Let $\sysvar$ be a machine with domains $D$, actions $A$ and domain function $\dom$. 
  Let $\archint = (A,\dom, \intn)$ be any interpretation for $\archvar_2$
  with this set of actions and domain  function. 
  Then 
   $\sysvar$ is TA-compliant with 
  $\archvar_1$ if and only if $\sysvar$ is \TFF-compliant with 
  $(\archvar_2,\archint)$.
\end{theorem}
\begin{proof} 
  We show that for all $u \in D$, $\ta{u}$ is isomorphic to
  $\tff{u}$, 
  in the sense that $\ta{u}(\alpha) = \ta{u}(\beta)$ iff 
  $\tff{u}(\alpha) = \tff{u}(\beta)$. 
  Let 
  $g$ be the function mapping nested triples to sequences,  defined
  inductively by 
  $g(\epsilon) = \epsilon$ and 
$g((t_1, t_2, a)) = g(t_1)\concat(g(t_2),a)$. 
  Conversely, let $h$ be the function from the range of $g$ to nested tuples, 
  defined inductively by $h(\epsilon) = \epsilon$ and 
  $h(\sigma\concat (x,a)) = (h(\sigma),h(x), a))$. 
  A straightforward induction shows that $h\circ g$ is the identity function, so $g$ is 1-1. 
  We show by induction
  on $\alpha$ that for all $u\in D$, $g(\ta{u}(\alpha)) = \tff{u}(\alpha)$.

  The base case is trivial. Assume that for all $u \in D$ we have
  $g(\ta{u}(\alpha)) = \tff{u}(\alpha)$ and consider $\alpha a$. Let
  $u \in D$. If $\dom(a) \not\flowsto_1 u$ then $\dom(a)
  \not\flowsto_2 u$ and
  \begin{align*}
    g(\ta{u}(\alpha a)) &= g(\ta{u}(\alpha)) &\text{by defn $\ta{u}$} \\
    &= \tff{u}(\alpha) &\text{by IH} \\
    &= \tff{u}(\alpha a) &\text{by defn $\tff{u}$}
  \end{align*}

  If $\dom(a) \flowsto_1 u$ then $(\dom(a), u, \top) \in \flowsto_2$, and 
  \begin{align*}
    g(\ta{u}(\alpha a)) &= g((\ta{u}(\alpha), \ta{\dom(a)}(\alpha), a)) &\text{by defn $\ta{u}$} \\
    &= g(\ta{u}(\alpha))\concat(g(\ta{\dom(a)}(\alpha)),a) &\text{by defn $g$} \\
    &= \tff{u}(\alpha) \concat(\tff{\dom(a)}(\alpha), a) &\text{by IH} \\
    &= \tff{u}(\alpha a) &\text{by defn $\tff{u}$}
  \end{align*}
  
  It now follows that $\ta{u}(\alpha) = \ta{u}(\beta)$ iff $\tff{u}(\alpha) = \tff{u}(\beta)$. 
  For, if $\ta{u}(\alpha) = \ta{u}(\beta)$ then $\tff{u}(\alpha) = g(\ta{u}(\alpha)) =
  g(\ta{u}(\beta)) = \tff{u}(\beta)$. Conversely, if $\tff{u}(\alpha) = \tff{u}(\beta)$
  then $g(\ta{u}(\alpha)) =  g(\ta{u}(\beta))$ and we get $\ta{u}(\alpha) = \ta{u}(\beta)$ by injectivity of 
  $g$. 
\end{proof}

\TFF-compliance requires that if $\tff{u}(\alpha) = \tff{u}(\alpha')$ then
the observations of $u$ in state $\run{s_0}{\alpha}$ and in state
$\run{s_0}{\alpha'}$ are equal. The following 
lemma 
\iftr
(similar to Lemma~\ref{lem:ta-view}) 
\fi
shows that in fact
\TFF-compliance 
implies that if $\tff{u}(\alpha) = \tff{u}(\alpha')$ then
we have that
$\view{u}(\alpha) = \view{u}(\alpha')$.

We require a technical assumption for this result: 
say that an interpreted architecture  $(\archvar, \archint)$
is {\em non-conflating} if 
for all $u,v\in D$, if  $\flowstoin[f]{u}{v}$ with $f\neq \top$ 
then for all actions $a,b\in A$ with $\dom(a) = u$ and $\dom(b) = v$, and for all $\alpha, \beta \in A^*$
we have  $\intn(f)(\alpha, a) \neq (\tff{v}(\beta), b)$. 
Recall that by reflexivity,  $\flowstoin[f]{u}{v}$ with $f\neq \top$ implies $u\neq v$. 
Intuitively, the condition states that, in the context of the definition of $\tff{v}$, 
it is always possible for domain $v$ to distinguish 
the type of information $\intn(f)(\alpha, a)$ transmitted to it  from the 
type of information $(\tff{v}(\beta), b)$ transmitted by $v$ to itself. 
That is, $v$ can distinguish between the effects of its own actions on $\tff{v}$ 
and the effects of other domains' actions. Since, intuitively,  $v$ should be aware of its 
own actions, it is reasonable to expect that this is generally satisfied. 

\begin{lemma}\label{lem:tff-view}
  If $\sysvar$ is \TFF-compliant with 
  non-conflating
  interpreted architecture
 $(\archvar, \archint)$, 
  with  $\archvar = (D,\flowsto)$ and $\archint = (A,\dom,\intn)$, 
  then for all agents $u\in D$ and all $\alpha,
  \alpha' \in A^*$ such that $\tff{u}(\alpha) = \tff{u}(\alpha')$ we
  have $\view{u}(\alpha) = \view{u}(\alpha')$
\end{lemma}
\begin{proof}
  By induction on $|\alpha| + |\alpha'|$. The base case is
  trivial. Suppose that the result holds for all sequences of shorter
  combined length, and consider strings $\alpha a$ and $\alpha'$ such that 
  $\tff{u}(\alpha a) = \tff{u}(\alpha')$.
Note that by \TFF-compliance this implies that   $\obs_{u}(\alpha a) = \obs_{u}(\alpha')$.
We consider several cases. 

\begin{itemize}
\item If $\notflowstoin{\dom(a)}{u}$ then $\tff{u}(\alpha) =
  \tff{u}(\alpha a) = \tff{u}(\alpha')$, so $\view{u}(\alpha) =
  \view{u}(\alpha')$, by the induction hypothesis, and
  $\obs_{u}(\alpha a) = \obs_u(\alpha)$, by 
  \TFF-compliance. 
  By
  reflexivity of $\flowsto$, we have $\dom(a) \neq u$, so it follows
  that
  \begin{align*}
    \view{u}(\alpha a) &= \view{u}(\alpha) \aconcat \obs_{u}(\alpha a) \\
 &= \view{u}(\alpha)\\
&= \view{u}(\alpha')
  \end{align*}

\item If $\flowstoin[f]{\dom(a)}{u}$ and $f \ne \top$ and
  $\intn(f)(\alpha, a) = \epsilon$ then 
  $$ \tff{u}(\alpha) = \tff{u}(\alpha)\concat \intn(f)(\alpha, a) =\tff{u}(\alpha a) =  \tff{u}(\alpha') \mathpunct.$$ 
  Thus we have $\view{u}(\alpha) =
  \view{u}(\alpha')$, by the induction hypothesis, and
 $\obs_{u}(\alpha a) = \obs_u(\alpha)$, by   \TFF-compliance. 
  By reflexivity of $\flowsto$, we have $\dom(a) \neq u$. Therefore
  \begin{align*}
    \view{u}(\alpha a) &= \view{u}(\alpha) \aconcat \obs_{u}(\alpha a)\\
    &= \view{u}(\alpha) \\ 
     &= \view{u}(\alpha')
  \end{align*}

\item If $\flowstoin[f]{\dom(a)}{u}$ and $f \ne \top$ and $\intn(f)(\alpha, a) \ne \epsilon$
  then 
  $ \tff{u}(\alpha') = \tff{u}(\alpha a) = \tff{u}(\alpha) \concat \intn(f)(\alpha, a)$ 
  yields that $\alpha' \ne \epsilon$. 
Assume
  $\alpha' = \beta b$. 

\begin{itemize}
\item If either $\notflowstoin{\dom(b)}{u}$, or  $\flowstoin[g]{\dom(b)}{u}$ and $g \ne \top$ and $\intn(g)(\beta, b) = \epsilon$,
then we can swap the roles of $\alpha a$ and $\beta b$ and apply the argument from the previous two cases.

\item If $\flowstoin[g]{\dom(b)}{u}$ and $g \ne \top$ and
  $\intn(g)(\beta, b) \ne \epsilon$ then 
  $$\tff{u}(\alpha a) =
  \tff{u}(\alpha)\concat \intn(f)(\alpha, a) = \tff{u}(\beta)\concat \intn(g)(\beta,
  b) = \tff{u}(\beta b)\mathpunct,$$ and so
   $\tff{u}(\alpha) = \tff{u}(\beta)$, 
   since neither of the appended elements is $\epsilon$. 
     By the inductive
  hypothesis, $\view{u}(\alpha) = \view{u}(\beta)$.
  Since $f \ne \top$ and $g \ne \top$, by the reflexivity of
  $\flowsto$, we have $\dom(a)\ne u$ and $\dom(b) \ne u$. Thus
    \begin{align*}
      \view{u}(\alpha a) &= \view{u}(\alpha) \aconcat \obs_u(\alpha a) \\
      &= \view{u}(\beta) \aconcat \obs_u(\beta b) \\
      &= \view{u}(\beta b).
    \end{align*}

  \item If $\flowstoin[\top]{\dom(b)}{u}$ then 
  $$\tff{u}(\alpha) \concat \intn(f)(\alpha, a) = \tff{u}(\alpha a)= \tff{u}(\beta b)
    = \tff{u}(\beta)\concat 
    (\tff{u}(\beta), b) \mathpunct.$$ 
    Since $ \intn(f)(\alpha, a) \neq \epsilon$, it follows that 
    $\tff{u}(\alpha) = \tff{u}(\beta)$ and
    $ \intn(f)(\alpha, a)  =  (\tff{u}(\beta), b)$. 
    However, 
        by the assumption that the interpretation is non-conflating, 
    the latter
    implies that $\dom(b) \neq u$. Also, $\dom(a) \neq u$ by reflexivity. 
     By the induction hypothesis, we obtain from   $\tff{u}(\alpha) = \tff{u}(\beta)$ that 
     $\view{u}(\alpha)  =\view{u}(\beta)$. 
         Thus
    \begin{align*}
      \view{u}(\alpha a) &= \view{u}(\alpha) \aconcat \obs_u(\alpha a) \\
      &= \view{u}(\beta) \aconcat \obs_u(\beta b) \\
      &= \view{u}(\beta b).
    \end{align*}
\end{itemize}

\item If $\flowstoin[\top]{\dom(a)}{u}$ then 
$ \tff{u}(\alpha') = \tff{u}(\alpha a) =
  \tff{u}(\alpha)\concat (\tff{\dom(a)}(\alpha), a) $. This means
  that $\alpha' \ne \epsilon$.  Assume $\alpha' = \beta b$.

  \begin{itemize}
  \item 
  If either $\notflowstoin{\dom(b)}{u}$, or $\flowstoin[g]{\dom(b)}{u}$ and $g \ne \top$ 
  then we can swap the roles of $\alpha a$ and $\beta b$ and apply the argument 
    from the first 
    three 
    cases above.

  \item If $\flowstoin[\top]{\dom(b)}{u}$ then $\tff{u}(\alpha a) =
 \tff{u}(\alpha) \concat (\tff{\dom(a)}(\alpha), a) $
  and $\tff{u}(\beta b) =  \tff{u}(\beta)\concat (\tff{\dom(b)}(\beta), b)$.  Therefore, 
 we have $a=b$,
    and $\tff{u}(\alpha) = \tff{u}(\beta)$.  By the inductive
    hypothesis, we have $\view{u}(\alpha)= \view{u}(\beta)$.  
   If $\dom(a) = u$, then 
    \begin{align*}
      \view{u}(\alpha a) &= \view{u}(\alpha)\, a\, \obs_u(\alpha a) \\
      &= \view{u}(\beta) \,b\, \obs_u(\beta b) \\
      &= \view{u}(\beta b).
    \end{align*}
Alternately, if $\dom(a) \neq u$, then 
  \begin{align*}
      \view{u}(\alpha a ) &= \view{u}(\alpha) \aconcat \obs_{u}(\alpha a)\\
      &= \view{u}(\beta) \aconcat \obs_{u}(\beta b)\\
      &= \view{u}(\beta b).
    \end{align*}

\end{itemize}
\end{itemize}
\end{proof}

We note that this result does not hold for conflating interpreted architectures. 
For example, consider an interpreted architecture in which 
$\flowstoin[f]{\dom(a)}{\dom(b)}$ and $\intn(f)(\alpha, a) = (\epsilon,b)$
and a system in which $\obs_u(s) = \bot$ for all $u\in D$ and states $s$. 
This system is necessarily \TFF-compliant with the architecture, but 
we have $\tff{\dom(b)}(a) = (\epsilon,b) = \tff{\dom(b)}(b)$
but $\view{\dom(b)}(a) = \bot$ and $\view{\dom(b)}(b) = \bot b \bot$. 

However, it is always possible to convert an interpretation $\intn$ 
 into another equivalent non-conflating interpretation $\intn'$, 
by defining $\intn'(f)(\alpha, a) = (\intn(f)(\alpha, a),x)$ for some
 fixed value $x$ that is not in $A$. 
Note that $\intn$ and $\intn'$ are equivalent in the sense that  
 $\intn(f)(\alpha, a) = \intn(g)(\beta, b)$ if and only if  
  $\intn'(f)(\alpha, a) = \intn'(g)(\beta, b)$, so the 
 ``information content''  of 
 the values of 
 $\intn$ and $\intn'$ 
  are the same. We therefore assume in what follows
  that interpreted architectures are non-conflating.

\subsection{Example: Starlight Interactive Link}

\newcommand{\strltArch}{\ensuremath{\mathcal{SL}}}

The Starlight Interactive Link~\citep{starlight-interactive} provides
interactive access from a high-security network to a low-security
network. This allows a user on the high-security network to have
windows open 
 on her screen
at differing security levels, while
ensuring no high-security information goes to the low-security
network.

Starlight has both hardware and software components. The hardware
device is connected to both the high-security and low-security
networks, and has a keyboard and mouse attached. There is a switch
that can toggle between the high-security and low-security networks; input from
the mouse and keyboard are sent to the network currently selected by
the switch.  Starlight allows data from the low-security network to be
transferred to the high-security network, but not vice versa. The
Starlight software components include proxy window clients and servers
to allow the windowing environment to work in the presence of the
Starlight hardware.

The following diagram shows 
extended
architecture {\strltArch}, an
architecture for the Starlight Interactive Link. 
The architecture uses
a filter function to specify what information the Starlight
Interactive Link may send to the low-security network.

 \begin{center}
   \includegraphics[height=2.5\diagsize]{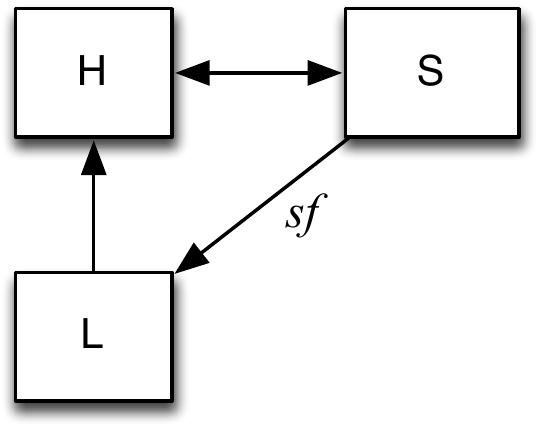}
 \end{center}

Domain $H$ represents the high-security network (including the user's
computer); domain $L$ represents the low-security network; domain $S$
represents the Starlight Interactive Link (including input devices),
which routes keyboard and mouse events to either the high-security or
low-security network.
Note that there is no edge from domain $L$ to domain $S$, as no information is sent directly from the low-security network to the Starlight Interactive Link. Instead, data from the low-security network (such as updates to the contents of a window) are sent to the high-security network, and thence to software components of the Starlight Interactive Link.

\newcommand{\strltFltr}{\mathit{sf}}
\newcommand{\strltToggle}{\mathtt{t}}

The edge labeled  $\strltFltr$
restricts what information is allowed to
flow 
from 
the Starlight Interactive Link to  low-security network
$L$. 
We
 present an architectural specification $\archspec_{\strltArch}$ based on {\strltArch}
that expresses a constraint on 
interpretations of $\strltFltr$. 
An interpretation $(A, \dom, \intn)$ is included in $\archspec_{\strltArch}$ if the 
following conditions hold. 
Let $A_S = \{a\in A~|~\dom(a) =S\}$ 
be the set of actions belonging to domain $S$. We
assume that there is a distinguished action $\strltToggle \in A_S$
that toggles which network is receiving the input events. 
Intuitively, $L$ 
is permitted to 
know about the occurrence of any
$t$ action, and the occurrence of any other action in $A_S$ (e.g.,
keyboard or mouse input) that happens while the low-security network
is selected (i.e., after an odd number of toggle actions).
We capture this by the following assumption on 
interpretation  $\intn$: 
\begin{align*}
 \intn(\strltFltr)(\alpha, a) &=
  \begin{cases}
    a &\text{if $a=\strltToggle$ or} \\
      &~ \text{$\numof_\strltToggle(\alpha)$ is odd and $a \in A_S$} \\
    \epsilon & \text{otherwise}
  \end{cases}
\end{align*}
where $\numof_\strltToggle(\alpha)$ is the number of occurrences of $t$ in $\alpha$. 
For example, with the $\alpha_i$ consisting of only $H$ and $L$ actions, and 
$a\neq\strltToggle$ an $S$ action, we have 
\begin{tabbing} 
xxxxxxx\=\kill \\
\> $\intn(\strltFltr)(\alpha_0, a) = \epsilon$,  \\
\> $\intn(\strltFltr)(\alpha_0\, a ,\strltToggle ) = t$,  \\
\> $\intn(\strltFltr)(\alpha_0\, a\, \strltToggle \,\alpha_1,  a) = a$, and \\
\> $\intn(\strltFltr)(\alpha_0\, a\, \strltToggle\, \alpha_1\,  a\, \strltToggle\, \alpha_2, a) = \epsilon$, etc. \\ 
\end{tabbing} 

It is straightforward to check 
for such an interpretation  \iftr that it is non-conflating and \fi
that $\intn(\strltFltr)$ is $\tff{S}$-compatible. 

\iftr
\subsubsection{Information security properties}
\fi %

The component $S$, corresponding to the Starlight Interactive Link, is
a trusted component, and the $\strltFltr$ filter is a local
constraint on the component. To verify that a system satisfies
specification $(\strltArch, \archspec_{\strltArch})$, we would need to
verify that $S$ appropriately filters information sent to $L$, and
that all other communication in the system complies with the
architecture, to wit, that $H$ cannot communicate directly with $L$, and $L$
cannot communicate directly with $S$.

If a system does satisfy
specification  $(\strltArch, \archspec_{\strltArch})$, 
then we can show that
domain $L$ never 
knows 
any
$H$-dependent propositions.
Indeed, we can show something stronger, that $L$ never 
knows 
any
proposition about $H$ actions and $S$ actions that occur while the
high-network is selected. 

\newcommand{\togH}{\ensuremath{\mathit{togH}}}
 \newcommand{\togL}{\ensuremath{\mathit{togL}}}

To fully capture this intuition requires a little care, since a proposition 
such as ``$H$ did action $a$ between the first and second $\strltToggle $ actions" 
should not be known to $L$, but refers both to something that should be 
hidden from $L$ and to something that $L$ is permitted to observe (the $\strltToggle $ actions). 
We handle this using an approach that is similar to the way we used 
$G$-dependent propositions above.

Let the 
{\em canonical form} of a sequence $\alpha \in A^*$ be the (unique) representation
$\alpha = \alpha_0 \strltToggle  \alpha_1 \strltToggle  \alpha_2 \ldots \strltToggle  \alpha_n$ where each $\alpha_i$ 
contains no $\strltToggle $ actions.  
\commentout{ 
Define the function $\togH: A^* \rightarrow A^*$ 
so that if $\alpha = \alpha_0 \strltToggle  \alpha_1 \strltToggle  \alpha_2 \ldots \strltToggle  \alpha_n$ is the canonical 
form then $$\togH(\alpha) = 
 (\alpha_0\restrict HS)\, \strltToggle \, (\alpha_1\restrict H) \, \strltToggle \, (\alpha_2 \restrict HS)\, \strltToggle  \ldots \strltToggle \, (\alpha_n\restrict G_n)$$
 where $G_n$ is $HS$ if $n$ is even and $G_n = H$ if $n$ is odd. 
 (Here we abbreviate the set  $\{H\}$ to $H$ and $\{H,S\}$ to $HS$.)  
Intuitively, $\togH(\alpha)$ is the subsequence of $\alpha$ consisting of events that 
$L$ is not permitted to know about, plus any $\strltToggle $ events. Formally, 
a proposition $X \subseteq A^*$ is {\em toggle-H dependent}  if 
for all $\alpha, \beta \in A^*$, if $\togH(\alpha) = \togH(\beta)$, then $\alpha \in X$ iff
$\beta \in X$. 

\newcommand{\domT}{T} 
To capture information that can be deduced from just the toggle events 
in a sequence (e.g., whether the number of $\strltToggle $ events is prime), we add to the logic a
new domain $\domT$ that sees only the 
$\strltToggle$ events. 
Formally, this domain 
corresponds to the view function defined by $\view{T}(\alpha) = \numof_\strltToggle(\alpha)$, 
and is associated with the knowledge operator $\knows{T}$ with semantics 
derived in the usual way from this view function. 

We can now formally  capture that the only toggle-$H$ dependent propositions that 
$L$ knows are those that can be derived just from the $t$ events that  $L$ sees, in the 
following result. 

\begin{theorem}  \label{thm:starlight-toggle-action-dependent}  
Suppose system $M$ is \TFF-compliant with  
  architectural specification  $(\strltArch, \archspec_{\strltArch})$.   
  If $\pi(p)$
is a  
toggle-$H$ dependent proposition, then
$ \interpsysvf{\sysvar}{\pi} \sat  \knows{L} p \rimp \knows{T}p $. 
\end{theorem}

\begin{proof} 
We first establish a fact about $\tff{L}$. Let $\togL:A^*\rightarrow A^*$ 
be the function such that if $\alpha = \alpha_0 \strltToggle  \alpha_1 \strltToggle  \alpha_2 \ldots \strltToggle  \alpha_n$
is the canonical form, then 
$$\togL(\alpha) =  (\alpha_0\restrict L)\, \strltToggle \, (\alpha_1\restrict LS) \, \strltToggle \, (\alpha_2 \restrict L)\, \strltToggle  \ldots \strltToggle \, (\alpha_n\restrict H_n)\mathpunct,$$ where $H_n = L$ if $n$ is even and $H_n = LS$ if $n$ is odd. 
Intuitively, this is the subsequence consisting of the events that $L$ is permitted to observe. 
Note that since $\numof_\strltToggle(\alpha) = \numof_\strltToggle(\togL(\alpha))$, 
we have  $\intn(\strltFltr)(\alpha, a) =  \intn(\strltFltr)(\togL(\alpha), a) $ when $\dom(a) = S$. 
}%
Define $\togL:A^*\rightarrow A^*$ 
be the function such that if $\alpha = \alpha_0 \strltToggle  \alpha_1 \strltToggle  \alpha_2 \ldots \strltToggle  \alpha_n$
is the canonical form, then 
$$\togL(\alpha) =  (\alpha_0\restrict L)\, \strltToggle \, (\alpha_1\restrict LS) \, \strltToggle \, (\alpha_2 \restrict L)\, \strltToggle  \ldots \strltToggle \, (\alpha_n\restrict H_n)\mathpunct,$$ where $H_n = L$ if $n$ is even and $H_n = LS$ if $n$ is odd. 
Intuitively, this is the subsequence consisting of the events that $L$ is permitted to observe. 
All other events in $\alpha$ are $H$ events, and $S$ events that occurred while the system was toggled to High. 
We say that a proposition $X$ is {\em toggle-High dependent at $\alpha$} if there exists a sequence $\beta \in A^*$ such that 
$\togL(\alpha) = \togL(\beta)$ but $\alpha\in X$ iff $\beta \not \in X$. Intuitively, this says that  changing  
i$\alpha$ by adding or deleting $H$ events and $S$ events that $L$ is not permitted to know about, we can change the value  
of the proposition $X$.%
\footnote{This definition of toggle-High-dependent differs from that used in the earlier version of
this work published at LAW12. The present formulation simplifies and strengthens 
 the corresponding result, Theorem ~\ref{thm:starlight-toggle-action-dependent}. 
 }  

\begin{theorem}  \label{thm:starlight-toggle-action-dependent}  
Suppose system $M$ is \TFF-compliant with  
  architectural specification  $(\strltArch, \archspec_{\strltArch})$.   
  If $\pi(p)$
is a  
toggle-$H$ dependent proposition at $\alpha$, then
$ \interpsysvf{\sysvar}{\pi}, \alpha \sat  \neg  \knows{L} p $.
\end{theorem}
\begin{proof} 
We 
first 
show that $ \tff{L}(\alpha) = \tff{L}(\togL(\alpha))$ for all $\alpha \in A^*$. 
The proof is by induction on $\alpha$. The base case of $\alpha = \epsilon$ is trivial. 
Consider the sequence $\alpha a$ and assume that $ \tff{L}(\alpha) = \tff{L}(\togL(\alpha))$. 
Let the canonical form of $\alpha $ be $\alpha = \alpha_0 \strltToggle  \alpha_1 \strltToggle  \alpha_2 \ldots \strltToggle  \alpha_n$

If $\dom(a)  = L$ then $\togL(\alpha a ) = \togL(\alpha) a$ and 
\begin{align*} 
\tff{L}(\alpha a)  & = \tff{L}(\alpha )\concat (\tff{L}(\alpha ),a)  \\ 
& = \tff{L}(\togL(\alpha) )\concat (\tff{L}\togL((\alpha )),a) & \text{by induction}   \\ 
& = \tff{L}(\togL(\alpha) a ) \\ 
& = \tff{L}(\togL(\alpha a ) ) \\ 
\end{align*} 
as required.
If $\dom(a)  = S$ and $\numof_\strltToggle(\alpha)$ is odd then $\togL(\alpha a ) = \togL(\alpha) a$ 
 and 
\begin{align*} 
\tff{L}(\alpha a)  & = \tff{L}(\alpha )\concat \intn(\strltFltr)(\alpha ,a)  \\ 
 & = \tff{L}(\togL(\alpha) )\concat \intn(\strltFltr)(\alpha ,a)  & \text{by induction}\\ 
& = \tff{L}(\togL(\alpha) )\concat \intn(\strltFltr)(\togL(\alpha ),a) & \text{by observation above}   \\ 
& = \tff{L}(\togL(\alpha) a )  \\ 
& = \tff{L}(\togL(\alpha a))   \\ 
\end{align*} 
as required. 
If $\dom(a)  = S$ and $\numof_\strltToggle(\alpha)$ is even then $\togL(\alpha a ) = \togL(\alpha) $ 
 and $ \intn(\strltFltr)(\alpha ,a) = \epsilon$, so 
\begin{align*} 
\tff{L}(\alpha a)  & = \tff{L}(\alpha )\concat \intn(\strltFltr)(\alpha ,a)  \\ 
  & = \tff{L}(\alpha )  \\ 
 & = \tff{L}(\togL(\alpha) ) & \text{by induction}\\ 
 & = \tff{L}(\togL(\alpha a))   \\ 
\end{align*} 
as required. 
Finally, if $\dom(a)  = H$  then $\togL(\alpha a ) = \togL(\alpha) $ , 
 and since $\notflowstoin{H}{L}$, we have 
\begin{align*} 
\tff{L}(\alpha a)  & = \tff{L}(\alpha )  \\ 
 & = \tff{L}(\togL(\alpha) ) & \text{by induction}\\ 
 & = \tff{L}(\togL(\alpha a))   \\ 
\end{align*} 
as required. 
This completes the proof that $\tff{L}(\alpha) = \tff{L}(\togL(\alpha))$. 

We now prove the theorem. Suppose system $M$ is \TFF-compliant with  
  architectural specification  $(\strltArch, \archspec_{\strltArch})$ and let 
   $\pi(p)$ be a toggle-$H$ dependent proposition
   at $\alpha$.  
   \commentout{ 
   We prove 
$ \interpsysvf{\sysvar}{\pi} \sat  \knows{L} p \rimp \knows{T}p $ 
by considering the contrapositive. Suppose that  
$ \interpsysvf{\sysvar}{\pi},\alpha  \sat  \neg  \knows{T}p $ for  some $\alpha \in A^*$. 
We prove $ \interpsysvf{\sysvar}{\pi},\alpha  \sat  \neg  \knows{L}p $. 
By the semantics of $\knows{T}$, there exists $\alpha' \in A^*$ such that 
$\numof_\strltToggle(\alpha ) = \numof_\strltToggle(\alpha')=n$ and 
$ \interpsysvf{\sysvar}{\pi},\alpha'  \sat  \neg  p $. 
Let the canonical forms of these sequences be 
$\alpha = \alpha_0 \strltToggle  \alpha_1 \strltToggle  \alpha_2 \ldots \strltToggle  \alpha_n$
and 
$\alpha' = \alpha'_0 \strltToggle  \alpha'_1 \strltToggle  \alpha'_2 \ldots \strltToggle  \alpha'_n$.
Define the sequence $\beta = \beta_0 \strltToggle  \beta_1 \strltToggle  \ldots \strltToggle  \beta_n$, 
where $\beta_i =( \alpha_i\restrict L )( \alpha'_i\restrict HS )$ when $i$ is even and
  $\beta_i =( \alpha_i\restrict LS )( \alpha'_i\restrict H )$ when $i$ is odd. 

Note that $\togL(\beta) = \gamma_0 \strltToggle  \gamma_1 \strltToggle  \ldots \strltToggle  \gamma_n$ where 
$\gamma_i = (( \alpha_i\restrict L )( \alpha'_i\restrict HS ))\restrict L =  \alpha_i\restrict L$
when $i$ is even and
  $\gamma_i =(( \alpha_i\restrict LS )( \alpha'_i\restrict H ))\restrict LS
  =  \alpha_i\restrict LS   $ when $i$ is odd. Thus $\togL(\alpha) = \togL(\beta)$. 
  It follows from the fact shown above that 
  $\tff{L}(\alpha) = \tff{L}(\togL(\alpha)) = \tff{L}(\togL(\beta)) = 
  \tff{L}(\beta) $. Since $M$ is \TFF-compliant, it follows that 
$\view{L}(\alpha) = \view{L}(\beta)$. 

Moreover, $\togH(\beta) = \delta_0 \strltToggle  \delta_1 \strltToggle  \ldots \strltToggle  \delta_n$ where 
$\delta_i = (( \alpha_i\restrict L )( \alpha'_i\restrict HS ))\restrict HS =  \alpha'_i\restrict HS$
when $i$ is even and
  $\delta_i =(( \alpha_i\restrict LS )( \alpha'_i\restrict H ))\restrict H
  =  \alpha'_i\restrict H$ when $i$ is odd. Thus $\togH(\beta) = \togH(\alpha')$. 
  Since $\pi(p)$ is toggle-$H$ dependent and 
$ \interpsysvf{\sysvar}{\pi},\alpha'  \sat  \neg  p $, it follows that 
$ \interpsysvf{\sysvar}{\pi},\beta  \sat  \neg  p $. As shown above, 
$\view{L}(\alpha) = \view{L}(\beta)$, so it follows that 
$ \interpsysvf{\sysvar}{\pi},\alpha  \sat  \neg  \knows{L}p $,as required. 
}%
Them there exists $\beta\in A^*$ such that $\togL(\alpha) =\togL(\beta)$ and 
$\alpha \in \pi(p)$ iff $\beta \not \in \pi(p)$. 
By what was shown above, we have 
$\tff{L}(\alpha) = \tff{L}(\beta)$. Since $M$ is \TFF-compliant, it follows that 
$\view{L}(\alpha) = \view{L}(\beta)$, and we have $ \interpsysvf{\sysvar}{\pi}, \alpha \sat  \neg  \knows{L} p $.
\end{proof}

Note that we would not be able show this
security property in an architecture without filter functions, since
the domain $S$ can communicate with both $H$ and $L$. It is the filter
function that allows us to show that $S$'s communication with $L$
reveals nothing about $H$ actions, and 
only
limited information about $S$
actions.

\commentout{ 
\iftr
\iffixlater
\subsubsection{Feasibility}

XXX TODO LOW PRIORITY

In this section we show that the requirements on the Starlight
Interactive Link system are feasible, by presenting a system
$\sysvar$ that satisfies the architecture and the local semantic
assumptions.

\newcommand{\switch}{\mathit{switch}}
The state of domain $H$ is a sequence of actions. We assume that $H$
observes all $H$, $L$, and $S$ events, and thus the state (and
observations) of $H$ after sequence $\alpha \in A^*$ is simply
$\alpha$. We assume that the state of $S$ is the sequence of all $S$
and $H$ actions performed so far and a single bit indicating if $H$ or
$L$ is selected. Initially $H$ is selected, and the selection is
flipped by each $t$ (toggle) action. So the state of $S$ after
sequence $\alpha \in A^*$ is the pair $(\alpha \restrict \{S,H\}, \switch)$
where $\switch=H$ if an even number of $t$ actions occur in $\alpha$, and
$\switch=L$ otherwise. Finally, we assume that the state of $L$ is the
sequence of all $L$ actions, all $t$ actions performed by $S$, and any
$S$ actions performed while the selection bit of $S$ equals $L$.

\begin{theorem}
  System $\sysvar$ satisfies \TFF-security.
\end{theorem}
\begin{proof}

  We need to show that for all agents $u$ and all $\alpha,
    \alpha' \in A^*$, if $\tff{u}(\alpha) = \tff{u}(\alpha')$ then
    $\obs_u(\run{s_0}{\alpha}) = \obs_u(\run{s_0}{\alpha'})$.

    We proceed by induction on $|\alpha| + |\alpha'|$. The base case,
    $\alpha = \alpha' = \epsilon$ is trivial. Assume that the
    inductive hypothesis holds for all sequences of shorter combined
    length.

    Consider the sequences $\alpha a$ and $\alpha'$. Suppose
    $\tff{u}(\alpha a) = \tff{u}(\alpha')$. We consider cases based on
    possibilities for agent $u$.

    \begin{enumerate}
    \item $u = H$
      Regardless of the domain of action $a$, we have
      $\flowstoin{\dom(a)}{H}$, and so $\tff{H}(\alpha a) =
      \tff{H}(\alpha)\fdefault_{\dom(a)}(\alpha a)$. Since
      $\tff{u}(\alpha a) = \tff{u}(\alpha')$ we must have $\alpha' =
      \beta b$ and $\tff{H}(\beta b) =
      \tff{H}(\beta)((\view{\dom(b)}(\beta), b), \epsilon)$ where
      $\tff{H}(\alpha) = \tff{H}(\beta)$. Therefore, $a=b$, and by the
      inductive hypothesis, $\obs_H(\run{s_0}{\alpha}) =
      \obs_H(\run{s_0}{\beta})$. Since the state of $H$ is simply the
      sequence of all actions, $\alpha = \beta$, and so,
      $\obs_H(\run{s_0}{\alpha a}) = \alpha a = \beta b =
      \obs_H(\run{s_0}{\beta b})$ as required.

    \item $u = S$

      The state of $S$ is the sequence of $S$ actions performed so
      far, and a bit indicating whether an even or odd number of $t$
      actions have been performed. If $\dom(a) = L$, then
      $\tff{S}(\alpha a) = \tff{S}(\alpha)$, and $\obs_S(\alpha a) =
      \obs_S(\alpha)$. By the inductive hypothesis, $\obs_S(\alpha) =
      \obs_S(\alpha')$ and so $\obs_S(\alpha a) = \obs_S(\alpha')$. If
      $\dom(a) \in \{H,S\}$ then $\tff{S}(\alpha a) =
      \tff{S}(\alpha)((\view{\dom(a)}(\alpha), a), \epsilon)$. Since
      $\tff{u}(\alpha a) = \tff{u}(\alpha')$ we must have $\alpha' =
      \beta b$ and $\tff{S}(\beta b) =
      \tff{S}(\beta)((\view{\dom(b)}(\beta), b), \epsilon)$ where
      $\tff{H}(\alpha) = \tff{H}(\beta)$. Therefore, $a=b$, and by the
      inductive hypothesis, $\obs_S(\run{s_0}{\alpha}) =
      \obs_S(\run{s_0}{\beta})$. Since the state of $S$ is the
      sequence of all $S$ and $H$ actions, and a bit that is toggled
      with every $t$ action, we have $\alpha \restrict \{H,S\} = \beta
      \restrict \{H,S\}$, and $\obs_S(\run{s_0}{\alpha}) =
      \obs_S(\run{s_0}{\beta}) = (\alpha\restrict\{H,S\}, s)$. Thus,
      $\obs_H(\run{s_0}{\alpha a}) = (\alpha a\restrict \{H,S\}, s') =
      (\beta b \restrict \{H,S\}, s') = \obs_S(\run{s_0}{\beta b})$,
      where $s' = s$ if $a \ne t$, $s' =H$ if $s=L$ and $a=t$, and $s'
      = L$ otherwise.

    \item $u = L$

      The state of $L$ is the sequence of all $L$ actions, all $t$
      actions performed by $S$, and any $S$ actions performed while
      the selection bit of $S$ equals $L$. If
      $\dom(a) = H$, then $\tff{L}(\alpha a) = \tff{L}(\alpha)$, and
      $\obs_L(\alpha a) = \obs_L(\alpha)$. By the inductive
      hypothesis, $\obs_L(\alpha) = \obs_L(\alpha')$ and so
      $\obs_L(\alpha a) = \obs_L(\alpha')$. 

      If $\dom(a) = L$ then $\tff{L}(\alpha a) =
      \tff{L}(\alpha)((\view{L}(\alpha), a), \epsilon)$. Since
      $\tff{u}(\alpha a) = \tff{u}(\alpha')$ we must have $\alpha' =
      \beta b$ and $\tff{L}(\beta b) =
      \tff{L}(\beta)((\view{\dom(b)}(\beta), b), \epsilon)$ where
      $\tff{L}(\alpha) = \tff{L}(\beta)$. Therefore, $a=b$, and by the
      inductive hypothesis, $\obs_L(\run{s_0}{\alpha}) =
      \obs_L(\run{s_0}{\beta})$. Thus, $\obs_L(\run{s_0}{\alpha a}) =
      \obs_L(\run{s_0}{\alpha}) a = \obs_L(\run{s_0}{\beta}) b =
      \obs_L(\run{s_0}{\beta b})$ as required.

      If $\dom(a) = S$ then $\tff{L}(\alpha a) =
      \tff{L}(\alpha)(\strltFltr(\alpha a), \epsilon)$.  If $\strltFltr(\alpha a)
      = \epsilon$, then $\tff{L}(\alpha a) = \tff{L}(\alpha)$, and $a
      \ne t$ and $\alpha$ contains an even number of $t$
      actions. Since the bit in the state of $S$ is set to $L$ only
      when there is an odd number of $t$ actions, $\obs_L(\alpha a) =
      \obs_L(\alpha)$. By the inductive hypothesis, $\obs_L(\alpha) =
      \obs_L(\alpha')$, and so $\obs_L(\alpha a) = \obs_L(\alpha')$ as
      required.

      Otherwise, $\strltFltr(\alpha a) = a$ and $a = t$ or $\alpha$ contains
      an odd number of $t$ actions. Since the bit in the state of $S$
      is set to $L$ only when there is an odd number of $t$ actions,
      $\obs_L(\alpha a) = \obs_L(\alpha) a$.  Since $\tff{u}(\alpha a)
      = \tff{u}(\alpha')$ we must have $\alpha' = \beta b$ and
      $\tff{L}(\beta b) = \tff{L}(\beta)(a, \epsilon)$ where
      $\tff{L}(\alpha) = \tff{L}(\beta)$ and $\strltFltr(\beta b) = b =
      a$. Thus, either $b=t$ or $\beta$ contains an odd number of $t$
      actions. Either way, $\obs_L(\beta b) = \obs_L(\beta) b$.  By
      the inductive hypothesis, $\obs_L(\alpha) = \obs_L(\beta)$, and so
      $\obs_L(\alpha a) = \obs_L(\beta b)$ as required.

    \end{enumerate}

\end{proof}
\fi %
\fi %

}%

\subsection{Example: Downgrader}\label{sec:downgrader}

Many systems downgrade information, releasing confidential information
to untrusted entities. Because it is a sensitive operation,
downgrading is typically restricted to certain trusted components,
called downgraders.

\newcommand{\downarch}{\ensuremath{\mathcal{DG}}}

The following diagram shows
extended
architecture $\downarch$, containing
low-security domain $L$, 
high-security data store $H$ for two domains of high-security users $ C$ and $P$, 
and downgrader $D$.
\begin{center}
  \includegraphics[height=2.5\diagsize]{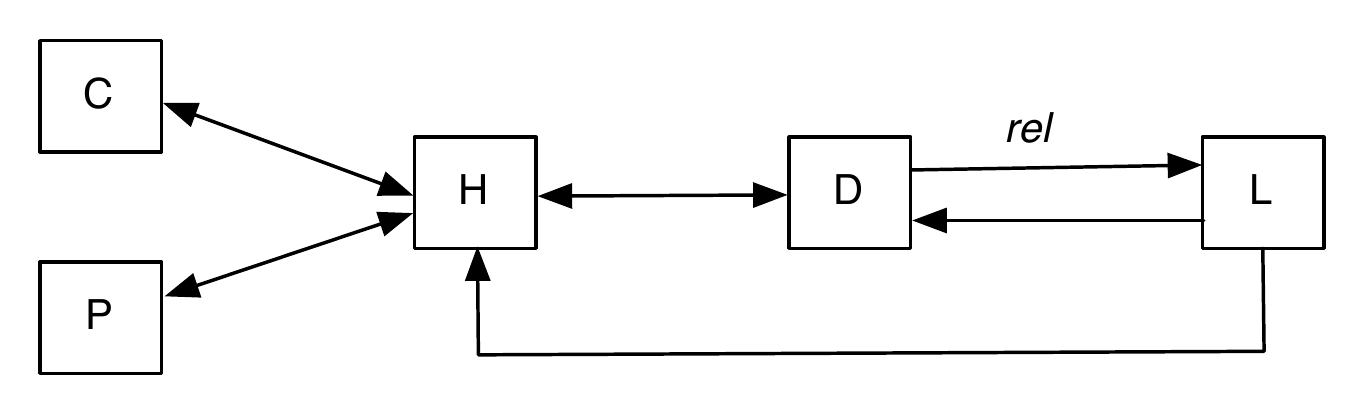}
\end{center}

This architecture represents a system where some information is
allowed to flow from the high-security domain $H$ to low-security
domain $L$, but only through the  downgrader $D$. In
particular, we assume that 
domain $L$ is permitted to 
know 
about the actions
of $P$ (via downgrader $D$), but should never 
know 
anything
about the actions of $ C$. This may be an appropriate model for a
government agency where some sensitive information may be released to
the public, but certain sensitive 
 data (e.g., the identity
 and activity of covert employees, represented by domain $C$) should never be released.

\newcommand{\downrelf}{\ensuremath{\mathit{rel}}} 

The filter function for $\downrelf$ restricts the
information that may be released from $D$ to $L$. 
It is a local constraint on the trusted component $D$.
The filter function
should ensure that nothing is revealed about the actions of $ C$. We
define an architectural specification $\archspec_\downarch$ to
restrict our attention to architectural interpretations with suitable
filter functions.

We first define an interpretation of the function $\downrelf$ that 
states that the information transmitted across this edge is the maximum 
information that $D$ would have if the domain $ C$ were completely cut off 
from the other domains. 
Let $\flowsto'$ be the information flow policy  
such that $(u,v,f) \in \flowsto'$ if and only if $f = \top$ and 
either $u= v= C$ or there exists an edge $(u,v,g)\in \flowsto$ and $\{u,v\} \subseteq \{P, H, D,L\}$.
This is the policy depicted in the following diagram, where, as usual, we 
omit reflexive edges: 
\begin{center}
  \includegraphics[height=2.5\diagsize]{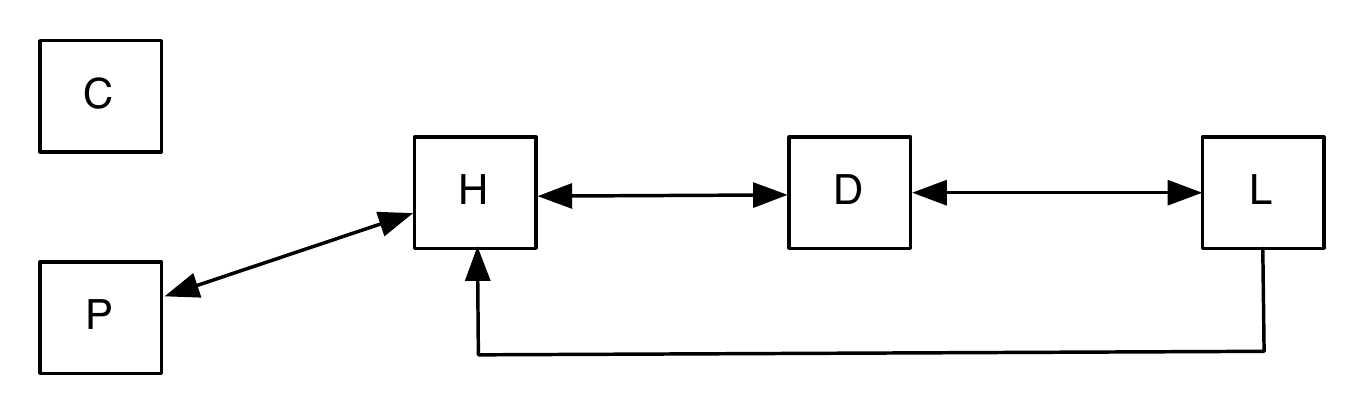}
\end{center}

Let $(A, \dom, \intn) \in \archspec_{\downarch}$ if and only if $(A,
\dom, \intn)$ is an 
interpretation of $\downarch$ such
that
 $\intn(\downrelf)(\alpha,a) = (\tff{D}^{\flowsto'}(\alpha), a)$
 for all $\alpha \in A^*$ and $a\in A$ with $\dom(a) = D$. 
Thus,
architectural interpretations in $\archspec_{\downarch}$ ensure that
$\tff{D}^{\flowsto'}$ is an upper-bound on information that may be
released from $D$ to $L$. 
It is straightforward to check that these architectural specifications are non-conflating. 
The required compatibility constraint is also satisfied, as shown in the following result. 

\begin{proposition} 
If $(A, \dom, \intn) \in \archspec_{\downarch}$ then 
$\intn(\downrelf)$ is $\tff{D}^{\flowsto}$-compatible. 
\end{proposition} 

\begin{proof} 
For $u\in \{P,H, D, L\}$, define the functions $F_u$ inductively, by $F_u(\epsilon) = \epsilon$ and 
$$ 
\begin{array}{l} 
F_u(\sigma \concat (\delta, a)) \\
\quad\quad= \begin{cases} 
F_u(\sigma) & \text {if $dom(a) =  C$ } \\ 
F_u(\sigma) \concat (F_{\dom(a)}(\delta), a)& \text {if $dom(a) \neq  C$ and $\dom(a) \neq D$ and $u\neq L$}\\
F_u(\sigma) \concat (\delta, a)& \text {if  $\dom(a)= D$ (so $dom(a) \neq  C$) and $u =  L$}
\end{cases} 
\end{array} 
$$ 
We claim that for all $\alpha\in A^*$ and  $u\in \{P,H,D,L\}$, we have 
 $F_u(\tff{u}^{\flowsto}(\alpha)) = \tff{u}^{\flowsto'}(\alpha)$. 
It is then immediate that 
$\tff{D}^{\flowsto}(\alpha) = \tff{D}^{\flowsto}(\beta)$ implies $\tff{D}^{\flowsto'}(\alpha) = \tff{D}^{\flowsto'}(\beta)$, and it easily follows that  $\intn(\downrelf)$ is $\tff{D}^{\flowsto}$-compatible. 

To prove the claim we proceed by induction on $\alpha$. the base case of $\alpha = \epsilon$ is 
trivial. Assuming the claim holds for $\alpha$, 
let $u\in \{P,H,D,L\}$ and consider the sequence $\alpha a$. 
There are several cases: 
\begin{itemize} 
\item Case 1: $\flowstoin[\top]{\dom(a)}{u}$ and $\dom(a) \neq  C$.  
By definition of $\flowsto'$, we have $\flowstoinf[\top]{\dom(a)}{u}{\flowsto'}$. 
Thus, 
\begin{align*} 
F_u(\tff{u}(\alpha a) )& = F_u(\tff{u}(\alpha)\concat (\tff{\dom(a)}(\alpha),a)) \\ 
& = F_u(\tff{u}(\alpha)) \concat (F_{\dom(a)}(\tff{\dom(a)}(\alpha)),a) \\ 
& = \tff{u}^{\flowsto'}(\alpha) \concat (\tff{\dom(a)}^{\flowsto'}(\alpha),a) & \text{by induction}\\ 
& = \tff{u}^{\flowsto'}(\alpha a)\mathpunct. 
\end{align*} 

\item Case 2: $\flowstoin[\top]{\dom(a)}{u}$ and $\dom(a) =  C$.  
In this case,  $\flowstoinf{\dom(a)}{u}{\not \flowsto'}$. 
Thus, 
\begin{align*} 
F_u(\tff{u}(\alpha a) )& = F_u(\tff{u}(\alpha)\concat (\tff{\dom(a)}(\alpha),a)) \\ 
& = F_u(\tff{u}(\alpha))  \\ 
& = \tff{u}^{\flowsto'}(\alpha)  & \text{by induction}\\ 
& = \tff{u}^{\flowsto'}(\alpha a)\mathpunct. 
\end{align*}

\item Case 3: $\flowstoin[\downrelf]{\dom(a)}{u}$, so  $\dom(a)= D$ and  $u=L$.  
By definition of $\flowsto'$, we have $\flowstoinf[\top]{\dom(a)}{u}{\flowsto'}$. 
Thus, 
\begin{align*} 
F_u(\tff{u}(\alpha a) ) & = F_u(\tff{u}(\alpha)\concat \intn(\downrelf)(\alpha,a) \\ 
& = F_u(\tff{u}(\alpha)\concat (\tff{\dom(a)}^{\flowsto'}(\alpha),a)) \\ 
& = F_u(\tff{u}(\alpha)) \concat (\tff{\dom(a)}^{\flowsto'}(\alpha),a) & \text{by def. $F_u$}\\ 
& = \tff{u}^{\flowsto'}(\alpha) \concat (\tff{\dom(a)}^{\flowsto'}(\alpha),a) & \text{by induction}\\ 
& = \tff{u}^{\flowsto'}(\alpha a)\mathpunct. 
\end{align*}

\item Case 4: $\notflowstoin{\dom(a)}{u}$. 
In this case also $\flowstoinf{\dom(a)}{u}{\not \flowsto'}$. 
Thus, 
\begin{align*} 
F_u(\tff{u}(\alpha a) ) & = F_u(\tff{u}(\alpha)) \\
& = \tff{u}^{\flowsto'}(\alpha)  & \text{by induction}\\ 
& = \tff{u}^{\flowsto'}(\alpha a)\mathpunct. 
\end{align*} 
Thus, the claim holds in all cases. 
\end{itemize}

\end{proof}

\iftr
\iffixlater
\note{Should we talk about the fact that if a filter function is
  $\tff{D}^{\flowsto'}$-compatible, then it is $\tff{D}^{\flowsto}$
  compatible? I believe that this follows quite easily from some of
  the refinement results...}
\fi %
\fi %

The architectural specification 
imposes restrictions on what information
 $L$ 
 knows 
 and
when. Domain $L$ can only learn about other domains via
downgrader $D$.  The restriction on filter functions for $\downrelf$
ensures that $L$ cannot 
know 
anything about the actions of $ C$,
since $\notflowstoin[\flowsto']{ C}{H}$, and so the function
$\tff{D}^{\flowsto'}$ does not contain any information about the
actions of $ C$. However, information about the actions of $H$, $P$ and $D$ may
be released to $L$.
\iftr
Also, $L$
cannot learn about actions of other domains that occurred after the
last $D$ action.
The following two theorems express these restrictions.
\fi

\begin{theorem}\label{thm:down-hc-action-local}
  If $\sysvar$ is 
\TFF-compliant
 with architectural specification $(\downarch, \archspec_\downarch)$ and $\pi(p)$ 
depends on $G$ actions at $\alpha$, then 
  $\interpsysv, \alpha \sat \neg \knows {L} p$.
\end{theorem}
\begin{proof}
Let   $G =  \overline{C} = \{P,H,D,L\}$.  
    We first claim that $\tff{u}^{\flowsto'}(\gamma \restrict  G) = \tff{u}^{\flowsto'}(\gamma)$ for any $\gamma\in A^*$ and 
    $u \in G$. 
  We proceed by induction on $|\gamma|$. 
The base
  case is trivial. Suppose that $\tff{u}^{\flowsto'}(\gamma \restrict
  G) = \tff{u}^{\flowsto'}(\gamma)$ for $u \in G$ and consider $\gamma
  a$. We consider several cases.
\begin{itemize}
\item If 
$\flowstoinf[f]{\dom(a)}{u}{\flowsto'}$  then $f= \top$ 
and $\dom(a) \in G$, by the definition of $\flowsto'$. Thus
\begin{align*}
  \tff{u}^{\flowsto'}((\gamma a) \restrict G) 
&=  \tff{u}^{\flowsto'}((\gamma \restrict G) a) \\
&= \tff{D}^{\flowsto'}(\gamma \restrict G) \concat (\tff{\dom(a)}(\gamma \restrict G), a) &\text{by defn $\tff{D}^{\flowsto'}$}\\
&= \tff{D}^{\flowsto'}(\gamma)\concat  (\tff{\dom(a)}(\gamma), a) &\text{by IH} \\
&= \tff{D}^{\flowsto'}(\gamma a)  &\text{by defn $\tff{D}^{\flowsto'}$} 
\end{align*}

\item If 
$\notflowstoin[\flowsto']{\dom(a)}{u}$ 
and $\dom(a) \in G$ then
\begin{align*}
  \tff{u}^{\flowsto'}((\gamma a) \restrict G) 
&=  \tff{u}^{\flowsto'}((\gamma \restrict G) a) \\
&=  \tff{u}^{\flowsto'}(\gamma \restrict G) &\text{by defn $\tff{u}^{\flowsto'}$} \\
&= \tff{u}^{\flowsto'}(\gamma)  &\text{by IH} \\
&= \tff{u}^{\flowsto'}(\gamma a)  &\text{by defn $\tff{u}^{\flowsto'}$} 
\end{align*}

\item If $\notflowstoin[\flowsto']{\dom(a)}{u}$ 
 and $\dom(a) \not\in G$ then
\begin{align*}
  \tff{u}^{\flowsto'}((\gamma a) \restrict G) 
&=  \tff{u}^{\flowsto'}(\gamma \restrict G) &\text{as $\dom(a) \not\in G$} \\
&= \tff{u}^{\flowsto'}(\gamma)  &\text{by IH} \\
&= \tff{u}^{\flowsto'}(\gamma a)  &\text{by defn $\tff{u}^{\flowsto'}$} 
\end{align*}
\end{itemize}
This completes the proof of the claim. 
In particular, we have $\tff{D}^{\flowsto'}(\gamma \restrict G) =
\tff{D}^{\flowsto'}(\gamma)$, for all $\gamma$. 
Note that this implies, for all $\gamma$ and $a\in A$ with $\dom(a) = D$, 
that $\intn(\downrelf)(\gamma \restrict G, a ) = 
(\tff{D}^{\flowsto'}(\gamma \restrict G), a) = 
(\tff{D}^{\flowsto'}(\gamma), a) =  \intn(\downrelf)(\gamma, a)$.

We now show that 
$\tff{L}(\gamma) = \tff{L}(\gamma \restrict G)$.
We proceed by
induction on $|\gamma|$. The base case is trivial. Consider $\gamma
a$, and assume that $\tff{L}(\gamma) = \tff{L}(\gamma \restrict G)$. 
We consider several cases.

\begin{itemize}
\item If $\dom(a) \in \{  C, P, H\}$ then 
$\notflowstoin{\dom(a)}{L}$, so 
$\tff{L}(\gamma a) =
  \tff{L}(\gamma)$ by definition of $\tff{L}$. 
  We have that $(\gamma a)\restrict G$ is either $\gamma\restrict G$ or $(\gamma\restrict G)a$, but in either case, 
 $\tff{L}((\gamma a)
  \restrict G) = \tff{L}(\gamma \restrict G)$
  by the definition of $\tff{L}$.  By the inductive hypothesis, we have
  $\tff{L}(\gamma a) = \tff{L}((\gamma a) \restrict G)$ as
  required.

  \item If $\dom(a) = L$ then
  $\gamma a\restrict G = (\gamma\restrict G)a $, so 
    \begin{align*}
      \tff{L}(\gamma a) &= \tff{L}(\gamma)\concat (\tff{L}(\gamma), a) \\
      &= \tff{L}(\gamma \restrict G) \concat (\tff{L}(\gamma \restrict G), a) &\text{by the IH} \\  
    & =   \tff{L}((\gamma \restrict G) a)  \\
   &= \tff{L}((\gamma a)\restrict G )
    \end{align*}

\item If $\dom(a) = D$ then
    \begin{align*}
      \tff{L}(\gamma a) &= \tff{L}(\gamma) \concat  \intn(\downrelf)(\gamma, a) \\
      &=  \tff{L}(\gamma) \concat  \intn(\downrelf)(\gamma \restrict G,a) &\text{as $\intn(\downrelf)(\gamma \restrict G, a) = \intn(\downrelf)(\gamma,a)$} \\   
      &=  \tff{L}(\gamma \restrict G) \concat \intn(\downrelf)(\gamma \restrict G,a) &\text{by the IH} \\
   &= \tff{L}((\gamma \restrict G) a ) \\
   &= \tff{L}((\gamma a)\restrict G ) & \text{since $\dom(a) = D\in G$}
    \end{align*}
\end{itemize}
This completes the proof that $\tff{L}(\gamma ) = \tff{L}(\gamma \restrict G)$ for all $\gamma$. 

Since $\pi(p)$ depends on $C$ actions at $\alpha$, 
there exists $\beta\in A^*$ such that $\alpha\restrict G = \beta\restrict G$ 
but $\alpha \in \pi(p) $ iff $\beta\not \in \pi(p)$. 
By the claim just proved, we have $\tff{L}(\alpha) = \tff{L}(\alpha\restrict G) = 
\tff{L}(\beta\restrict G) = \tff{L}(\beta)$, from which $\alpha \approx_L \beta$ follows by 
compliance. It follows that $M\,\pi, \alpha \models \neg \knows{L}(p)$. 
\end{proof}

Actions of the downgrader are the only means by which $L$ acquires
knowledge about other domains' actions. Thus, $L$ knows nothing about
any action of other domains that occur after the last $D$ action.
To state this formally, 
we say that proposition $X$
is \emph{about $ C$, $P$, and $H$ activity after the last $D$
  action} if there exists a set $Y \subseteq A^*$ such that for all
$\alpha \in A^*$,
$\alpha \in X \iff F(\alpha) \in Y$, where $F$ is
defined recursively by $F(\epsilon) = \epsilon$ and
\[
F(\beta a) = 
\begin{cases}
  \epsilon & \text{if $\dom(a)=D$} \\
  F(\beta) a & \text{if $\dom(a) \in \{ C, P, H\}$} \\
  F(\beta) & \text{otherwise} \\
\end{cases}
\]
Intuitively, $F$ returns
the $ C$, $P$, and $H$ actions that occur
after the last $D$ action, and so proposition $X$ is about $ C$,
$P$, and $H$ activity after the last $D$ action if membership in
$X$ depends only on the $ C$, $P$, and $H$ actions that occur after
the last $D$ action.

\begin{theorem}
  If $\sysvar$ is 
\TFF-compliant
 with architectural specification $(\downarch, \archspec_\downarch)$ and $\pi(p)$
  is a non-trivial proposition about $ C$, $P$, and $H$ activity
  after the last $D$ action, then $\interpsysv \sat \neg \knows {L}
  p$.
\end{theorem}
\begin{proof}
  Let $\alpha \in A^*$. Since $\pi(p)$ is non-trivial, there is an
  $\alpha'$ such that $\alpha \in \pi(p) \iff \alpha' \not\in
  \pi(p)$. Let $\alpha_1$ be the longest suffix of $\alpha$ such that
  $\alpha_1$ does not contain any $D$ actions. Then $\alpha = \alpha_0
  \alpha_1$, where either $\alpha_0 = \epsilon$ or the last action of
  $\alpha_0$ is a $D$ action.
  Let $\beta =  \alpha_0 (\alpha_1 \restrict \{L\})
  F(\alpha')$.  Note that $F(\beta) = F(\alpha')$, and so $\alpha \in
  \pi(p) \iff \beta \not\in \pi(p)$.
  We will show that $\tff{L}(\alpha) = \tff{L}(\beta)$.

  First, note that $\tff{L}(\beta) = \tff{L}(\alpha_0 (\alpha_1
  \restrict \{L\}) F(\alpha')) = \tff{L}(\alpha_0 (\alpha_1 \restrict
  \{L\}))$ since $F(\alpha')$ contains only actions $a$ such that
  $\notflowstoin{\dom(a)}{L}$.
  Note also that $\tff{L}(\alpha) = \tff{L}(\alpha_0 \alpha_1) =
  \tff{L}(\alpha_0 (\alpha_1 \restrict \{L\}))$ since $\alpha_1$
  contains only $L$ actions and actions $a$ such that
  $\notflowstoin{\dom(a)}{L}$. 

  Therefore, $\tff{L}(\beta) = \tff{L}(\alpha)$ as required.  Since
  $\sysvar$ is \TFF-compliant, by Lemma~\ref{lem:tff-view} we have
  $\view{L}(\alpha) = \view{L}(\beta)$, and so
  $\alpha \approx_L \beta$, which immediately
  implies the result.
\end{proof}

\subsection{Example: Electronic election} ~\label{sec:election} 
\newcommand{\announce}{\ensuremath{\mathit{results}}}
\newcommand{\elecarch}{\ensuremath{\mathcal{EE}}}
\newcommand{\elecauth}{\ensuremath{\mathit{ElecAuth}}}

The following diagram shows architecture $\elecarch$, an electronic
election for $n$ voters, coordinated by 
an election authority \elecauth.

\begin{center}
  \includegraphics[height=2.5\diagsize]{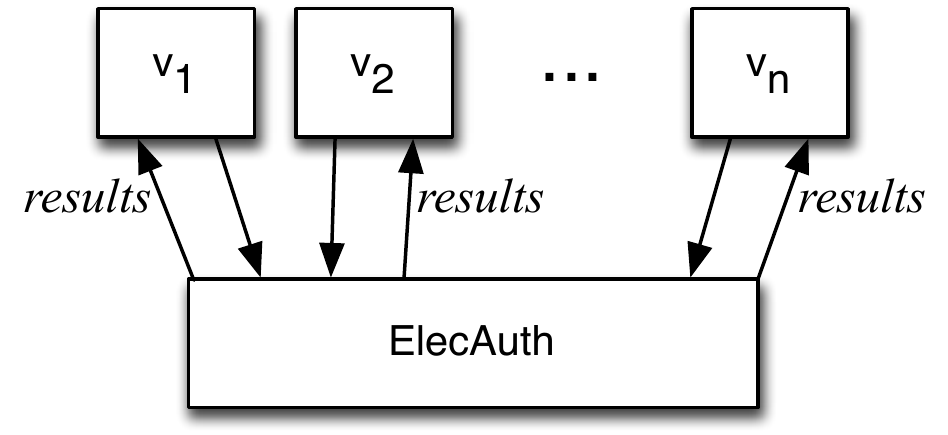}
\end{center}

There is one domain $v_i$ for each voter, and a single domain
$\elecauth$ for the election authority. The architecture permits
arbitrary information to flow from a voter to the election authority,
but uses filter function $\announce$ to restrict what information may
flow from the election authority to each voter.

In many elections, the behavior of individual voters is confidential
information: it should not be known how voters voted. (Elections have
several 
security requirements. For example, the final
result should be correctly computed from the votes---an integrity
requirement. We focus here on the confidentiality requirement for
voters.)

By specifying an additional local constraint on the election
authority, we can show that this architecture enforces anonymity on
the identity of the voters, thus satisfying the confidentiality
requirement about the behavior of voters. The election authority's
compliance with this local constraint might be assured by means of a
careful verification of its implementation, or carefully designed
cryptographic protocols (which may remove some or all of the trust
required to be placed in the election authority).

We first assume that voters are homogenous in that they have the same
set of actions. Let $(\elecarch, \archint)$ be an interpreted
architecture where $\archint = (A, \dom, \intn)$. We say that architecture interpretation
$\archint$ is \emph{voter homogenous} if for any voter 
$v$, the set
of possible actions $A_{v} = \{ a \in A \mid \dom(a) = v\}$ equals
$\{a^{v} \mid a \in A_V\}$, where $A_V$ is the set of 
action types 
available to voters. 
Intuitively, $a^v$ represent the action of type $a$ when performed by voter $v$. 

We define a voter permutation $P$ as a permutation over the set of
voters $V$. Since all voters have the same set of actions in a
voter-homogenous interpretation, we can apply a permutation $P$ to
sequence $\alpha \in A^*$, written $P(\alpha)$, as 
follows: 
\begin{align*}
  P(\epsilon) &= \epsilon \\
  P(\alpha a) &= P(\alpha) a &&\text{if } \dom(a) \not\in V \\
  P(\alpha a^{v}) &= P(\alpha) a^{P(v)} &&\text{if } 
  a \in A_V 
  ~.
\end{align*}
We apply permutation $P$ to proposition $X \subseteq A^*$ by applying
$P$ to each sequence $\alpha \in X$, i.e.,    
\[ P(X) = \{ P(\alpha) \mid \alpha \in X \}~. \] 

Using voter permutations, we can now state the local constraint that
election results do not depend on (and thus do not reveal) the identity
of any voter.
\begin{itemize} 
\item[A1.]  Election results have the following
  \emph{identity-oblivious} property: Given voter-homogenous interpreted architecture
  $(\elecarch, \archint)$ where $\archint = (A, \dom, \intn)$, for all
  voter permutations $P$, 
  sequences $\alpha \in A^*$,  and actions $a\in A$ with 
  $\dom(a) = \elecauth$, we have $\intn(\announce)(\alpha, a) = \intn(\announce)(P(\alpha), a)$.
\end{itemize}

There are several possible interpretations of 
$\announce$ that satisfy this constraint, such as a function that
returns each candidate 
and  her vote tally, or
a function that returns the total number of votes submitted.  However,
a function that returns a ballot and the identity of
the voter that submitted it doesn't satisfy constraint A1.

We define architecture specification $\archspec_{\elecarch}$ such that
 $\archint  \in \archspec_{\elecarch}$ if and only if $\archint$ is
a voter-homogenous 
interpretation of $\elecarch$ that
satisfies constraint A1.

Given 
the local constraint A1, 
we can show that if voter $v$
believes that some proposition $X$ may be satisfied, then $v$ also
believes that $P(X)$ may be satisfied, for 
any 
voter permutation $P$
with $P(v) =v$. 
For
example, if Alice considers it possible that Bob voted for Obama and
Charlie voted for Romney, then Alice considers it possible that
Charlie voted for Obama and Bob voted for Romney.

\iftr
Before stating and proving 
this result,
we first prove a useful lemma: a voter's view of an
election does not depend on the identity of other voters. That is, for
any permutation $P$ such that $P(v) = v$, the sequences $\alpha$
and $P(\alpha)$ appear the same to voter $v$.

\begin{lemma}\label{lem:voter-identity}
  Let system $\sysvar$ be \TFF-compliant with $(\elecarch,
  \archspec_{\elecarch})$. For all voters $v$, voter permutations
  $P$ such that $P(v) = v$, and sequences $\alpha$, $\alpha
  \approx_{v} P(\alpha)$.
\end{lemma}
\begin{proof}
   Since $\sysvar$ is \TFF-compliant with $(\elecarch,
  \archspec_{\elecarch})$, there is an $\archint = (A, \dom, \intn)\in \archspec_{\elecarch}$ such
  that $\sysvar$ is \TFF-compliant with interpreted architecture
  $(\elecarch,\archint)$.

   We show, by induction on the length of $\alpha$, that
  $\tff{v}(\alpha) = \tff{v}(P(\alpha))$.  The base case is trivial. Consider sequence
  $\alpha a$, and assume  that 
  $\tff{v}(\alpha) = \tff{v}(P(\alpha))$.
  We show that $\tff{v}(\alpha a) = \tff{v}(P(\alpha a))$.

  If $\dom(a) = v$ then 
  \begin{align*}
\tff{v}(P(\alpha a)) &=  \tff{v}(P(\alpha) a) & \text{as $P(a) = a$}\\
& = \tff{v}(P(\alpha))\concat (\tff{v}(P(\alpha)), a) & \text{by defn $\tff{v}$} \\
& = \tff{v}(\alpha) \concat (\tff{v}(\alpha), a) & \text{by IH} \\
& = \tff{v}(\alpha a) &  \text{by defn $\tff{v}$}     
  \end{align*}

  If $\dom(a) \in V\setminus \{v\}$ then $\notflowstoin{\dom(a)}{v}$, so 
  \begin{align*}
    \tff{v}(P(\alpha a)) &= \tff{v}(P(\alpha) b) & 
    \text{for some $b \in A$, $\dom(b) \in V\setminus\{v\}$}\\
    & = \tff{v}(P(\alpha)) & \text{by defn $\tff{v}$} \\
  &   = \tff{v}(\alpha)  & \text{by IH} \\
 &   = \tff{v}(\alpha a)  & \text{by defn $\tff{v}$} 
  \end{align*}

  If $\dom(a) = \elecauth$ then 
  \begin{align*}
\tff{v}(P(\alpha a)) &=  \tff{v}(P(\alpha) a) & \text{as $P(a) = a$}\\
& = \tff{v}(P(\alpha))\concat  \intn(\announce)(P(\alpha), a) & \text{by defn $\tff{v}$} \\
& = \tff{v}(\alpha) \concat \intn(\announce)(P(\alpha), a)) & \text{by IH} \\
& = \tff{v}(\alpha)\concat  \intn(\announce)(\alpha, a) & \text{by A1} \\
& = \tff{v}(\alpha a) &  \text{by defn $\tff{v}$}        
  \end{align*}

Thus, if all cases, $\tff{v}(\alpha) = \tff{v}(P(\alpha))$. Therefore, by Lemma~\ref{lem:tff-view}, $\alpha
  \approx_{v} P(\alpha)$. 
\end{proof}
\fi %

\begin{theorem}\label{thm:elections}
  Let system $\sysvar$ be \TFF-compliant with 
  architectural specification
  $(\elecarch,
  \archspec_{\elecarch})$. Let $v$ be a voter. For all voter permutations $P$
  such that $P(v) = v$, if 
$\pi(q) = P(\pi(p))$,  then $ \interpsysf{\sysvar}{\pi}{\alpha} \sat \neg\knows {v} \neg p
\rimp \neg\knows {v} \neg q$ for all $\alpha \in A^*$.  
\end{theorem}
\begin{proof}
Suppose   $\interpsysf{\sysvar}{\pi}{\alpha} \sat \neg\knows {v} \neg p$. Then  
there is some sequence $\beta$ such that
  $\view{v}(\alpha) = \view{v}(\beta)$ and $\beta \in \pi(p)$.
  By Lemma~\ref{lem:voter-identity}, $\view{v}(\beta) =
  \view{v}(P(\beta))$, so $\view{v}(\alpha) =
  \view{v}(P(\beta))$. Since $P(\beta) \in P(\pi(p))$, we have
$\interpsysf{\sysvar}{\pi}{P(\beta)} \sat q$, so $\interpsysf{\sysvar}{\pi}{\alpha} \sat \neg\knows {v} \neg q$ 
as required.
\end{proof}

Note that Theorem~\ref{thm:elections} does not imply that voter $v$
learns nothing about other voters. For example, if the election
results reveal that all voters voted for Obama, then voter $v$ knows
how every other voter voted. Also, if there are only 2 voters, $v$ and
$w$, then the results may reveal to $v$ exactly how $w$ voted.
However, Theorem~\ref{thm:elections} provides anonymity: given the
results, voter $v$ cannot distinguish the behavior of other voters.

\iftr
\iffixlater
\subsection{Feasibility}
XXX TODO LOW PRIORITY

In this section we show that the requirements on an election system are
feasible, by presenting a system $\sysvar$ that satisfies the architecture
and the requisite semantic properties.

\newcommand{\bidaction}[2]{\mathit{Bid}_{#1}(#2)}
\newcommand{\announceAction}{\mathit{hiBid}}
Let the set of actions $A$ be defined as 
\begin{align*}
  A &= \{ \bidaction{b}{n} \mid b \text{ is a voter}, n \in \Nat \}\\
    & \quad \cup \{ \announceAction \}
\intertext{where}
\dom(\bidaction{b}{n}) &= b \\
\dom(\announceAction) &= \elecauth \\
\end{align*}

Voters may issue a bid action whenever they desire. 
The
election authority observes all actions by voters. We use $V$ to
denote the set of voters, and require that there are at least 3
voters.

The auctioneer may issue a $\announceAction$ action at any time. The
effect of a $\announceAction$ action is to post to the bulletin board
the highest bid since the last $\announceAction$.

Each voter is able to observe its history of bids, and the state of
the bulletin board. Thus, a voter's state contains its history of
bids, and the state of the bulletin board. The bulletin board records
the announcements of the auctioneer, and so the bulletin board state
is simply a sequence of integers, representing the highest bids.  The
auctioneer observes the occurrence of each bid and its own
$\announceAction$ actions, and the auctioneer's state is a sequence of
actions.

The function $\announce(\alpha)$ returns the highest bid made between the last two auctioneer actions (or, if only a single auctioneer action was made, before the last auctioneer action).
\begin{align*}
  \announce(\epsilon) &= \epsilon \\
  \announce(\alpha a) &= 
  \begin{cases}
    \announce(\alpha) &\text{if $\dom(a) \ne \elecauth$}\\    
    \announce'(\alpha) &\text{otherwise}\\    
  \end{cases} \\
  \announce'(\epsilon) &= 0 \\
  \announce'(\alpha~\announceAction) &= 0 \\
  \announce'(\alpha~ \bidaction{b}{n}) &= \max(n, \announce'(\alpha)) \\
\end{align*}
It returns a pair of the maximum bid that occurred between the last
two $\announceAction$ actions, and the number of $\announceAction$
actions that have previously occurred.

\begin{theorem}
  System $\sysvar$ satisfies conditions A1--A4, BB1--BB3, and B1.
\end{theorem}
\begin{proof}
  \begin{itemize}
  \item[A1.]  As there is only a single auctioneer action, A1 is satisfied
  trivially. 

  \item[A2.]
  Let $\alpha$ and $\alpha'$ be either $\epsilon$ or a sequence of
  actions in $(A_{v_1}\cup \ldots\cup A_{v_n})^*\cdot A_{\elecauth}$
  such that $\announce(\alpha) = \announce(\alpha')$. Let
  $\beta,\beta'\in (A_{v_1}\cup \ldots\cup A_{v_n})^*$ such that
  $\beta\restrict v_i = \beta'\restrict v_i$ for all $i=1\ldots n$.
  The function $\announce$ returns the maximum bid that occurred
  between the last auctioneer actions, and so $\announce(\alpha\beta
  \announceAction)$ is the maximum bid occurring in $\beta$, and
  $\announce(\alpha' \beta' \announceAction)$ is the maximum bid
  occurring in $\beta'$. Since $\beta\restrict v_i = \beta'\restrict
  v_i$ for all $i=1\ldots n$, this maximum is the same, and so
  $\announce(\alpha\beta \announceAction) = \announce(\alpha' \beta'
  \announceAction)$ as required.

\item[A3.] Let $\alpha$ and $\alpha'$ be either $\epsilon$ or a
  sequence of actions in $(A_{v_1}\cup \ldots\cup A_{v_n})^*\cdot
  A_{\elecauth}$.  Let $\beta,\beta'\in (A_{v_1}\cup \ldots\cup
  A_{v_n})^*$ and $\dom(a) = \elecauth$. Let $v_i,v_j$ and $v_k$ be
  distinct voters, and let $V$ be the set of all voters.  Suppose
  $n = \announce(\alpha\beta a) = \announce(\alpha\beta' a)$. Then the
  maximum bid in $\beta$ is equal to the maximum bid in $\beta'$.

  Let $\gamma = \beta \restrict(B-\{v_j\})$. If the maximum bid in
  $\gamma$ equals $n$, then let $\beta'' = \gamma \cdot (\beta' \restrict
  v_j)$. 

  Otherwise, the maximum bid in $\gamma$ must be less than $n$. Let
  $\beta'' = \gamma \cdot (\beta' \restrict v_j) \cdot \bidaction{v_k}{n}$.

  Note that in both cases, $\beta''\restrict v_i= \beta \restrict v_i$
  and $\beta''\restrict v_j = \beta' \restrict v_j$ and
  $\announce(\alpha\beta a) = \announce(\alpha\beta'' a) =
  \announce(\alpha\beta' a) = n$ as required.

  \item[A4.] For all voter permutations $P$ and sequences $\alpha$,
  $\announce(\alpha) = \announce(P(\alpha))$, since
  $\announce(\alpha)$ depends on the amounts bid, not the identity of
  the voter.

  \item[BB1.] The bulletin board has no actions.
  \item[BB2.] The state of the bulletin board is a sequence of integers
    representing the high bids announced by the auctioneer. Thus, each
    auctioneer action increases the length of the sequence by one
    element, and so $\obs_{BB}(\alpha \announceAction) \ne
    \obs_{BB}(\alpha)$.
  \item[BB3.] Suppose that $\alpha$ and $\alpha'$ are sequences such
    that $\obs_{BB}(\alpha \announceAction) = \obs_{BB}(\alpha'
    \announceAction)$. The state of the bulletin board is a sequence
    of integers representing the high bids announced by the
    auctioneer. Thus, the most recent announcement by the auctioneer
    in $\alpha \announceAction$ is the same as the announcement in
    $\alpha' \announceAction$, that is, $\announce(\alpha
    \announceAction) = \announce(\alpha' \announceAction)$.

  \item[B1.] All voters have the same set of actions, as the set of
    possible actions for voter $b$ is $\{\bidaction{b}{n}\mid b \in
    \Nat\}$.

\end{itemize}
\end{proof}

\begin{theorem}
  System $\sysvar$ satisfies \TFF-security.
\end{theorem}
\begin{proof}
  \begin{enumerate}
  \item We need to show that for all agents $u$ and $v$ and all $s,t
    \in S$ if $\obsflowstoin{u}{v}$ and $\obs_v(s) = \obs_v(t)$ then
    $\obs_u(s) = \obs_u(t)$.  From the architecture, if
    $\obsflowstoin{u}{v}$, then $u = BB$, and $v$ is a voter. Since
    voter's observations are defined to include the bulletin board's
    observations, the result holds trivially.

  \item We need to show that for all agents $u$ and all $\alpha,
    \alpha' \in A^*$, if $\tff{u}(\alpha) = \tff{u}(\alpha')$ then
    $\obs_u(\run{s_0}{\alpha}) = \obs_u(\run{s_0}{\alpha'})$.

    We proceed by induction on $|\alpha| + |\alpha'|$. The base case,
    $\alpha = \alpha' = \epsilon$ is trivial. Assume that the
    inductive hypothesis holds for all sequences of shorter combined
    length.

    Consider the sequences $\alpha a$ and $\alpha'$. Suppose
    $\tff{u}(\alpha a) = \tff{u}(\alpha')$. We consider cases based on
    possibilities for agent $u$.

    \begin{enumerate}
    \item $u = BB$
      
      Since $BB$ has no actions, either $\dom(a) = \elecauth$ or
      $\dom(a) \in B$.  If $\dom(a) \in B$ then $\tff{BB}(\alpha a) =
      \tff{BB}(\alpha)$, and by the inductive hypothesis,
      $\obs_{BB}(\run{s_0}{\alpha}) =
      \obs_{BB}(\run{s_0}{\alpha'})$. Since $BB$ only observes the
      auctioneer's actions, $\obs_{BB}(\run{s_0}{\alpha}) =
      \obs_{BB}(\run{s_0}{\alpha a})$, and the result holds.

      If $\dom(a) = \elecauth$, then $a = \announceAction$ and 
      \begin{align*}
        \tff{BB}(\alpha a) &= \tff{BB}(\alpha) (\announce(\alpha a), \epsilon) 
      \end{align*}
      Since $\tff{u}(\alpha a) = \tff{u}(\alpha')$, $\alpha'$ cannot
      be $\epsilon$. Let $\alpha' = \beta b$. If $\dom(b) \in B$ then
      $\tff{BB}(\beta b) = \tff{BB}(\beta)$, and by the inductive
      hypothesis, $\obs_{BB}(\run{s_0}{\alpha a}) =
      \obs_{BB}(\run{s_0}{\beta})$, and, since $BB$ only observes the
      auctioneer's actions, $\obs_{BB}(\run{s_0}{\beta}) =
      \obs_{BB}(\run{s_0}{\beta b})$, and the result holds. So assume
      that $\dom(b) = \elecauth$, and so $b = a = \announceAction$
      and $\announce(\alpha a) = \announce(\beta b)$. Since
      $\tff{BB}(\alpha) = \tff{BB}(\beta)$, by the inductive
      hypothesis $\obs_{BB}(\run{s_0}{\alpha}) =
      \obs_{BB}(\run{s_0}{\beta})$. Since $BB$ observes the sequence
      of the function evaluations $\announce$, and $\announce(\alpha
      a) = \announce(\beta b)$, we have $\obs_{BB}(\run{s_0}{\alpha
        a}) = \obs_{BB}(\run{s_0}{\beta b})$, as required.
 
    \item $u = \elecauth$

      Either $\dom(a) = \elecauth$ or $\dom(a) \in B$. If $\dom(a) =
      \elecauth$, then $\tff{\elecauth}(\alpha a) =
      \tff{\elecauth}(\alpha) (f(\alpha a), \epsilon)$.  Since
      $\tff{u}(\alpha a) = \tff{u}(\alpha')$, $\alpha'$ cannot be
      $\epsilon$. Let $\alpha' = \beta b$. Note that
      $\obsflowstoin{\dom(b)}{\elecauth}$, and so let
      $\tff{\elecauth}(\beta b) = \tff{\elecauth}(\beta) (g(\beta
      b), \epsilon)$. By the assumption on architecture \note{refer to
        this?}, $b = a$. Since $\tff{\elecauth}(\alpha) =
      \tff{\elecauth}(\beta)$, by the inductive hypothesis
      $\obs_{\elecauth}(\run{s_0}{\alpha}) =
      \obs_{\elecauth}(\run{s_0}{\beta})$. Since the auctioneer's
      state is a sequence of the voter and auctioneer actions, and
      $a=b$, we have $\obs_{\elecauth}(\run{s_0}{\alpha a}) =
      \obs_{\elecauth}(\run{s_0}{\beta b})$, as required.

    \item $u = v_i$ for $v_i \in B$

      Either $\dom(a) = \elecauth$ or $\dom(a) \in B$. If $\dom(a) =
      \elecauth$ then $\tff{v_i}(\alpha a) = \tff{v_i}(\alpha)
      (\epsilon, \obs_{BB}(\run{s_0}{\alpha a}))$.  Since
      $\tff{u}(\alpha a) = \tff{u}(\alpha')$, $\alpha'$ cannot be
      $\epsilon$. Let $\alpha' = \beta b$. If $\tff{v_i}(\beta b) =
      \tff{v_i}(\beta)$ then $b \in B-\{v_i\}$ and $\obs_{v_i}(\beta
      b) = \obs_{v_i}(\beta)$, and the result holds by the inductive
      hypothesis. So assume $\tff{v_i}(\beta b) \ne
      \tff{v_i}(\beta)$. By Lemma~\ref{lem:tff-well-founded}, $\dom(b)
      \in \{\elecauth, v_i\}$. If $\dom(b) = v_i$, then
      $\tff{v_i}(\beta b) \ne \tff{v_i}(\alpha a)$, a
      contradiction. So $\dom(b) = \elecauth$, and, since
      $\tff{BB}(\alpha a) = \tff{BB}(\beta b)$,
      $\obs_{BB}(\run{s_0}{\alpha a}) = \obs_{BB}(\run{s_0}{\beta
        b})$. By the inductive hypothesis,
      $\obs_{v_i}(\run{s_0}{\alpha})=
      \obs_{v_i}(\run{s_0}{\beta})$. Since the observations of $v_i$
      are simply the sequence of its actions and the states of the
      bulletin board, $\obs_{v_i}(\run{s_0}{\alpha a})=
      \obs_{v_i}(\run{s_0}{\beta b})$ as required.

      Assume that $\dom(a) \in B$. If $\dom(a) \ne v_i$ then
      $\tff{v_i}(\alpha a) = \tff{v_i}(\alpha)$ and by
      Lemma~\ref{lem:tff-well-founded}, $\obs_{v_i}(\run{s_0}{\alpha
        a}) = \obs_{v_i}(\run{s_0}{\alpha})$ and the result holds by
      the inductive hypothesis. Assume that $\dom(a) = v_i$. Then
      $\tff{v_i}(\alpha a) = \tff{v_i}(\alpha) (f(\alpha a),
      \epsilon)$.  Since $\tff{u}(\alpha a) = \tff{u}(\alpha')$,
      $\alpha'$ cannot be $\epsilon$. Let $\alpha' = \beta b$.  If
      $\dom(b) \in B - \{v_i\}$ then $\tff{v_i}(\beta b) =
      \tff{v_i}(\beta)$, and $\obs_{v_i}(\beta b) =
      \obs_{v_i}(\beta)$, and the result holds by the inductive
      hypothesis. Let $\tff{v_i}(\beta b) = \tff{v_i}(\beta) (g(\beta
      b), o)$. Since $\tff{v_i}(\alpha a)= \tff{v_i}(\beta b)$, we
      have $o = \epsilon$, and $g(\beta b) = f(\alpha a)$, and so by
      the assumption on architecture \note{refer to this?}, $b=a$.
      Since the observations of $v_i$ are simply the sequence of its
      action and the states of the bulletin board,
      $\obs_{v_i}(\run{s_0}{\alpha a})= \obs_{v_i}(\run{s_0}{\beta
        b})$ as required.
    \end{enumerate}

  \end{enumerate}
\end{proof}
\fi %
\fi %

\section{Architectural refinement}\label{sec:refine}

During the process of system development, architectural designs are
often refined by decomposing components into
sub-components. We investigate the preservation of
information security properties with respect to architectural
refinement, including refinement of extended architectures.

\Citet{arch-refinement} defines architectural refinement for
information-flow architectures as follows: 
architecture 
$\archvar_1=(D_1, \flowsto_1)$ is a refinement of architecture
$\archvar_2=(D_2, \flowsto_2)$ if there is a function $r:D_1
\rightarrow D_2$ such that $r$ is onto $D_2$ and for all $u,v\in D_1$,
if ${u}\flowsto_1{v}$ then 
${r(u)}\flowsto_2{r(v)}$.  
We write
$\archvar_1 \refines_r \archvar_2$ if $\archvar_1$ refines
$\archvar_2$ via refinement function $r$.

Intuitively, if $\archvar_1$ refines $\archvar_2$, then $\archvar_1$
provides more detail than $\archvar_2$, as it is a finer-grain
specification of security domains and the information flows between
them.  
Refinement function $r$ indicates how domains in $D_2$ are
decomposed into subdomains in $D_1$: for all $u \in D_1$, $u$ is a
subdomain of $r(u)\in D_2$.  
(Also, we say that $r(u)$ is the superdomain of domain $u$.)
The refinement function $r$ ensures that
information flows between subdomains in $\archvar_1$ are in accordance
with information flows between domains in $\archvar_2$.

\begin{figure}
\centerline{\includegraphics[height=6cm]{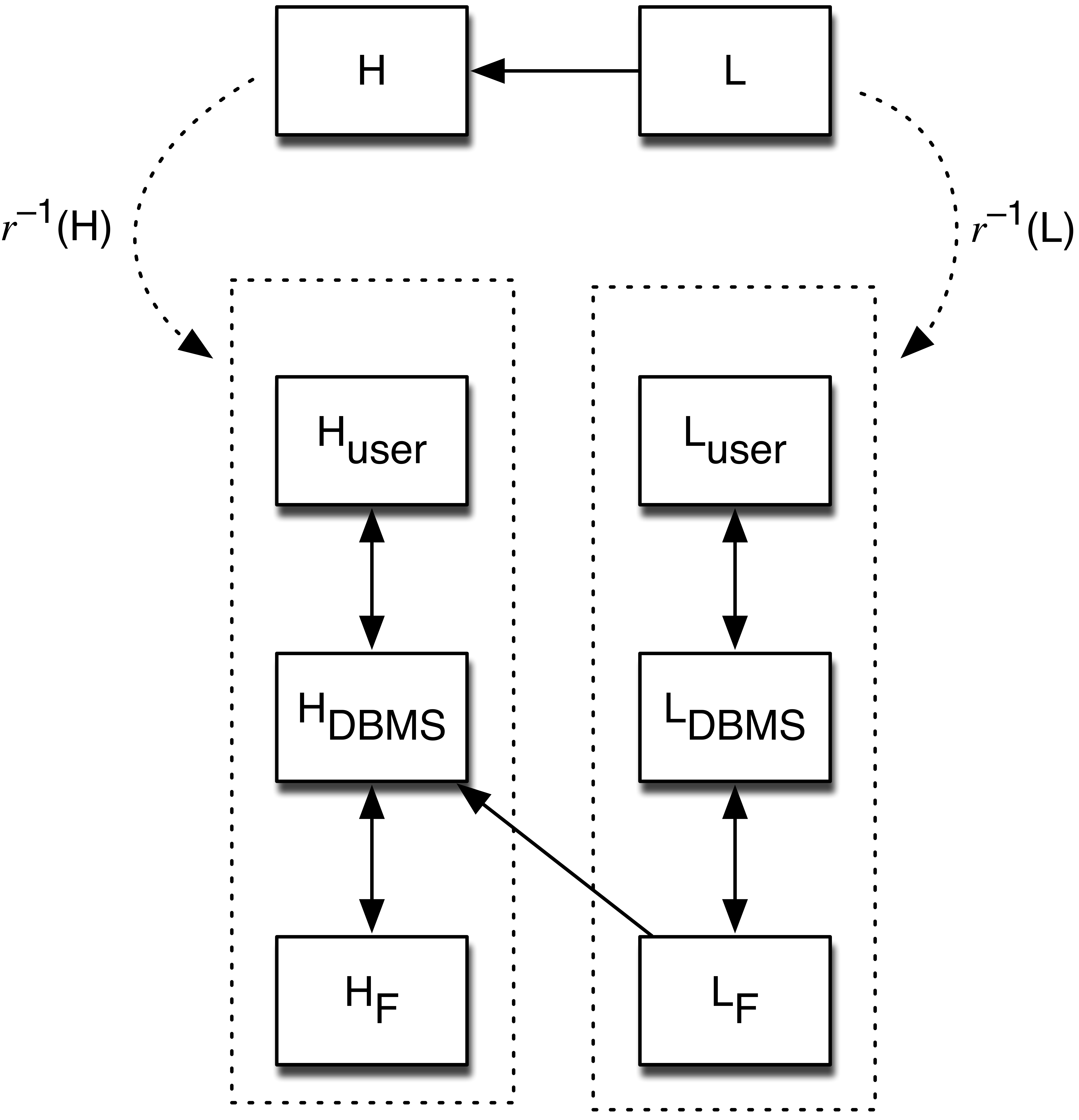}}
  \caption{The Hinke-Schaefer architecture {\hsarch} refines architecture
{\hlarch}. The dashed lines and rectangles indicate how the refinement
function maps domains in {\hsarch} to the $H$ and $L$ domains of {\hlarch}.}
  \label{fig:hs-refines-hl}
\end{figure} 

\begin{example} 
The Hinke-Schaefer architecture {\hsarch} refines architecture
{\hlarch} via refinement function $r$ that maps $\Huser$,
$\Hdbms$, and $H_F$ to $H$, and maps $\Luser$, $\Ldbms$, and
$L_F$ to $L$. All information flows within $\hsarch$ are permitted by
$\hlarch$, including the flow from $L_{F}$ to $\Hdbms$, which is
permitted as $\flowstoin{L}{H}$. Figure~\ref{fig:hs-refines-hl} show the
architectural refinement graphically.
\end{example} 

It is shown by \citet{arch-refinement}  that TA-compliance is preserved by
refinement. That is, if machine 
$\sysvar =\langle S, s_0, A,$ $D_1,$ $\step,$ $\obs, \dom\rangle$ 
is TA-compliant with  architecture
$\archvar_1$, and $\archvar_1 \refines_r \archvar_2$, then
$r(\sysvar)$ is TA-compliant with  $\archvar_2$.  
(Recall that $r(\sysvar)$ is the system $\sysvar$ with the domains abstracted 
via the function $r$, see Section~\ref{sec:infosec}.)

\subsection{Architectural refinement for interpreted extended architectures}
We
generalize 
architectural refinement to 
extended architectures.
Since the meaning of an extended architecture relies upon the interpretation of 
filter functions provided in an architectural interpretation, the generalization is
stated semantically, at the level of interpreted extended architectures. 
However, we later develop some more syntactic sufficient conditions, which 
may be more convenient for giving proofs that this semantic refinement relation holds. 

Let $\intarch_1= (\archvar_1, \archint_1)$ and $\intarch_2= (\archvar_2, \archint_2)$ 
be  interpreted extended architectures, where 
$\archvar_i=(D_i, \flowsto_i)$, $\archint_i= (A, \dom_i, \intn_i)$ for $i=1,2$.
(We assume that the  interpretations have the same set of actions $A$---we
do not attempt to deal with action refinement.)

\begin{definition} \label{def:sem-refinement-mapping}
We say that a mapping 
$r: D_1\rightarrow D_2$ 
is a {\em semantic refinement mapping} from 
$\intarch_1$ to $\intarch_2$, and write $\intarch_1 \srefines_r \intarch_2$, 
if $r$ is onto $D_2$, 
we have 
$\dom_2 = r\compose \dom_1$, and for all $u\in D_1$ 
and sequences $\alpha, \beta \in A^*$, 
if $\tff{r(u)}^{\intarch_2}(\alpha) = \tff{r(u)}^{\intarch_2}(\beta)$, then 
$\tff{u}^{\intarch_1}(\alpha) = \tff{u}^{\intarch_1}(\beta)$. 
\end{definition}  

This definition directly implies preservation of \TFF-compliance: 

\begin{theorem} \label{thm:sref} 
If $\intarch_1 \srefines_r \intarch_2$, then for all machines $\sysvar$, if 
$\sysvar$ is \TFF-compliant with $\intarch_1$, then $r(\sysvar)$ is \TFF-compliant with 
$\intarch_2$. 
\end{theorem} 

\begin{proof} 
Suppose $\intarch_1 \srefines_r \intarch_2$. 
Then for all $u_2 \in D_2$, if we have
  $\tff{u_2}(\alpha) = \tff{u_2}(\alpha')$, then $\tff{v_1}(\alpha)
  = \tff{v_1}(\alpha')$ for all $v_1 \in r^{-1}(u_2) \subseteq
  D_1$. Since $\sysvar$ is \TFF-compliant with $\intarch_1$, this will imply that
  $\obs_{v_1}(\alpha) = \obs_{v_1}(\alpha')$ for all $v_1 \in
  r^{-1}(u_2)$, which by definition of $\obs'$ will imply that
 $\obs'_{u_2}(\alpha) = \obs'_{u_2}(\alpha')$, 
thus showing that
  $r(\sysvar)$ is \TFF-compliant with respect to $\intarch_2$. 
\end{proof} 

\subsubsection{Sufficient conditions for semantic refinement}

While semantic refinement therefore has the key property 
that we  want from a notion of architectural refinement, 
establishing semantic refinement requires reasoning about 
the functions $\tff{u}$, which have a complicated 
inductive definition that, for a given domain $u$,  
potentially ranges over the  entire set of domains. 
It would be advantageous if refinements could be proved in a more local way. 
We therefore develop a number of sufficient conditions for semantic refinement 
that have this property. 

It proves to be helpful to restrict attention to interpreted architectures satisfying 
a stronger version of non-conflation than that introduced above. Say that an interpreted architecture 
is {\em strongly non-conflating}, if for all $u,v,w \in D$ such that 
 $\flowstoin[f]{u}{w}$ with $f \neq \top$ and $\flowstoin[\top]{v}{w}$, 
 we have that for all sequences $\alpha,\beta \in A^*$ 
 and actions $a,b\in A$ with $\dom(a) =u$ and  $\dom(b) = v$ that 
 $\intn(f)(\alpha, a) \neq (\tff{v}(\beta),b)$. 
Intuitively, this  says that $w$ is able to distinguish between information received from $u$ and $v$, 
or alternately, that domain $v$, which is unconstrained in the information it is able to send to 
$w$, is moreover able to authenticate its messages to $w$. Note that the special case where 
$v=w$ gives that a strongly non-conflating architecture is also non-conflating.

As a first step towards a local definition of refinement, 
we note that it suffices to focus on the information that is 
``transmitted" when an action occurs, as represented by the functions 
$\inff{u,v}^{\intarch}$. 
Let $\intarch_1= (\archvar_1, \archint_1)$ and $\intarch_2= (\archvar_2, \archint_2)$ 
be fully filtered  interpreted extended architectures, where 
$\archvar_i=(D_i, \flowsto_i)$, and $\archint_i= (A, \dom_i, \intn_i)$ for $i=1,2$.
(As above, we require  that the  interpretations have the same set of actions $A$.)
Formally, define a refinement mapping from $\intarch_1$ to $\intarch_2$
to be a function 
$r:D_1 \rightarrow D_2$ such that the following conditions hold.
\begin{enumerate}[T1]
\item \label{tref-cond-1} 
The function $r$ is onto $D_2$, and $\dom_2 = r\compose \dom_1$. 

\item \label{tref-cond-2} For all $u,v\in D_1$,
$\alpha \in A^*$ and $a \in A$ with 
$\dom_1(a) = u$,  if $\inff{r(u),r(v)} ^{\intarch_2}(\alpha, a) = \epsilon$, then $\inff{u,v}^{\intarch_1}(\alpha, a) =
  \epsilon$.

\item \label{tref-cond-3} 
 For all $u\in D_2$ and $v\in D_1$, 
  $\alpha, \beta \in A^*$ 
and $a,b \in A$,  
 if 
 $$\inff{\dom_2(a),r(v)} ^{\intarch_2}(\alpha, a)= \inff{\dom_2(b),r(v)} ^{\intarch_2}(\beta, b) \neq \epsilon$$
 then 
 either (i) $a=b$  and $\flowstoin[\top]{\dom_2(a)}{r(v)}$ 
 or (ii) 
 $\inff{\dom_1(a),v}^{\intarch_1}(\alpha, a) =\inff{\dom_1(b),v}^{\intarch_1}(\beta, b)$. 

\end{enumerate} 
We write $\intarch_1 \trefines_r \intarch_2$ when these conditions are satisfied. 

Intuitively, condition T\ref{tref-cond-2} says that in situations where no information is 
permitted to flow, by the abstract architecture $\intarch_2$, 
between superdomains of the concrete domains $u,v$, 
no information is permitted to flow between the domains $u,v$. 
It can be seen to be a generalization of the refinement condition 
(for architectures without filter functions) 
considered by  \citet{arch-refinement}. Stated in the contrapositive, 
T\ref{tref-cond-2} says that  if $\inff{u,v}^{\intarch_1}(\alpha, a) \neq \epsilon$
then 
$\inff{r(u),r(v)} ^{\intarch_2}(\alpha, a) \neq \epsilon$. 
In particular, in the case where $\flowstoin[\top]{u}{v}$, we have 
 $\inff{u,v}^{\intarch_1}(\alpha, a) = (\tff{u}(\alpha),a) \neq \epsilon$, 
 so $\inff{r(u),r(v)} ^{\intarch_2}(\alpha, a) \neq \epsilon$. This requires that there
 exists an edge $\flowstoin[g]{r(u)}{r(v)}$.

Condition T\ref{tref-cond-3}  is somewhat technical, but intuitively states that if 
it is possible for a superdomain $r(v)$ to detect the occurrence of both 
action $a$ performed after $\alpha$,  and action $b$ performed after $\beta$, 
but is not able to distinguish these actions, then it should also not be possible for the
subdomain $v$ to distinguish these actions, as stated in case (ii) of the condition. 
We separate out case (i) for technical reasons.  It covers one situation, easily checked by 
inspection of the architecture, where it can be shown that $v$ cannot distinguish the two actions
(we prove this in the context of Theorem~\ref{thm:tref} below).
  
We now show that this more localized notion of refinement provides a sufficient condition 
for semantic refinement. In what follows, to lighten the notation, 
we drop the superscripts $\intarch_1$ and $\intarch_2$
from terms like $\tff{u}^{\intarch_1}(\alpha)$ and $\inff{\dom_2(a),r(u)}^{\intarch_2}(\alpha,a)$, since 
it will always remain clear from the domains (here $u\in D_1$ or $\dom_2(a), r(u)\in D_2$)
which architecture is intended.  

\begin{theorem} \label{thm:tref}
Suppose that $\intarch_1$ and $\intarch_2$ are strongly non-conflating interpreted architectures. 
If $\intarch_1 \trefines_r \intarch_2$ then  $\intarch_1 \srefines_r \intarch_2$. 
\end{theorem} 

\begin{proof} 
Suppose $\intarch_1 \trefines_r \intarch_2$. 
We need to show that for all domains $u\in D_1$ and sequences $\alpha,\beta\in A^*$, 
if $\tff{r(u)}(\alpha) =  \tff{r(u)}(\beta)$  
then $\tff{u}(\alpha) =  \tff{u}(\beta)$. 
We show this by induction on the combined  length of $\alpha$ and $\beta$. 
The case of $\alpha = \beta = \epsilon$ is trivial. 
Consider sequences $\alpha a$ and $\beta$, such that 
$\tff{r(u)}(\alpha a) =  \tff{r(u)}(\beta)$, 
where the claim holds for shorter sequences. 
We consider two cases, depending on whether $\inff{\dom_2(a),r(u)}(\alpha,a) = \epsilon$. 

If $\inff{\dom_2(a),r(u)}(\alpha,a) = \epsilon$, then by condition~T\ref{tref-cond-2}, 
we have $\inff{\dom_1(a),u}(\alpha,a) = \epsilon$.
In this case, we also have, by definition, that 
$$\tff{r(u)}(\alpha ) =
 \tff{r(u)}(\alpha )\concat \inff{\dom_2(a),r(u)}(\alpha,a)= 
  \tff{r(u)}(\alpha a) =  \tff{r(u)}(\beta)\mathpunct.$$ 
 Thus, by the induction hypothesis, we have   
$\tff{u}(\alpha) =  \tff{u}(\beta)$. 
Since $\inff{\dom_1(a),u}(\alpha,a) = \epsilon$, 
we obtain that 
$$\tff{u}(\alpha a) =
 \tff{u}(\alpha )\concat \inff{\dom_1(a),u}(\alpha,a)= 
  \tff{u}(\alpha) =  \tff{u}(\beta)\mathpunct,$$
as required. 

Alternately, if $\inff{\dom_2(a),r(u)}(\alpha,a) \neq \epsilon$,
then since 
$$\tff{r(u)}(\beta ) =\tff{r(u)}(\alpha a) =
 \tff{r(u)}(\alpha )\concat \inff{\dom_2(a),r(u)}(\alpha,a)\mathpunct, $$
 it follows that $\beta$ is not $\epsilon$. We may therefore write $\beta = \beta' b$, where $b\in A$. 
 We now have two further cases. If $\inff{\dom_2(b),r(u)}(\beta',b) = \epsilon$, then 
 we may apply the argument above with the roles of $\alpha a$ and $\beta' b$ switched. 
 Otherwise,  $\inff{\dom_2(b),r(u)}(\beta',b) \neq \epsilon$.  
 Since $\tff{r(u)}(\alpha a) =\tff{r(u)}(\beta' b)$
 states that 
$$
 \tff{r(u)}(\alpha )\concat \inff{\dom_2(a),r(u)}(\alpha,a) 
 =  
 \tff{r(u)}(\beta')\concat \inff{\dom_2(b),r(u)}(\beta',b)$$
and neither of the appended elements is $\epsilon$, it follows  that 
$ \tff{r(u)}(\alpha ) =\tff{r(u)}(\beta')$
 and 
 $\inff{\dom_2(a),r(u)}(\alpha,a) =  \inff{\dom_2(b),r(u)}(\beta',b)$. 
 By the induction hypothesis, we obtain that 
 $ \tff{u}(\alpha ) =\tff{u}(\beta')$. Also, by 
  condition T\ref{tref-cond-3} we  have 
  either (i) $a=b$  and $\flowstoin[\top]{\dom_2(a)}{r(u)}$ 
 or (ii) 
 $\inff{\dom_1(a),u}(\alpha,a) =  \inff{\dom_1(b),u}(\beta',b)$. 
We claim that (ii) also holds in case (i).  
It then  follows, by definition, that $\tff{u}(\alpha a) =\tff{u}(\beta' b)$, 
as required.  

To prove the claim, note that in case (i), we have 
 $$\inff{\dom_2(a),r(u)}(\alpha,a) =  \inff{\dom_2(b),r(u)}(\beta',b) 
 =  \inff{\dom_2(a),r(u)}(\beta',a) 
 $$ and since we have an edge $\flowstoin[\top]{\dom_2(a)}{r(u)}$ this  states 
 that 
 $(\tff{\dom_2(a)}(\alpha),a) =  (\tff{\dom_2(a)}(\beta'),a)$.
 In particular, we have $\tff{\dom_2(a)}(\alpha) =  \tff{\dom_2(a)}(\beta')$.
  By induction,  
 it follows that  $\tff{\dom_1(a)}(\alpha)=  \tff{\dom_1(a)}(\beta')$.  
 There are now several possibilities, depending on the nature of the edge, if any, 
 from $\dom_1(a)= \dom_1(b)$ to $u$ in $\intarch_1$. 
\begin{enumerate} 
 \item if $\notflowstoin{\dom_1(a)}{u}$, then 
 by definition, $ \inff{\dom_1(a),u}(\alpha,a)= \epsilon = 
 \inff{\dom_1(b),u}(\beta',b)$. 
 \item If $\flowstoin[\top]{\dom_1(a)}{u}$, 
 then $ \inff{\dom_1(a),u}(\alpha,a)= 
 (\tff{\dom_1(a)}(\alpha),a) = 
  (\tff{\dom_1(b)}(\beta'),b) = 
 \inff{\dom_1(b),u}(\beta',b)$. 

 \item If $\flowstoin[f_1]{\dom_1(a)}{u}$ with $f_1\neq \top$, then 
 it follows from $\tff{\dom_1(a)}(\alpha)=  \tff{\dom_1(a)}(\beta')$,
 using compatibility of $\intn_1(f_1)$, 
 that
$ \inff{\dom_1(a),u}(\alpha,a)= 
\intn_1(f_1)(\alpha, a) = \intn_1(f_1)(\beta', a) = \intn_1(f_1)(\beta', b) = 
 \inff{\dom_1(b),u}(\beta',b)$. 
 \end{enumerate} 
 Thus, 
 in all cases, we have 
 $\inff{\dom_1(a),u}(\alpha,a) =  \inff{\dom_1(b),u}(\beta',b)\mathpunct.$
\end{proof} 

The refinement relation $\trefines$ is quite general and captures the essence of 
the relation between semantic architectural refinement and refinement of the information that 
flows over the edges of an architecture. However, to establish that the conditions hold 
may still require reasoning about the complex ``global" functions $\tff{u}$. 
For example, in the case where an edge $\flowstoin[f]{r(u)}{r(v)}$ with $f\neq \top$ is 
refined by an edge  $\flowstoin[\top]{u}{v}$, we need to prove a relationship 
between the ``local'' function $\intn_2(f)$,  and the recursive
 function $\tff{u}$. We therefore develop a more specific 
set of conditions that do support a more local approach to proof of a refinement, 
in which we need to  consider only the edges not labelled by $\top$. 

Define a {\em strict refinement mapping} from $\intarch_1$ to $\intarch_2$
to be a function 
$r:D_1 \rightarrow D_2$ such that the following conditions hold.
\begin{enumerate}[S1]
\item \label{stref-cond-1} 
The function $r$ is onto $D_2$, and $\dom_2 = r\compose \dom_1$. 

\item \label{stref-cond-2}
For $u,v \in D_1$, if $\flowstoin[\top]{u}{v}$ then $\flowstoin[\top]{r(u)}{r(v)}$.

\item \label{stref-cond-3} 
 For all $u,v\in D_1$, if $\flowstoin[f_1]{u}{v}$ with $f_1\neq \top$ 
 then there exists an edge  $\flowstoin[f_2]{r(u)}{r(v)}$ with 
 either $f_2 = \top$  or $f_2 \neq \top$  and 
 for all $\alpha\in A^*$ and $a\in A$ with $\dom_1(a) = u$
 we have $\intn_2(f_2)(\alpha, a) = \epsilon$ implies $\intn_1(f_1)(\alpha, a) = \epsilon$. 
 
 \item \label{stref-cond-4} 
 For all $u\in D_1$ and actions $a,b\in A$, if 
 $\flowstoin[f_2]{\dom_2(a)}{r(u)}$ and   $\flowstoin[g_2]{\dom_2(b)}{r(u)}$
 where $f_2 \neq \top$ and $g_2 \neq \top$, 
 then either 
 \begin{enumerate} 
 \item 
 $\notflowstoin{\dom_1(a)}{u}$ and   $\notflowstoin{\dom_1(b)}{u}$, 
 or
 \item 
 there exists an edge 
 $\flowstoin[f_1]{\dom_1(a)}{u}$ with $f_1 \neq \top$ but  $\notflowstoin{\dom_1(b)}{u}$,
 and for all $\alpha, \beta \in A^*$ such that 
 $\intn_2(f_2)(\alpha, a) =\intn_2(g_2)(\beta, b)\neq \epsilon $ 
 we have  $\intn_1(f_1)(\alpha, a ) = \epsilon$, or 
 
 \item 
  $\notflowstoin{\dom_1(a)}{u}$ but 
 there exists an edge  $\flowstoin[g_1]{\dom_1(b)}{u}$ with $g_1\neq \top$, 
 and for all $\alpha, \beta \in A^*$  such that 
 $\intn_2(f_2)(\alpha, a) =\intn_2(g_2)(\beta, b)\neq \epsilon$ 
 we have  $\intn_1(g_1)(\beta, b) = \epsilon$, or

 \item 
  there exist edges 
 $\flowstoin[f_1]{\dom_1(a)}{v}$ and $\flowstoin[g_1]{\dom_2(b)}{v}$
with $f_1\neq \top$ and $g_1 \neq \top$, and 
 for all $\alpha, \beta \in A^*$, such that 
$\intn_2(f_2)(\alpha, a) =\intn_2(g_2)(\beta, b)\neq \epsilon$ 
 we have 
$\intn_1(f_1)(\alpha, a)   = \intn_1(g_1)(\beta, b)$. 
\end{enumerate} 
\end{enumerate}

Note that in the case where all edges in the two architectures are labelled by $\top$, 
conditions S\ref{stref-cond-1}-\ref{stref-cond-2} amount to the notion of architectural refinement of 
\citet{arch-refinement}. For edges labelled with non-$\top$ labels, 
Condition S\ref{stref-cond-3} states that whenever information is not permitted to flow 
between two superdomains, it is not permitted to flow between their subdomains. 

Condition S\ref{stref-cond-4}, like condition T\ref{tref-cond-3}, is intended to capture 
that if a superdomain cannot distinguish two actions $a$ and $b$, then neither can its subdomains. 
However, S\ref{stref-cond-4} restricts the statement of this property to filtered edges, 
i.e., edges not labelled $\top$. 
The essence of the property is most easily visible in condition S\ref{stref-cond-4}(d), which 
corresponds to a situation  where both $\dom_1(a)$ and $\dom_1(b)$ have 
filtered edges to a subdomain. The remaining cases deal with all the other possible configurations
that are consistent with  S\ref{stref-cond-2}. 

Note that cases (b) and (c) of S\ref{stref-cond-4} say that where a superdomain $r(u)$ is 
not able to distinguish the actions $a$ and $b$,  it is not 
permitted for the  subdomain $u$ to receive information flow 
as a result of one action but not the other.
Intuitively, this would imply that the 
subdomain can deduce which action occurred, giving it more 
information than its superdomain, whereas our intuition for refinement is that
it should not increase the amount of information flow. 

We remark that S\ref{stref-cond-4} does not need to consider the situation where 
$\flowstoin[f_2]{\dom_2(a)}{r(u)}$ with $f_2 \neq \top$ and   $\flowstoin[\top]{\dom_2(b)}{r(u)}$
because in strongly  non-conflating architectures, it is impossible that 
$\inff{\dom_2(a),r(u)}(\alpha, a) = \inff{\dom_2(b),r(u)}(\beta,b)$. 

\begin{theorem} \label{thm:sref-imp-tref}
If 
$\intarch_2$ is 
strongly non-conflating and 
$\intarch_1$ strictly refines  $\intarch_2$ by function $r$, 
then  $\intarch_1 \trefines_r  \intarch_2$.  
\end{theorem} 

\begin{proof} 
Assume that $\intarch_1$ and $\intarch_2$ are strongly non-conflating
and that conditions S\ref{stref-cond-1}--\ref{stref-cond-4} hold. 
We prove conditions T\ref{tref-cond-1}-\ref{tref-cond-3}. 
T\ref{tref-cond-1} is  trivially identical to S\ref{stref-cond-1}. 

For condition T\ref{tref-cond-2}, we consider that contrapositive. 
Suppose that $u,v \in D_1$ and $\alpha  \in A^*$ and $a\in A$ with 
$\dom_1(a) = u$ and $\inff{u,v}(\alpha,  a)\neq  \epsilon$. 
We need to prove that  $\inff{r(u),r(v)}(\alpha,  a) \neq \epsilon$. 
From $\inff{u,v}(\alpha,  a)\neq  \epsilon$ we obtain that 
there exists an edge $\flowstoin[f_1]{u}{v}$. 
If $f_1 = \top$, then by S\ref{stref-cond-2}, we have 
that $\flowstoin[\top]{r(u)}{r(v)}$, and since $\dom_2(a) = r(\dom_1(a)) = r(u)$
by S\ref{stref-cond-1}, this implies that $\inff{r(u),r(v)}(\alpha,  a) 
= (\tff{r(u)}(\alpha),  a) \neq \epsilon$. 
Alternately, if $f_1\neq \top$, then 
by S\ref{stref-cond-3},  
there exists an edge $\flowstoin[f_2]{r(u)}{r(v)}$, 
with either $f_2 = \top$ or 
 $f_2 \neq  \top$ and $\intn_2(f_2)(\alpha, a) = \epsilon$ implies 
$\intn_1(f_1)(\alpha, a) = \epsilon$.
In the case $f_2 = \top$ we argue exactly as above. In case 
 $f_2 \neq  \top$, since we have 
 $\intn_1(f_1)(\alpha, a) = \inff{u,v}(\alpha  a)
 \neq  \epsilon$, we obtain 
 $\epsilon \neq \intn_2(f_2)(\alpha, a) = \inff{r(u),r(v)}(\alpha , a)$, again as 
 required. 
 
 For condition  T\ref{tref-cond-3}, suppose 
 $\inff{\dom_2(a),r(v)}(\alpha,  a) = \inff{\dom_2(b),r(v)}(\beta, b) \neq \epsilon$. 
 We have to show that either (i) $a=b$ and $\flowstoin[\top]{\dom_2(a)}{r(v)}$, or (ii) 
 $\inff{\dom_1(a),v}(\alpha,  a)= \inff{\dom_1(b),v}(\beta, b)$. 
 From $\inff{\dom_2(a),r(v)}(\alpha,  a) = \inff{\dom_2(b),r(v)}(\beta, b) \neq \epsilon$
 it follows that there exist edges 
$\flowstoin[f_2]{\dom_2(a)}{r(v)}$ and $\flowstoin[g_2]{\dom_2(a)}{r(v)}$. 
We consider several cases, depending on whether these edges are labelled $\top$ or not. 

If $f_2 = g_2 = \top$, then 
$\inff{\dom_2(a),r(v)}(\alpha,  a)  = (\tff{\dom_2(a)}(\alpha),  a)$  
and $\inff{\dom_2(b),r(v)}(\beta, b)  = (\tff{\dom_2(b)}(\beta ),  b)$, and it follows that  $a=b$ and
we have (i). 

The case that $f_2 = \top$ and $g_2 \neq \top$ is not possible, by the assumption that 
the architecture $\intarch_2$ is strongly non-conflating. 

If $f_2 \neq \top$ and $ g_2 \neq  \top$, then by S\ref{stref-cond-4}, we have one of 
four possibilities. 

\begin{itemize}
\item[(a)] 
$\notflowstoin{\dom_1(a)}{v}$ and $\notflowstoin{\dom_1(b)}{v}$. 
In this case,  $\inff{\dom_1(a),v}(\alpha  a)= \epsilon = \inff{\dom_1(b),v}(\beta b)$. 

\item[(b)] 
There exists an edge $\flowstoin[f_1]{\dom_1(a)}{v}$ with $f_1\neq \top$ 
but $\notflowstoin{\dom_1(b)}{v}$, 
and $\intn_2(f_2)(\alpha, a) = \intn_2(g_2)(\beta, b)\neq \epsilon$ implies $\intn_1(f_1)(\alpha, a) = \epsilon$. 
In this case, since 
$$\intn_2(f_2)(\alpha, a) = \inff{\dom_2(a),r(v)}(\alpha, a) =
\inff{\dom_2(b),r(v)}(\beta, b)  = \intn_2(g_2)(\beta, b)\neq \epsilon\mathpunct,$$ 
we obtain that $\intn_1(f_1)(\alpha, a) = \epsilon$. 
Thus, $\inff{\dom_1(a),v}(\alpha, a) = \intn_1(\alpha, a) = \epsilon = \inff{\dom_1(b),v}(\beta, b)$. 

\item[(c)] This case is identical to case (b) with the roles of $\alpha a$ and $\beta b$ reversed. 

\item[(d)] 
There exists edges $\flowstoin[f_1]{\dom_1(a)}{v}$ and  $\flowstoin[g_1]{\dom_1(b)}{v}$
with $f_1\neq \top$ and $g_1\neq \top$ and 
$\intn_2(f_2)(\alpha, a) = \intn_2(g_2)(\beta, b)\neq \epsilon$ implies $\intn_1(f_1)(\alpha, a) = \intn_1(g_1)(\beta, b)$. 
Since  $
\intn_2(f_2)(\alpha, a) = 
\inff{\dom_2(a),r(v)}(\alpha, a) =  
\inff{\dom_2(b),r(v)}(\beta, b) = 
 \intn_2(g_2)(\beta, b) \neq \epsilon
$ 
we obtain 
$\intn_1(f_1)(\alpha, a) = \intn_1(g_1)(\beta, b)$, which is identical to 
$\inff{\dom_1(a),v}(\alpha, a) = \inff{\dom_1(b),v}(\beta, b)$. 
\end{itemize}
Thus, in each case, we have
$\inff{\dom_1(a),v}(\alpha, a) = \inff{\dom_1(b),v}(\beta, b)$, as required. 
\end{proof}

We remark that although condition~S\ref{stref-cond-4} is somewhat complex, under a reasonable 
assumption it can be replaced by the following much simpler condition~S\ref{stref-cond-5}, 
which  states more transparently that at least as much information is permitted to 
flow along abstract edges as is permitted to flow along any corresponding 
concrete edges. 

\begin{enumerate}[S1]
\setcounter{enumi}{4}
\item  \label{stref-cond-5} 
For $\alpha , \beta \in A^*$ and $a,b \in A$ and $v\in D_1$ such that $\dom_1(a) =\dom_1(b)$, 
and $\flowstoinf[f_1]{\dom_1(a)}{v}{\flowsto_1}$ and $\flowstoinf[f_2]{\dom_2(a)}{r(v)}{\flowsto_2}$, 
where $f_1\neq \top$ and $f_2\neq \top$, 
if  $\intn_2(f_2)(\alpha, a) = \intn_2(f_2)(\beta, b)\neq \epsilon$, then 
$\intn_1(f_1)(\alpha, a) =  \intn_1(f_1)(\beta, b)$. 
\end{enumerate}

Say that  {\em  messages are source-identifying with respect to $r$},  if there 
exists a function $\source$ mapping the union of the ranges of the filter functions $\intn_2(f)$ of 
$\intarch_2$ to the set of domains $D_1$  of  $\intarch_1$, such that 
for all $\alpha\in A^*$ and $a\in A$, if $\flowstoinf[f_2]{r(\dom_1(a))}{w}{\flowsto_2}$ 
and $\intn_2(f_2)(\alpha, a) \neq \epsilon$, 
then $\source(\intn_2(f_2)(\alpha, a)) = \dom_1(a)$. 
One example of when this condition can be 
met is when $\intn_2(f_2) (\alpha, a) = (g(\alpha), a)$ for some function $g$, since 
then the function $\source((x,a)) = \dom_1(a)$ obviously has the  required property. 

\begin{lemma} \label{lem:simpref} 
If messages are source-identifying with respect to $r$ 
and  $r$ satisfies conditions~S\ref{stref-cond-1},~S\ref{stref-cond-2} and ~S\ref{stref-cond-5} then 
$r$ satisfies condition~S\ref{stref-cond-4}. 
\end{lemma} 

\begin{proof} 
Suppose $\flowstoinf[f_2]{\dom_2(a)}{r(v)}{\flowsto_2}$,   
 $\flowstoinf[g_2]{\dom_2(b)}{r(v)}{\flowsto_2}$ with $f_2\neq \top$ and $g_2\neq \top$. Let $u\in D_1$. 
By condition~S\ref{stref-cond-2}, we cannot have $\flowstoin[\top]{\dom_1(a)}{u}$ or 
$\flowstoin[\top]{\dom_1(b)}{u}$. This leaves four possibilities, depending on whether there is 
an non-$\top$-labelled  edge between $\dom_1(a)$ or $\dom_1(b)$ and $u$ or not. 
In case there are no such edges, we have condition~S\ref{stref-cond-4}(a). 
We consider the three other possibilities. For each case, observe that if 
there exists $\alpha, \beta\in A^*$ with $\intn_2(f_2)(\alpha, a) = \intn_2(g_2)(\beta, b) \neq \epsilon$, 
then  $\dom_1(a) = \source(\intn(f_2)(\alpha, a)) = \source(\intn(g_2)(\beta, b)) = \dom_1(b)$, 
by the assumption that messages are source-identifying with respect to $r$. 
Conversely, if $\dom_1(a) \neq \dom_1(b)$, then $\intn_2(f_2)(\alpha, a) = \intn_2(g_2)(\beta, b) \neq \epsilon$
cannot be satisfied. 
\begin{itemize}
\item Suppose there exists an edge $\flowstoin[f_1]{\dom_1(a)}{u}$ with $f_1\neq \top$ and  
$\notflowstoin{\dom_1(b)}{u}$. In this case, we must have $\dom_1(a) \neq \dom_1(b)$, 
so by the observation above, $\intn_2(f_2)(\alpha, a) = \intn_2(g_2)(\beta, b) \neq \epsilon$
cannot be satisfied. Thus, condition S\ref{stref-cond-4}(b) is vacuously satisfied. 
\item Suppose  $\notflowstoin{\dom_1(b)}{u}$ but
there exists an edge $\flowstoin[g_1]{\dom_1(a)}{u}$ with $g_1\neq \top $. This case is similar
to the previous one, yielding S\ref{stref-cond-4}(c). 
\item Suppose there exist edges $\flowstoin[f_1]{\dom_1(a)}{u}$  and $\flowstoin[g_1]{\dom_1(a)}{u}$ with $f_1,g_1\neq \top $. If $\intn_2(f_2)(\alpha, a) = \intn_2(g_2)(\beta, b) \neq \epsilon$, then 
as observed above, we have $\dom_1(a) = \dom_1(b)$, and hence 
also $f_1 = g_1$ by the fact that there is at
most one edge between any two domains. Thus, 
by  S\ref{stref-cond-5}, we have $\intn(f_1)(\alpha, a) = \intn(f_1)(\beta, b) = 
\intn(g_1)(\beta, b)$. This establishes condition S\ref{stref-cond-4}(d). 
 \end{itemize} 
Thus, in all cases, we have proved condition S\ref{stref-cond-4}. 
\end{proof} 

\subsubsection{Properties of architectural refinement}

Architectural refinement provides a design methodology in which we may prove certain 
security properties at a high level of abstraction, and preserve the validity of those properties 
as details of the architectural design are specified. 
Since an interpreted architecture $\intarch$ specifies a set of actions $A$,
any interpretation function $\pi$ mapping atomic propositions to subsets of $A^*$ can 
be treated as an interpretation for any system $M$ that is \TFF-compliant with $\intarch$. 
For a formula $\phi$, we write $\intarch, \pi \models \phi$ if 
$M, \pi \models \phi$ for all systems $M$ that are \TFF-compliant with $\intarch$. 

Recall that $\intarch_1 \srefines_r \intarch_2$ 
entails that $\intarch_1$ and $ \intarch_2$ have the same set of actions $A$. 
Since  interpretations $\pi$ map each propositional constant to a subset of $A^*$, the
two interpreted architectures also have the same interpretations. 
The following result shows that a property that has been shown to follow from 
compliance with $\intarch_2$ can be translated to a property that follows
from compliance with  $\intarch_1$.

\begin{theorem}\label{thm:knowledge-refine}
Suppose $\intarch_1$ and $\intarch_2$ are 
interpreted extended architectures 
such that $\intarch_1  \srefines_r \intarch_2$. 
Let $\pi$ be an interpretation for these architectures. 
Suppose $\phi$ is a formula for the domains of $\intarch_2$, and $\intarch_2, \pi  \sat  \phi$. 
Then $\intarch_1, \pi  \sat  r^{-1}(\phi)$. 
\end{theorem}

\begin{proof}
Suppose that   system $\sysvar$ is 
  \TFF-compliant with   $\intarch_1$. 
  Since $\intarch_1  \srefines_r \intarch_2$, it follows by Theorem~\ref{thm:sref} that
  $r(\sysvar)$ is \TFF-compliant with $\intarch_2$. 
  Hence $r(\sysvar), \pi \sat \phi$. 
By   Theorem~\ref{thm:Kpullback} 
we obtain that $\sysvar, \pi \models r^{-1}(\phi)$. 
  \end{proof}

The level of abstraction used to model a system may affect the success
or efficiency of a proof of a given global security property.
Theorem~\ref{thm:knowledge-refine} can facilitate proofs of security
properties that are preserved under refinement: simpler, more
abstract, architectures can be used to reason about the security
property, and preservation under refinement ensures that the security
property will hold of systems satisfying a more refined architecture.
We demonstrate this in Sections \ref{sec:hink-schaefer-refinement} and
\ref{sec:downgrader-refine} below, using more abstract architectures
to prove security results about the Hinke-Schaefer architecture and a
refinement of the downgrader architecture.

\subsection{Specification Refinement} 

We can also define refinement at the level of architectural specifications. 
Suppose that $\archvar_1=(D_1, \flowsto_1)$ and 
$\archvar_2=(D_2, \flowsto_2)$ are architectures, and  $r$ is a function mapping $D_1$ onto $D_2$. 
If $\archspec_1$ and $\archspec_2$ are architectural specifications for 
$\archvar_1$ and $\archvar_2$, respectively, then we write  
$(\archvar_1,\archspec_1) \srefines_r (\archvar_2, \archspec_2)$
when for all interpretations $\archint_1\in \archspec_1$ there exists an 
interpretation $\archint_2\in \archspec_2$ such that 
$(\archvar_1, \archint_1) \srefines_r (\archvar_2, \archint_2)$.

Architectural specifications may allow multiple architectural interpretations, 
and these may allow the set of actions $A$ in the systems being
specified to vary. In order to interpret the logic of knowledge, we need
an interpretation $\pi$ that maps each propositional constant to a subset of $A^*$. 
When $A$ varies, we need $\pi$ to vary correspondingly. To ensure proper coordination 
between interpretations $\pi$ and sets of actions $A$, we work with  sets ${\cal M}$ of 
{\em interpreted systems}, 
i.e., pairs $(M,\pi)$ where $M$ is a system and $\pi$ is an interpretation.  
Such a set ${\cal M}$ may express weak conditions such as 
``$\pi(p)$ is a 
$G$-dependent
proposition". We write 
${\cal M},  (\archvar, \archspec) \models \phi$ when $M, \pi \models \phi$
for all $(M,\pi) \in {\cal M}$ such that $M$ \TFF-complies with $(\archvar, \archspec)$. 
We also write 
$r({\cal M})$ for $\{(r(M), \pi)~|~(M,\pi)\in {\cal M}\}$.

\begin{theorem} 
Let ${\cal M}_1, {\cal M}_2$ be sets of interpreted systems for architectures $\archvar_1$, $\archvar_2$, respectively, 
and let  $\phi$ be a formula for the domains of $\archvar_2$. 
Suppose that $(\archvar_1,\archspec_1) \srefines_r (\archvar_2, \archspec_2)$ and 
$r({\cal M}_1) \subseteq {\cal M}_2$ and ${\cal M}_2,(\archvar_2, \archspec_2) \models \phi$. 
Then  we have ${\cal M}_1,(\archvar_1, \archspec_1) \models r^{-1}(\phi)$. 
\end{theorem} 

\begin{proof} 
Let $(M,\pi) \in {\cal M}_1$ and suppose $M$ is  \TFF-compliant with $(\archvar_1, \archspec_1)$. 
We need to show that $M, \pi \models r^{-1}(\phi)$. 
By definition, there exists 
an interpretation $\archint_1$ such that $M$ is \TFF-compliant with $(\archvar_1, \archint_1)$. 
Since $(\archvar_1,\archspec_1) \srefines_r (\archvar_2, \archspec_2)$, it follows that there exists an 
interpretation $\archint_2$ such that $(\archvar_1,\archint_1) \srefines_r (\archvar_2, \archint_2)$. 
By Theorem~\ref{thm:sref}, it follows that $r(M)$ is \TFF-compliant with 
$(\archvar_2, \archint_2)$, and hence also with $(\archvar_2, \archspec_2)$. 
Moreover, since $r({\cal M}_1)\subseteq {\cal M}_2$, we have $(r(M) ,\pi ) \in {\cal M}_2$. 
Thus $r(M), \pi \models \phi$. 
By Theorem~\ref{thm:Kpullback} we obtain that $M, \pi \models r^{-1}(\phi)$, as required. 
\end{proof} 

We may similarly define $(\archvar_1,\archspec_1) \refines_r (\archvar_2, \archspec_2)$
to hold  when for all interpretations $\archint_1\in \archspec_1$ there exists an 
interpretation $\archint_2\in \archspec_2$ such that 
$(\archvar_1, \archint_1) \refines_r (\archvar_2, \archint_2)$. 
Additionally, we may define $(\archvar_1,\archspec_1)$ to be a strict refinement of $(\archvar_2, \archspec_2)$ by function $r$ 
when for all interpretations $\archint_1\in \archspec_1$ there exists an 
interpretation $\archint_2\in \archspec_2$ such that 
$(\archvar_1, \archint_1)$ is a strict refinement of $(\archvar_2, \archint_2)$.

\begin{theorem} 
If $(\archvar_1,\archspec_1)$ is a strict refinement of $(\archvar_2, \archspec_2)$ by function $r$ 
then $(\archvar_1,\archspec_1) \trefines_r (\archvar_2, \archspec_2)$, and if $(\archvar_1,\archspec_1) \trefines_r (\archvar_2, \archspec_2)$
then $(\archvar_1,\archspec_1) \srefines_r (\archvar_2, \archspec_2)$.
\end{theorem} 

\begin{proof} 
Straightforward using Theorem~\ref{thm:tref} and Theorem~\ref{thm:sref-imp-tref}. 
\end{proof}

\subsection{Example: Hinke-Schaefer}\label{sec:hink-schaefer-refinement}

Since the Hinke-Schaefer architecture $\hsarch$ refines architecture
$\hlarch$, we can apply the information security result for $\hlarch$,
Theorem~\ref{thm:hl-h-action-local}, to $\hsarch$: since domain $L$
never 
knows 
any $H$-dependent 
proposition, the domains $\Luser$,
$\Ldbms$, and $L_{F}$ never 
know 
any $u$-dependent 
proposition,
for $u \in \{\Huser, \Hdbms,H_{F}\}$. This information security
property was stated as Theorem~\ref{thm:hinke-schaefer}, in
Section~\ref{sec:hinke-schaefer}. We give a simple proof for it here.

\begin{proofof}{Theorem~\ref{thm:hinke-schaefer}}
  Follows easily from Theorem~\ref{thm:hl-h-action-local},
  Lemma~\ref{lem:prop-refine}, and Theorem~\ref{thm:knowledge-refine},
  since $\hsarch \refines_r \hlarch$ for refinement function $r$ such
  that $r(\Luser) =r(\Ldbms) =r(L_{F}) = L$ and $r(\Huser)
  =r(\Hdbms) =r(H_{F}) = H$.
\end{proofof}

Thus, we were able to prove an information security property about
$\hsarch$ by proving an appropriate policy in the much simpler
architecture $\hlarch$.

\subsection{Example: Downgrader}\label{sec:downgrader-refine}

\newcommand{\downarchr}{\ensuremath{\mathcal{DGR}}}

In the architecture $\downarch$ of Section~\ref{sec:downgrader}, 
there is a filter function on the edge from downgrader $D$ to Low security domain $L$, 
specifying that $D$ should not release to $L$ any information about $C$.  
This makes $D$ a trusted component in the system: in any implementation, we would need 
to verify that $D$ correctly enforces this information flow constraint. 
However, information about $C$ may become co-mingled with information about $P$
in the data store $H$, so it is not immediately clear how $D$ could, on its own, 
guarantee enforcement of the constraint. Thus, it seems that the architecture 
implies constraints on other components. One approach that could be pursued to 
implement this architecture is to ensure that the data store $H$ maintains secure provenance
information, which $D$ can use to check that information being released is not tainted with 
information from $C$. In this section, we pursue another approach, which is to move the 
trust boundary in such a way that $D$ is prevented from obtaining information about $C$. 
We develop an architecture $\downarchr$ that has this property, and show it to be a 
refinement of $\downarch$. It follows that any system compliant with $\downarchr$ 
is also complaint with $\downarch$. Since it is much clearer how architecture 
$\downarchr$ could be implemented with just  local verification of its trusted components, 
this moves us closer to a practical implementation of $\downarch$. 

\newcommand{\fT}{f}

The following diagram shows architecture $\downarchr$, 
with grouping of its domains (indicated by dashed rectangles) 
indicating a refinement mapping $r$ to the architecture $\downarch$. 
In $\downarchr$, domain $H$ of $\downarch$ is decomposed into 
two domains $H_C$ and $H_P$ corresponding respectively to 
data stores of information about $C$ and $P$. Domain $P$ is 
decomposed into two domains $T$ and $U$, corresponding to 
trusted and untrusted users within this domain. Thus, the 
mapping $r$ from the domains of $\downarchr$ to the domains
of $\downarch$ is given by 
$r(H_P) = r(H_C) = H$, $f(T) = r(U) = P$ and $r(u) = u$ for $u\in \{C,D,L\}$.

\centerline{\includegraphics[height=6cm]{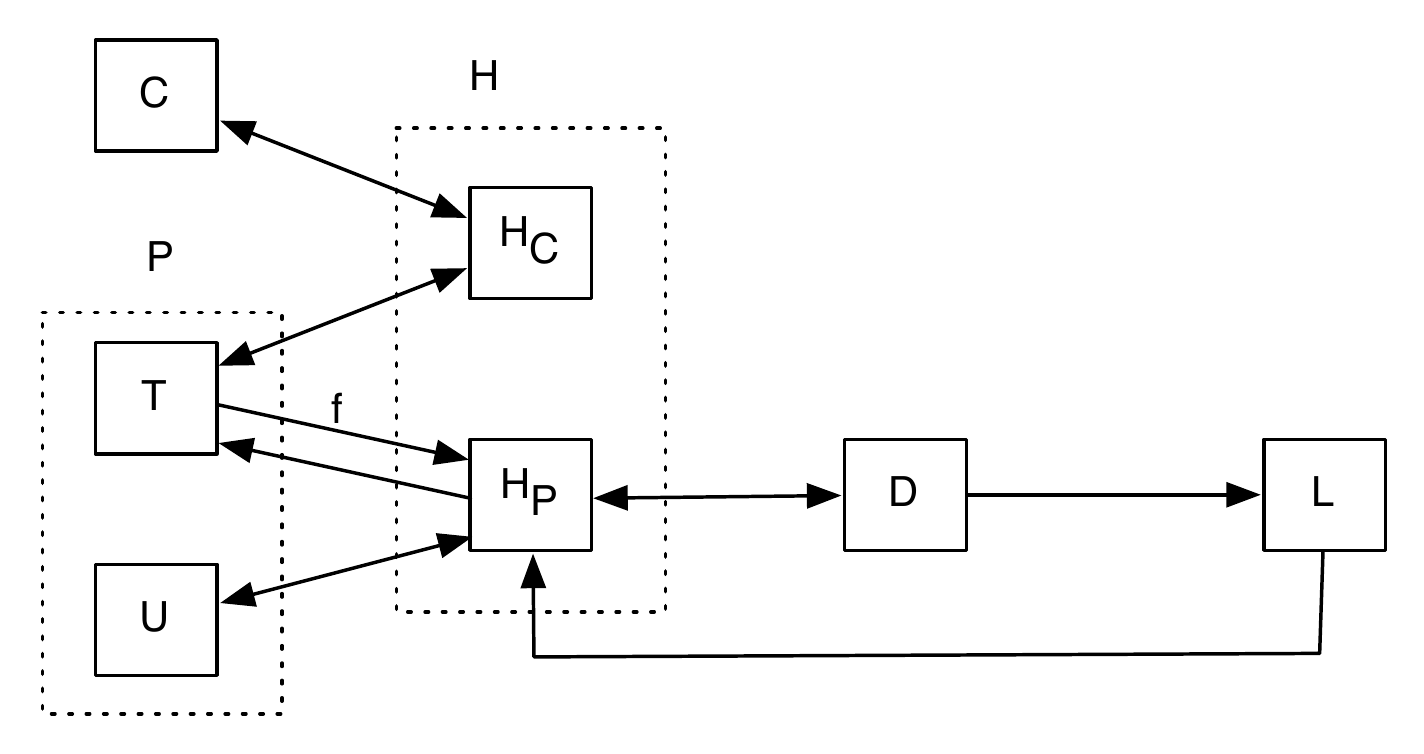}}

In $\downarchr$,  $T$ is a trusted component, since this domain is required to 
enforce an information flow constraint, represented by the 
edge $\flowstoin[\fT]{T}{H_P}$. On the other hand, note that 
whereas in $\downarch$, domain $D$ is trusted, 
because of the edge $\flowstoin[\downrelf]{D}{L}$, in 
$\downarchr$ the edge from $D$ to $L$ is labelled $\top$, 
so $D$ is no longer a trusted component in $\downarchr$. 

We define architectural specification $\archspec_{\downarchr}$ so that $(A, \dom,
\intn) \in \archspec_{\downarchr}$ if and only if $(A, \dom, \intn)$ is
an interpretation of $\downarchr$ such that $\intn(\fT)$ is
defined by $\intn(\fT)(\alpha,a ) = a$ for all $\alpha\in A^*$ and $a\in A$ with $\dom(a) = T$, 
and $\intn(\fT)(\alpha,a ) = \epsilon$ otherwise. 
Intuitively, this states that $\fT$ permits information about $T$
actions to flow from $T$ to $H_P$, but nothing more. It is trivial to check that this 
interpretation is $\tff{T}$-compatible.

\begin{theorem} \label{thm:downr-refines-down}
$(\downarchr,\archspec_{\downarchr}) \srefines_r (\downarch,\archspec_{\downarch})$ 
\end{theorem} 

\begin{proof} 
We need to work with three distinct architectures in this proof. 
To distinguish them, we use $\flowsto_1$ to refer to the information flow policy of $\downarchr$, 
 $\flowsto_2$ to refer to the information flow policy of $\downarch$, 
 and $\flowsto'$ to refer to the policy used in the definition of 
 $\intn(\downrelf)$ (see Section~\ref{sec:downgrader}).  
Let $\intarch_1 = (\downarchr, \archint_1)$ where 
$\archint_1 = (A,\dom_1,\intn_1) \in \archspec_{\downarchr}$. 
We need to show that there exists $\archint_2  \in \archspec_{\downarchr}$
with $\intarch_1 \srefines \intarch_2$. 
In fact, there exists a unique architectural interpretation $\archint_2  = (A,\dom_2, \intn_2)\in \archspec_{\downarchr}$ with the same set $A$ of actions and $\dom_2 = r\circ \dom_1$, 
so let $\archint_2$ be this interpretation. Recall from  Section~\ref{sec:downgrader} that 
$\intn_2$ is defined by reference to the information flow policy $\flowsto'$ once 
$A$ and $\dom_2$ are fixed, by $\intn_2(\downrelf)(\alpha,a) = (\tff{D}{\flowsto'}(\alpha),a)$
when $\alpha\in A^*$ and $\dom(a) = D$.  

We show that refinement mapping $r$ satisfies conditions T\ref{tref-cond-1}-T\ref{tref-cond-3}. 
(Note that we cannot apply strict refinement, because $r(D) = D$, $r(L) = L$ and  the edges
$\flowstoinf[\top]{D}{L}{\flowsto_1}$ and $\flowstoinf[\downrelf]{D}{L}{\flowsto_2}$ violate 
S\ref{stref-cond-2}.)
T\ref{tref-cond-1} is immediate from the definitions above. 
For T\ref{tref-cond-2}, note the only ways that we could have 
$\inff{\dom_2(a),r(v)}^{\intarch_2}(\alpha,a) = \epsilon$ for $\alpha \in A^*$ and $a\in A$ is 
when $\flowstoinf{\dom_2(a)}{r(v)}{\not\flowsto_2}$ or 
$\dom_2(a) = D$ and $r(v)=L$ and $\intn(\downrelf)(\alpha,a) = \epsilon$. 
In case $\flowstoinf{\dom_2(a)}{r(v)}{\not \flowsto_2}$, we also have 
$\flowstoinf{\dom_1(a)}{v}{\not\flowsto_1}$, so also 
$\inff{\dom_1(a),v}^{\intarch_1}(\alpha,a) = \epsilon$. On the other hand, 
the situation $\dom_2(a) = D$ and $r(v)=L$ and $\intn(\downrelf)(\alpha,a) = \epsilon$
is impossible, since $\intn_2(\downrelf)(\alpha,a) = (\tff{D}(\alpha),a)\neq \epsilon$. 
Thus T\ref{tref-cond-2} holds. 

For T\ref{tref-cond-3}, we need to show that 
for all domains $u$ of $\downarchr$  and $v$ of $\downarch$,   
$\alpha \in A^*$ and $a,b \in A$,   if 
 $\inff{\dom_2(a),r(v)} ^{\intarch_2}(\alpha, a)= \inff{\dom_2(b),r(v)} ^{\intarch_2}(\beta, b) \neq \epsilon$
 then 
 either (i) $a=b$  and $\flowstoinf[\top]{\dom_2(a)}{r(v)}{\flowsto_2}$ 
 or (ii) 
 $\inff{\dom_1(a),v}^{\intarch_1}(\alpha, a) =\inff{\dom_1(b),v}^{\intarch_1}(\beta, b)$. 

Suppose first that $r(v) \in \{C,P,H,D\}$. Note that the only incoming edges for these 
domains in $\downarch$ are labelled $\top$. Hence, if  
$\inff{\dom_2(a),r(v)} ^{\intarch_2}(\alpha, a)= \inff{\dom_2(b),r(v)} ^{\intarch_2}(\beta, b) \neq \epsilon$
it follows that 
$\flowstoinf[\top]{\dom_2(a)}{r(v)}{\flowsto_2}$ and   $\flowstoinf[\top]{\dom_2(b)}{r(v)}{\flowsto_2}$ 
and $(\tff{\dom_2(a)}(\alpha), a) = (\tff{\dom_2(b)}(\beta), b)  $.  Thus $a=b$, and we have (i). 

 Alternately,  suppose that $r(v) = L$. If $\dom_2(a) = \dom_2(b) = L$ then 
 the previous argument also applies. If $\dom_2(a) = L$ and $\dom_2(b) = D$
 (or {\em vice versa}),  then  
$\inff{\dom_2(a),r(v)} ^{\intarch_2}(\alpha, a)= \inff{\dom_2(b),r(v)} ^{\intarch_2}(\beta, b)$
states $(\tff{\dom_2(a)}^{\flowsto_2}(\alpha), a) = (\tff{\dom_2(b)}^{\flowsto'}(\beta), b)  $, 
which is impossible since we cannot have $a=b$ when these actions are in different domains. 

Thus, the only case remaining to be considered is when $r(v) = L$ and $\dom_2(a) = \dom_2(b) = D$. 
Here we have 
$\inff{\dom_2(a),r(v)} ^{\intarch_2}(\alpha, a)= (\tff{D}^{\flowsto'}(\alpha), a)$ and 
$\inff{\dom_2(b),r(v)} ^{\intarch_2}(\beta, b)= (\tff{D}^{\flowsto'}(\beta), b)$, 
so equality of these terms implies that $a=b$ and 
$\tff{D}^{\flowsto'}(\alpha) = \tff{D}^{\flowsto'}(\beta)$. It is not the case that 
$\flowstoinf[\top]{D}{L}{\flowsto_2}$, so we need to establish 
that  
 $\inff{\dom_2(a),r(v)} ^{\intarch_2}(\alpha, a) = \inff{\dom_2(b),r(v)} ^{\intarch_2}(\beta, b)$, 
 which, in light of the edge $\flowstoinf[\top]{D}{L}{\flowsto_1}$, 
 amounts to 
$(\tff{D}^{\flowsto_1}(\alpha),a) = (\tff{D}^{\flowsto_1}(\beta),b)$. 
As we already have $a=b$, it suffices to show $\tff{D}^{\flowsto_1}(\alpha) = \tff{D}^{\flowsto_1}(\beta)$. 

To obtain the required result, we claim that for all $\alpha, \beta\in A^*$, we have that 
$\tff{D}^{\flowsto'}(\alpha) = \tff{D}^{\flowsto'}(\beta)$ implies 
$\tff{D}^{\flowsto_1}(\alpha) = \tff{D}^{\flowsto_1}(\beta)$. 
For this, we prove that there exist functions $F_u$ for $u\in \{T,U, H_U,D,L\}$
such that  $F_u(\tff{r(u)}^{\flowsto'}(\alpha)) = \tff{u}^{\flowsto_1}(\alpha)$. 
The claim is then immediate from the case $u=D$. We define the 
$F_u$ inductively by $F_u(\epsilon) = \epsilon$ and 
$$ F_u(\sigma\concat (\delta,a)) 
=\begin{cases} 
F_u(\sigma) \concat (F_{\dom_1(a)}(\delta), a) & \text{if $\flowstoinf[\top]{\dom_1(a)}{u}{\flowsto_1}$} \\
F_u(\sigma) \concat a & \text{if $\dom_1(a) = T$ and $u=H_U$}\\
F_u(\sigma) & \text{if $\flowstoinf{\dom_1(a)}{u}{\not\flowsto_1}$~.} 
\end{cases} 
$$
Note that in case $\flowstoinf[\top]{\dom_1(a)}{u}{\flowsto_1}$ and $u\in \{U,H_U,D,L\}$, 
we must have $\dom_1(a)\in \{U,H_U,D,L\}$, so the recursion in the first case is well defined. 

We prove by induction on $\alpha\in A^*$ 
that   $F_u(\tff{r(u)}^{\flowsto'}(\alpha)) = \tff{u}^{\flowsto_1}(\alpha)$
for $u\in \{U, H_U,D,L\}$. The base case of $\alpha = \epsilon$ is trivial. 
Consider $\alpha a$, where the statement holds for $\alpha$. 
There are several possibilities: 

\noindent 
Case 1: $\flowstoinf[\top]{\dom_1(a)}{u}{\flowsto_1}$. 
In this case, $\flowstoinf[\top]{\dom_2(a)}{r(u)}{\flowsto'}$. 
Thus, 
\begin{align*} 
F_u(\tff{r(u)}^{\flowsto'}(\alpha a))  
& =  F_u(\tff{r(u)}^{\flowsto'}(\alpha)\concat  (\tff{\dom_2(a)}^{\flowsto'}(\alpha),a))  \\ 
& =  F_u(\tff{r(u)}^{\flowsto'}(\alpha))\concat  (F_u(\tff{\dom_2(a)}^{\flowsto'}(\alpha)),a)  \\ 
& =  \tff{u}^{\flowsto_1}(\alpha))\concat  (\tff{\dom_1(a)}^{\flowsto_1}(\alpha)),a)  & \text{by induction}\\ 
& = \tff{u}^{\flowsto_1}(\alpha a)\mathpunct . 
\end{align*} 

\noindent 
Case 2: $\dom_1(a) = T$ and $u= H_U$. 
In this case, $\flowstoinf[\fT]{\dom_1(a)}{r(u)}{\flowsto_1}$ and 
$\dom_2(a) = P$, $r(u) = H$, and $\flowstoinf[\top]{P}{H}{\flowsto'}$. 
Thus, 
\begin{align*} 
F_u(\tff{r(u)}^{\flowsto'}(\alpha a))  
& =  F_u(\tff{r(u)}^{\flowsto'}(\alpha)\concat  (\tff{\dom_2(a)}^{\flowsto'}(\alpha),a))  \\ 
& =  F_u(\tff{r(u)}^{\flowsto'}(\alpha))\concat a \\ 
& =  \tff{u}^{\flowsto_1}(\alpha))\concat a & \text{by induction}\\ 
& =  \tff{u}^{\flowsto_1}(\alpha))\concat \intn_1(f)(\alpha,a) \\ 
& = \tff{u}^{\flowsto_1}(\alpha a)\mathpunct . 
\end{align*} 

\noindent 
Case 3: $\flowstoinf{\dom_1(a)}{u}{\not \flowsto_1}$ and  $\flowstoinf{\dom_2(a)}{r(u)}{\not \flowsto'}$. 
In this case, 
\begin{align*} 
F_u(\tff{r(u)}^{\flowsto'}(\alpha a))  
& =  F_u(\tff{r(u)}^{\flowsto'}(\alpha))  \\ 
& =  \tff{u}^{\flowsto_1}(\alpha) & \text{by induction}\\ 
& = \tff{u}^{\flowsto_1}(\alpha a)\mathpunct . 
\end{align*} 

\noindent 
Case 4: $\flowstoinf{\dom_1(a)}{u}{\not \flowsto_1}$ and  $\flowstoinf[\top]{\dom_2(a)}{r(u)}{\flowsto'}$. 
Then 
\begin{align*} 
F_u(\tff{r(u)}^{\flowsto'}(\alpha a))  
& =  F_u(\tff{r(u)}^{\flowsto'}(\alpha)\concat  (\tff{\dom_2(a)}^{\flowsto'}(\alpha),a))  \\ 
& =  F_u(\tff{r(u)}^{\flowsto'}(\alpha)) \\ 
& =  \tff{u}^{\flowsto_1}(\alpha)\concat a & \text{by induction}\\ 
& = \tff{u}^{\flowsto_1}(\alpha a)\mathpunct . 
\end{align*}

\end{proof}

Any machine $\sysvar$ that is compliant with architectural
specification $({\downarch}, \archspec_\downarch)$  does
not reveal any 
$C$-dependent proposition to
$L$. Since 
 $\archspec_\downarchr$ refines $\archspec_\downarch$,
the same property holds
for any machine $\sysvar$ that is compliant with $({\downarchr},
\archspec_\downarchr)$.
We thus prove an information security property about architecture $\downarchr$ by reference to
the more abstract architecture $\downarch$.

\begin{theorem}
  If $\sysvar$ is \TFF-compliant with architectural specification $({\downarchr}, \archspec_\downarchr)$ and $\pi(p)$
  depends on $C$ actions at $\alpha$ then $\interpsysv , \alpha \sat \neg \knows {L} p$.
\end{theorem}

\iftr 
\begin{simpleproof}
  Follows easily from Theorem~\ref{thm:down-hc-action-local} and
Theorem~\ref{thm:knowledge-refine}, since 
$\archspec_\downarchr \refines_r \archspec_\downarch$.
\end{simpleproof}
\fi

\newcommand{\ComTC}{\mathtt{ComTC }}
The proof of Theorem~\ref{thm:downr-refines-down} uses the conditions T\ref{tref-cond-1}-T\ref{tref-cond-3}, 
and requires a somewhat laborious induction on the length of $\alpha$ to prove condition T\ref{tref-cond-3}. 
The simpler and more local conditions for strict refinement cannot be used in this case, 
because $r(D) = D$, $r(L) = L$ and  the edges
$\flowstoinf[\top]{D}{L}{\flowsto}$ in $\downarchr$ and
$\flowstoinf[\downrelf]{D}{L}{\flowsto}$ in $\downarchr$ violate condition
S\ref{stref-cond-2}. To illustrate the application of strict refinement, we consider an 
architectural specification that varies $\archspec_{\downarchr}$ by varying the 
allowed interpretations of the filter function $f$. According to $\archspec_{\downarchr}$, 
{\em every}  action of domain $T$ is permitted to have an effect on domain $H_P$. 
In practice, some of the actions of domain $T$ will have the purpose of communicating
information to domain $C$, and it would not be desirable for such actions to be recorded in $H_P$. 
Let $\ComTC$ be the set of such communications actions. 
We specify that information about such actions is not permitted to flow from $T$ to $H_P$. 
Let $\archspec_{\downarchr}'$ be the architectural specification  so that $(A, \dom,
\intn') \in \archspec_{\downarchr}'$ if and only if $(A, \dom, \intn')$ is
an interpretation of $\downarchr$ such that $\intn'(\fT)$ is
defined by $\intn'(\fT)(\alpha,a ) = a$ for all $\alpha\in A^*$ and $a\in A\setminus \ComTC$ with $\dom(a) = T$, 
and $\intn'(\fT)(\alpha,a ) = \epsilon$ otherwise. In particular, we now have, for the 
case where $\dom(a) = T$ and $a \in \ComTC$, that $\intn'(\fT)(\alpha,a ) = \epsilon$, 
whereas in the corresponding interpretation in $\archspec_{\downarchr}$ we would have $\intn(\fT)(\alpha,a ) = a$
in this case. 

\begin{theorem} 
$(\downarchr,\archspec_{\downarchr}')$  strictly refines $(\downarchr,\archspec_{\downarchr})$ by the identity function $r$.%
\footnote{We remark that this result implicitly also requires the technical side condition that $\epsilon \not \in A$ for any set of actions
$A$ in an interpretation.} 
\end{theorem} 

\begin{proof} 
Let $\archint_1= (A, \dom_1, \intn_1) \in \archspec_{\downarchr}'$. We may define an architectural interpretation 
for $\downarchr$ by $\archint_2= (A, \dom_2, \intn_2)$ where $\dom_2 = \dom_1$, and 
 the interpretation $\intn_2$ for $\downarchr$ is defined by 
$\intn_2(\fT)(\alpha,a ) = a$ for all $\alpha\in A^*$ and $a\in A$ with $\dom_2(a) = T$, 
and $\intn_2(\fT)(\alpha,a ) = \epsilon$ otherwise. 

Let $\intarch_1 = (\downarchr, \archint_1)$ and $\intarch_2 = (\downarchr, \archint_2)$. 
The domains of these two interpreted architectures are the same,  so the 
identity function $r$ on the set of domains $D_1 = D_2 = \{C,T,U, H_C, H_P, D,L\}$ has the right type to be a refinement mapping. 
We show that $\intarch_1$ strictly refines $\intarch_2$ by $r$, from which we may conclude that 
$(\downarchr,\archspec_{\downarchr}')$  strictly refines $(\downarchr,\archspec_{\downarchr})$ by $r$.  
For the proof, we use conditions~S\ref{stref-cond-1}-S\ref{stref-cond-3} and~S\ref{stref-cond-5}, and invoke Lemma~\ref{lem:simpref} to obtain~S\ref{stref-cond-4}. 

Conditions S\ref{stref-cond-1} and S\ref{stref-cond-2} are trivial  from the fact that the architectures in 
$\intarch_1$ and $\intarch_2$ are the same and $r$ is the identity function. 
For condition S\ref{stref-cond-3}, suppose that  $u,v\in D_1$ with $\flowstoin[f_1]{u}{v}$ with $f_1\neq \top$. 
Then $f_1= f$, $u = T$ and $v= H_P$. We need to show that 
there exists an edge  $\flowstoin[f_2]{r(u)}{r(v)}$ with  either $f_2 = \top$  or $f_2 \neq \top$  and 
 for all $\alpha\in A^*$ and $a\in A$ with $\dom_1(a) = u$
 we have $\intn_2(f_2)(\alpha, a) = \epsilon$ implies $\intn_1(f_1)(\alpha, a) = \epsilon$. Plainly, the edge 
 $\flowstoin[f]{T}{H_P}$ provides the required edge witness, with $f_2 = f \neq \top$, 
 so it remains to check that for all $\alpha\in A^*$ and $a\in A$ with $\dom_1(a) = T$
 we have $\intn_2(f)(\alpha, a) = \epsilon$ implies $\intn_1(f)(\alpha, a) = \epsilon$.
 This holds trivially, because $\dom_1(a) = T$ implies $\dom_2(a) = T$, 
 in which case $\intn_2(f)(\alpha, a) = a \neq \epsilon$. 
 
 For condition~S\ref{stref-cond-5},  let $a,b \in A$ and $v\in D_1$ such that $\dom_1(a) =\dom_1(b)$, 
and $\flowstoinf[f_1]{\dom_1(a)}{v}{\flowsto_1}$ and $\flowstoinf[f_2]{\dom_2(a)}{r(v)}{\flowsto_2}$, 
where $f_1\neq \top$ and $f_2\neq \top$. Then we must have $\dom_1(a) = \dom_2(b) = T$, $v = H_P$ and $f_1 = f_2 = f$. 
We need to show that for $\alpha , \beta \in A^*$, 
if  $\intn_2(f)(\alpha, a) = \intn_2(f)(\beta, b)\neq \epsilon$, then 
$\intn_1(f)(\alpha, a) =  \intn_1(f)(\beta, b)$. But 
$\intn_2(f)(\alpha, a) = \intn_2(f)(\beta, b)\neq \epsilon$ implies that $a = b \in A\setminus \ComTC$, 
so $\intn_1(f)(\alpha, a) = a = b=  \intn_1(f)(\beta, b)$, as required. 

To obtain condition~S\ref{stref-cond-4} using 
 Lemma~\ref{lem:simpref}, we need also that messages are source-identifying with respect to $r$. 
 For this, define $\source(x) = \dom_1(x)$. 
 Let $\alpha\in A^*$ and $a\in A$, and suppose $\flowstoinf[f_2]{r(\dom_1(a))}{w}{\flowsto_2}$ 
and $\intn_2(f_2)(\alpha, a) \neq \epsilon$. Then we have $f_2 = f$, $\dom_1(a) = T$ and $w = H_P$. 
Thus, $\intn_2(f_2)(\alpha, a) =\intn_2(f)(\alpha, a) = a$, and  
$\source(\intn_2(f_2)(\alpha, a)) = \dom_1(a)$, as required to establish that messages are source-identifying. 
\end{proof}

\section{Implementing Architectures using Access Control}\label{sec:access-control} 

One of the mechanisms that might be used to enforce compliance with an information 
flow architecture is access control restrictions on the ability of domains to read and write
objects. This idea was already implicit in  the   \citet{BLP} 
 approach of enforcing that high level information should not flow to low level domains 
through a ``no read up" and ``no write down" access control policy. 
The idea was given a more semantically well-founded expression by  \citet{Rushby_92}, 
who established a  formal relation between access control systems and 
a  theory of information flow based on intransitive noninterference 
policies. Rushby's  ``reference monitor conditions"
give semantics to the notion of reading and writing, which was absent in the work of Bell and 
La Padula. Rushby's formulation was sharpened and shown to be closely related to 
TA-security by \citet{wiini}.  

In this section, we present a generalization of van der Meyden's formulation 
of access control, and show how enforcement of an access control policy  together with 
local verification of trusted components can be used to assure that
a system is compliant with an  extended architecture. 

\newcommand{\Obj}{\mathtt{Obj}}
\newcommand{\oceq}{\approx^{oc}}

We first recall some definitions and results from \citet{wiini}. 
The system model we have used to this point does 
not require the states of a system to be  equipped with any internal structure. 
In practice, systems typically will be constructed as an assembly of components. 
To capture this, \citet{Rushby_92} introduced the notion of 
a \emph{system with structured state},  which is a system $M$ (with states $S$ and domains $D$) together
with 
\begin{enumerate} 
\item a set $\Obj$ of \emph{objects}, 
\item a set $V$ of \emph{values}, and functions
\item $\contents: S\times \Obj\rightarrow V$, with $\contents(s,n)$
  interpreted as the value of object $n$ in state $s$, and 
\item $\observe,\alter:D \rightarrow {\cal P}(\Obj)$, with $\observe(u)$
  and $\alter(u)$ interpreted as the set of objects that domain $u$
  can observe (or read) and alter (or write), respectively.
\end{enumerate} 
For brevity, we write $s(x)$ for $\contents(s,x)$. We call the pair
$(\observe, \alter)$ the \emph{access control} table of the machine.
For each domain $u$, we define an equivalence relation of 
``observable content equivalence" on states $s,t\in S$ by $s\oceq_u t$ if 
$s(x) = t(x)$ for all $x\in \observe(u)$.

Rushby introduced \emph{reference monitor conditions} on such machines
in order to capture formally the intuitions associated with the pair
$(\observe, \alter)$ being an access control table that restricts the
ability of the actions to ``read'' and ``write'' the objects $\Obj$. 
\Citet{wiini} sharpened these conditions to the following (the difference is in RM2): 
\begin{enumerate} 
\item[RM1.] If $s\oceq_{u}t$ then $\obs_u(s) = \obs_u(t)$ . 

\item[RM2.]
For all actions $a\in A$, states $s,t\in S$ and 
objects  $x\in \alter(dom(a))$, 
if $s\oceq_{\dom(a)}t$ and $s(x) = t(x)$ 
then $(s\cdot a)(x) = (t\cdot a)(x)$. 

\item[RM3.] If  $x\not \in \alter(dom(a))$ then  $s(x) = (s\cdot a)(x)$
 \end{enumerate} 
Intuitively, RM1 states that a domain's observation
depends only on the values of the objects that it can observe (or read). 
RM2 states that if action $a$ is performed in a domain $u$ that 
is permitted to alter an object $x$, then the new value of the object
after the action depends only on its old value and the values of objects 
that domain $u$ is permitted to observe. The final conditions RM3 
says that if action $a$ is performed in a domain that is not permitted
to alter (or write) an object $x$, then the value of $x$ does not change.  

We note that the terminology ``reference monitor conditions" points to the 
fact that these conditions can be enforced by a reference monitor that 
mediates all attempts to perform an action, simply by denying requests by a domain $u$ to 
read an object not in $\observe(u)$ or write to an object not in $\alter(u)$. 

In addition to the reference monitor assumptions, Rushby 
considers a condition stating that if there is an object that
may be altered by domain $u$ and observed by domain $v$, 
then the information flow policy should permit flow of
information from $u$ to $v$. (Obviously, the 
object $x$ provides a channel for information to flow
from $u$ to $v$.) 

\begin{enumerate}
\item[AOI.] If $\alter(u)\cap \observe(v)\neq \emptyset$ then $\flowstoin{u}{v}$.
\end{enumerate}  

\Citet{wiini} shows the following, strengthening a result of \citet{Rushby_92}. 

\begin{theorem} \label{thm:ac-ta}
If $M$ is a system with structured state satisfying RM1-RM3 and AOI  with respect to  
$\flowsto$ then  $M$ is TA-secure with respect to $\flowsto$. 
\end{theorem} 

We now develop a generalization of this result to extended architectures. 
As a first step, note that in extended architectures, the situation where the
information flow policy potentially permits flow of information from domain $u$ to domain $v$ 
corresponds to the existence of an edge $\flowstoin[f]{u}{v}$ for some label $f$ (possibly $\top$). 
This motivates the following variant of condition AOI: 

\begin{enumerate}
\item[AOI$'$.] If $\alter(u)\cap \observe(v)\neq \emptyset$ then $\flowstoin[f]{u}{v}$ for some $f$.
\end{enumerate} 

Next, we develop a set of conditions that check that information flow constraints
of the form $\flowstoin[f]{u}{v}$ with $f \neq \top$ have been correctly implemented in a 
system. Let $\archint = (A,\dom,\intn)$ be an interpretation of architecture $\archvar= (D, \flowsto)$. 
Consider the following constraints in a system $M$ with actions $A$, domains $D$ and 
domain function $\dom$: 

\begin{enumerate} 
\item[I1.] If $\flowstoin[f]{\dom(a)}{u}$ for $f\neq \top$ and $\intn(f)(\alpha, a) = \epsilon$ and 
$x\in \observe(u)\cap \alter(\dom(a))$ then $(s_0\cdot \alpha a)(x) = (s_0\cdot \alpha)(x)$.

\item[I2.] If $\flowstoin[f]{\dom(a)}{u}$ with $f\neq \top$ 
and 
$\flowstoin[g]{\dom(b)}{u}$ with  $f\neq \top$ 
and $\intn(f)(\alpha, a) = \intn(g)(\beta,b)\neq \epsilon$ and 
$x\in \observe(u)\cap (\alter(\dom(a))\cup \alter(\dom(b)))$
and $(s_0\cdot \alpha) (x) = (s_0\cdot\beta)(x)$
then $(s_0\cdot \alpha a)(x) = (s_0\cdot \beta b)(x)$. 
\end{enumerate}

Condition I1 ensures that if filter function $f$ restricts how the
domain of action $a$ may interact with domain $u$
(i.e., $\flowstoin[f]{\dom(a)}{u}$ with $f\neq \top$), and the filter
function interpretation does not allow any information flow ($\intn(f)(\alpha, a) = \epsilon$), then the action does not
change the state of any
object $x$ that domain $u$ may observe and domain 
$\dom(a)$ 
is allowed
to alter. Condition I2 states that if an action $a$ in domain 
$\dom(a)$
may alter an object
$x$ that is observable by domain $u$ and  $\flowstoin[f]{\dom(a)}{u}$
then the new state of object $x$ is determined by the interpretation
of filter function $f$.

We note that verification of these constraints requires consideration 
only of domains that are trusted, in the sense that they have outgoing edges not 
labelled $\top$, and the objects that such domains are permitted to alter. 
Thus verification of these constraints can be localized to the trusted domains. 
The following result states that such local verification, together with 
enforcement of an access control policy consistent with the information flow 
policy via a mechanism satisfying the reference monitor constraints, suffices to 
assure that an information flow policy has been satisfied:

\begin{theorem} \label{thm:acimp}
Let $\intarch$ be a strongly non-conflating interpreted architecture. 
Suppose that $M$ is a system with structured state satisfying RM1-RM3, AOI$'$ and I1-I2. 
Then $M$ is  \TFF-compliant with $\intarch$. 
\end{theorem} 

\begin{proof} 
We show that $\tff{u}(\alpha) = \tff{u}(\beta)$ implies that 
$s_0\cdot \alpha \oceq_u s_0\cdot \beta$, for all domains $u$ and $\alpha, \beta \in A^*$. 
Note that it then follows using RM1 that $\tff{u}(\alpha) = \tff{u}(\beta)$ implies 
$\obs_u(\alpha) = \obs_u(\beta)$, which shows that $M$ is \TFF-compliant with $\intarch$. 

The proof proceeds by induction on the combined length of $\alpha$ and $\beta$. 
The base case of $\alpha = \beta = \epsilon$ is trivial.
Suppose that $\tff{u}(\alpha a) = \tff{u}(\beta)$, where the statement holds for 
sequences of shorter combined length. We consider several cases: 

Case 1: $\notflowstoin{\dom(a)}{u}$. In this case
$\tff{u}(\alpha) = \tff{u}(\alpha a) = \tff{u}(\beta)$, so by induction we
have $s_0\cdot \alpha \oceq_u s_0\cdot \beta$. We need to 
show that $s_0\cdot \alpha a \oceq_u s_0\cdot \beta$, i.e., 
that $(s_0\cdot \alpha a)(x)  =  (s_0\cdot \beta)(x)$ for all $x\in \observe(u)$. 
Let $x\in \observe(u)$. Since $\notflowstoin{\dom(a)}{u}$, it follows using 
AOI$'$ that $x\not \in \alter(\dom(a))$. Thus, by RM3, we obtain
that $(s_0\cdot \alpha a)(x)  = 
((s_0\cdot \alpha)\cdot a)(x)  = (s_0\cdot \alpha)(x)  =  (s_0\cdot \beta)(x)$, 
as required.

Case 2: $\flowstoin[\top]{\dom(a)}{u}$. In this case
$\tff{u}(\beta) = \tff{u}(\alpha a)  = \tff{u}(\alpha) \concat (\tff{\dom(a)}(\alpha), a) $, 
so $\beta \neq \epsilon$. Let $\beta = \gamma b$. 
If $\notflowstoin{\dom(b)}{u}$, then we can swap the roles of $\alpha a$ and $\gamma b$ 
and apply Case 1. If $\flowstoin[f]{\dom(b)}{u}$ with $f\neq \top$, then 
$\tff{u}(\gamma b)  = \tff{u}(\gamma) \concat \intn(f)(\gamma,b)$, 
and we obtain that $\intn(f)(\gamma,b) = (\tff{\dom(a)}(\alpha), a)$. 
Since the architecture is strongly non-conflating, this case is not possible. 
We are left with the case that $\flowstoin[\top]{\dom(b)}{u}$. 
Here $\tff{u}(\gamma b)  = \tff{u}(\gamma) \concat (\tff{\dom(b)}(\gamma),b)$, 
and we conclude that $\tff{u}(\gamma)= \tff{u}(\gamma)$ and  
$\tff{\dom(a)}(\alpha) = \tff{\dom(b)}(\gamma)$ and $a=b$. 
By the induction hypothesis, we obtain that 
$s_0\cdot \alpha \oceq_u s_0\cdot \gamma$ and 
$s_0\cdot \alpha \oceq_{\dom(a)} s_0\cdot \gamma$. 
The former states that for all objects $x\in \observe(u)$, 
we have $(s_0\cdot \alpha )(x) = (s_0\cdot\gamma)(x)$. 
Thus, by RM2 and the fact that $a=b$, it follows that 
$(s_0\cdot \alpha a)(x) =((s_0\cdot \alpha)\cdot  a)(x) =
((s_0\cdot \gamma)\cdot  b)(x) =
 (s_0\cdot\gamma b)(x)$ for all $x\in \observe(u)$. This shows that $(s_0\cdot \alpha a)  \oceq_u 
 (s_0\cdot \gamma b)$, as required. 

Case 3:  $\flowstoin[f]{\dom(a)}{u}$.
In this case we have 
$$\tff{u}(\beta) = \tff{u}(\alpha a)  = \tff{u}(\alpha) \concat \intn(f)(\alpha,a)\mathpunct.$$
There are several possibilities: 
\begin{itemize} 
\item Case 3a: $\intn(f)(\alpha,a) = \epsilon$. 
In this case we obtain $ \tff{u}(\alpha) = \tff{u}(\beta)$, 
so $s_0\cdot \alpha \oceq_u s_0\cdot \beta$ by induction. 
By RM3, we have $(s_0\cdot \alpha a)(x) = ((s_0\cdot \alpha)\cdot a)(x) = (s_0\cdot \alpha)(x)$ for 
$x\in \observe(u)$ such that $x\not \in \alter(\dom(a))$. 
Further, by I1, we have 
$(s_0\cdot \alpha a)(x) = (s_0\cdot \alpha)(x)$ for $x\in \observe(u)\cap \alter(\dom(a))$. 
Thus, $s_0\cdot \alpha a \oceq_u s_0\cdot \alpha$, and we 
conclude that $s_0\cdot \alpha a \oceq_u s_0\cdot \beta$, as required. 

\item Case 3b: $\intn(f)(\alpha,a) \neq \epsilon$. 
Here we must have $\beta \neq \epsilon$, so let $\beta =\gamma b$. 
Arguing as above, we cannot have $\flowstoin[\top]{\dom(b)}{u}$ since that
contradicts the assumption that the architecture is strongly non-conflating.  
The case that $\notflowstoin{\dom(b)}{u}$ can be handled using Case 1 above. 
We are left with the possibility that $\flowstoin[g]{\dom(b)}{u}$ for some $g\neq \top$. 
Here $\tff{u}(\gamma b )  = \tff{u}(\gamma) \concat \intn(g)(\gamma,b)$. 
The situation where $\intn(g)(\gamma,b)=\epsilon$ can be handled using the 
argument of Case 3a. Thus, we may assume that  $\intn(g)(\gamma,b)\neq \epsilon$. 
We then obtain that 
$\tff{u}(\alpha )  = \tff{u}(\gamma)$ and $\intn(f)(\alpha, a) = \intn(g)(\gamma,b) \neq \epsilon$. 
By the induction hypothesis, 
we have $s_0\cdot \alpha \oceq_u s_0\cdot \gamma$, 
so $(s_0\cdot \alpha)(x) = (s_0\cdot \gamma)(x)$ for all $x \in \observe(u)$. 
We need to show that 
$(s_0\cdot \alpha a)(x) = (s_0\cdot \gamma b )(x)$ for all $x \in \observe(u)$. 
If $x\in \observe(u) \cap (\alter(\dom(a) \cup \alter(\dom(b))$, then this 
follows using I2. On the other hand, if 
$x\not \in \alter(\dom(a))$  and $x\not \in \alter(\dom(b))$, then 
by  RM3 we get 
$(s_0\cdot \alpha a)(x) = ((s_0\cdot \alpha)\cdot a)(x) = (s_0\cdot \alpha)(x)
= (s_0\cdot \gamma)(x) = ((s_0\cdot \gamma)\cdot b)(x) = (s_0\cdot \gamma b)(x)$, 
so the desired conclusion holds in either case.  
\end{itemize} 
\end{proof}

Conditions I1, I2 are still somewhat non-local because of the reference to $\alpha$. 
Since architectural specifications are stated in terms of these state sequences, references
to them cannot be completely eliminated. However, it is often 
convenient to factor this reference via properties of the state 
of the system. Suppose that for the edge labels $f$  there exist functions $F(f)$ with domain $S\times A$, 
satisfying the following constraints:  
\begin{enumerate} 
\item[E1.] $\intn(f)(\alpha,a) =\intn(g)(\beta, b)$ implies 
$F(f)(s_0\cdot \alpha, a) = F(g)(s_0\cdot \beta,b)$
\item[E2.] $\intn(f)(\alpha,a)=\epsilon$ implies $F(f)(s_0\cdot \alpha, a) = \epsilon$. 
\end{enumerate} 
Intuitively, $F$  gives a state-based encoding of  an approximation to $\intn$.

The following conditions use $F$ to give a variant of  the conditions I1, I2 that is 
stated with respect to the states of the machine: 
\begin{enumerate} 
\item[I1$'$.] If $\flowstoin[f]{\dom(a)}{u}$ for $f\neq \top$ and $F(f)(s, a) = \epsilon$ and 
$x\in \observe(u)\cap \alter(\dom(a))$ then $(s\cdot  a)(x) = (s)(x)$.

\item[I2$'$.] If $\flowstoin[f]{\dom(a)}{u}$ with $f\neq \top$ 
$\flowstoin[g]{\dom(b)}{u}$ with  $g\neq \top$ 
and $F(f)(s, a) = F(g)(t,b)\neq \epsilon$ and 
$x\in \observe(u)\cap (\alter(\dom(a))\cup \alter(\dom(b)))$ 
and $s(x) = t(x)$ 
then $(s\cdot  a)(x) = (t\cdot  b)(x)$. 
\end{enumerate}
The following result states that this state-factored representation 
implies the conditions I1-I2. 

\begin{theorem} 
Conditions E1-E1 and I1$'$-I2$'$ imply I1 and I2.
\end{theorem} 

\begin{proof} 
Assume E1-E1 and I1$'$-I2$'$. For I1, 
 suppose that $\flowstoin[f]{\dom(a)}{u}$ for $f\neq \top$ and $\intn(f)(\alpha, a) = \epsilon$ and 
let $x\in \observe(u)\cap \alter(\dom(a))$.  
We need  to show that $(s_0\cdot \alpha a)(x) = (s_0\cdot \alpha)(x)$. 
From $\intn(f)(\alpha, a) = \epsilon$, we get by E2 that $F(s_0\cdot \alpha,a) = \epsilon$. 
By I1$'$ with $s= s_0\cdot\alpha$, we the conclude that $(s_0\cdot \alpha a)(x) = (s_0\cdot \alpha)(x)$, 
as required.

For I2, suppose that  $\flowstoin[f]{\dom(a)}{u}$ with $f\neq \top$ 
and 
$\flowstoin[g]{\dom(b)}{u}$ with  $f\neq \top$ 
and $\intn(f)(\alpha, a) = \intn(g)(\beta,b)\neq \epsilon$ and 
$x\in \observe(u)\cap (\alter(\dom(a))\cup \alter(\dom(b)))$
and $(s_0\cdot \alpha) (x) = (s_0\cdot\beta)(x)$. 
We need to show that $(s_0\cdot \alpha a)(x) = (s_0\cdot \beta b)(x)$. 
By E1, we have $F(\alpha,a) = F(\beta,b)$. 
If both values are equal to $\epsilon$, then we use I1$'$ to conclude that 
$(s_0\cdot \alpha a)(x) = (s_0\cdot \alpha)(x)$ 
and $(s_0\cdot \beta b)(x) = (s_0\cdot \beta)(x)$. It then  
follows that $(s_0\cdot \alpha a)(x) = (s_0\cdot \beta b)(x)$, as required. 
Alternately, if $F(\alpha,a) = F(\beta,b)\neq \epsilon$, then by I2$'$, 
with $s = s_0\cdot \alpha$ and $t = s_0\cdot \beta$, 
we conclude that 
$(s_0\cdot \alpha a ) (x) = (s\cdot a)(x) =  (t\cdot b)(x) = (s_0\cdot\beta b)(x)$, 
as required.
\end{proof} 

To illustrate the application of these results, we consider the examples 
of architectural specifications introduced above, and show how some particular systems
may be proved to be implementations of these specifications. 

\subsection{Starlight Architecture} 

\newcommand{\togpos}{\mathit{togpos}}  
\newcommand{\logL}{\mathit{logL}}  
\newcommand{\logH}{\mathit{logH}}  
\newcommand{\getH}{\mathtt{get}}  
\newcommand{\putL}{\mathtt{put}} 
\newcommand{\ES}{\mathcal{K}}   
\newcommand{\EL}{\mathit{E_L}}  
\newcommand{\EH}{\mathit{E_H}}

We present an implementation of the Starlight architectural specification $(\strltArch, \archspec_{\strltArch})$
as a system  with structured state.  We first select an architectural interpretation 
$\archint = (A,\dom, \intn)$ in $\archspec_{\strltArch}$. 
The actions $A$ of the interpretation are given with their associated domain as follows: 
\begin{itemize} 
\item The actions of domain $S$ consist of the toggle action $\strltToggle$, plus
actions $k$ from a set $\ES$, which intuitively represents the set of 
keyboard actions. 

\item The actions of domain $H$ consist of the action $\getH$, plus actions $h$ drawn from some set 
$\EH$ representing possible High level events. 

\item The actions of domain $L$ consist of the action $\putL$, plus actions $l$  drawn from some set 
$\EL$ representing the possible Low level events. 
\end{itemize} 
As required by $\archspec_{\strltArch}$, we assume that $\intn$ satisfies
\begin{align*}
 \intn(\strltFltr)(\alpha, a) &=
  \begin{cases}
    a &\text{if $a=\strltToggle$ or} \\
      &~ \text{$\numof_\strltToggle(\alpha)$ is odd and $a \in A_S$} \\
    \epsilon & \text{otherwise}
  \end{cases}
\end{align*}
where $\numof_\strltToggle(\alpha)$ is the number of occurrences of $t$ in $\alpha$. 

Next, we construct a system $M$ for the interpreted architecture $(\strltArch, \archint)$. 
The system $M$ is based on the set of objects 
$\Obj = \{ \logH, \logL,\togpos\}$. 
Intuitively,  $\logH$ and $\logL$ are logs of actions observable  to $H$ and $L$, respectively, 
and $\togpos$ indicates the position of the toggle switch.
The objects $\logH$ and $\logL$ have values in $A^*$, and $\togpos$ has 
a value in $\{H,L\}$. A state $s$ is just an assignment of value of the given type 
to each of the objects, and we have $\contents(s,x) = s(x)$. 
The initial state $s_0$ is defined by 
$s(\logH)=s(\logL) = \epsilon$ and $s(\togpos) = H$. Transitions are defined by the 
following code associated to each of the actions: 
\begin{itemize} 
\item for domain $S$: 
\begin{itemize} 
\item[] $\strltToggle$:  if $\togpos=H$ then $togpos := L$ else $\togpos := H$
\item[] $k$: if $\togpos=H$ then $\logH := \logH \concat k$ else $\logL := \logL \concat k$  
\end{itemize} 
That is, the toggle action $\strltToggle$ changes the position of the toggle, and $k$ is recorded in either the 
High level log or the Low level log, depending on the position of the toggle. 

\item for domain $H$: 
\begin{itemize} 
\item[] $h$:  $\logH := \logH \concat h$
\item[] $\getH$:  $\logH := \logH \concat \logL$
\end{itemize} 
Thus, the events $h$ are recorded in the High level log, and $\getH$ fetches a 
copy of the low level log.

\item for domain $L$: 
\begin{itemize} 
\item[] $l$:  $\logL := \logL \concat l$
\item[] $\putL$:  $\logH := \logH \concat \logL$
\end{itemize} 
Similarly the events $l$ are recorded in the Low level log, and $\putL$ pushes a 
copy of the low level log into the High level log.  

\end{itemize} 

The access control functions $\alter$ and $\observe$ are defined by the following 
access control table. 
\begin{table}[h]
\centerline{ 
\begin{tabular}{|c||c|c|c|} 
\hline
              & $H$ & $L$ & $S$ \\ 
 \hline 
$\logH$ &  a,o   &  a        & a        \\ 
$\logL$ &    o    &   a,o      & a        \\ 
$\togpos$ &          &         &   a,o      \\ 
\hline
\end{tabular} }
\end{table} 
Here, for an object $x$ and a domain $u$, we have 
an entry $a$ (or $o$) in the corresponding cell of the table just when 
$x\in \alter(u)$ (respectively, $x \in \observe(u)$).  We define observations
in the system by allowing each domain to observe the values of all of its observable 
objects. That is, for each domain $u$ and state $s$, we define $\obs_u(s) = s\restrict \observe(u)$. 

\begin{proposition} 
The system $M$ is \TFF-compliant with architectural specification $(\strltArch, \archspec_{\strltArch})$.
\end{proposition} 

\begin{proof} 
We show that $M$ satisfies RM1-RM3, AOI$'$ and I1-I2 with respect to $(\strltArch, \archint)$, 
 and invoke Theorem~\ref{thm:acimp}
to conclude \TFF-compliance. 
RM1 is immediate from the fact that we have defined $\obs_u(s) = s\restrict \observe(u)$.

For RM2, we consider each of the actions $a$ and $x\in \alter(\dom(a))$ in turn, 
and show that if $s\oceq_{\dom(a)} t$ and $s(x) = t(x)$ then $(s\cdot a) (x) = (t\cdot a)(x)$. 
Note that we need to consider only those $x\in \alter(\dom(a))$ whose value could be changed by action $a$, since otherwise 
$(s\cdot a) (x) = (t\cdot a)(x)$ is immediate from $s(x) = t(x)$. (Since the only way an action could change an object, given the 
type of code we have used to define the transition relation is by use of the assignment operator, for each action 
we can take the set of objects that it could change to be the set of objects to which the code for the action
makes an assignment.) 
\begin{itemize} 
\item {\em Actions $a$ with $\dom(a) = S$.}  Suppose $s\oceq_{S} t$. 
Since $\observe(S) = \{\togpos\}$, this means $s(\togpos)=t(\togpos)$. 
Because $\alter(S) = \Obj$, we need to consider  all objects $x \in \Obj$ that could be changed by $a$, 
assuming that $s(x) = t(x)$.  

The only object in $\alter(S)$ that could be changed by 
action $\strltToggle$ is $\togpos$. If $s(\togpos) = t(\togpos) = H$, 
 we have $(s\cdot \strltToggle)(\togpos) = L = (t\cdot \strltToggle)(\togpos)$, as required, 
 and similarly if $s(\togpos) = t(\togpos) = L$.  
 
 For action $k$, the only objects $x$ that could be changed by $k$ are $\logH$ and $\logL$. 
 We consider the case of $\logL$. Assume $s(\logL) = t(\logL) $. 
 If $s(\togpos) = t(\togpos) = H$ then $(s\cdot k)(\logL) = s(\logL) = t(\logL) = (t\cdot k)(\logL)$, 
 as required. If $s(\togpos) = t(\togpos) = L$, then $(s\cdot k)(\logL) = s(\logL)\concat k = t(\logL)\concat k = (t\cdot k)(\logL)$.  
  The argument for $\logH$ is similar.
 
 \item  {\em Actions $a$ with $\dom(a) = H$.}  Assume $s\oceq_{H} t$. 
By the access control table, this means that $s(\logH) = t(\logH)$ and $s(\logL) = t(\logL)$. 
 For actions $h$ and $\getH$,  note that the only object these actions could change is $\logH$. 
We have $(s\cdot h)(\logH) = s(\logH) \concat k = 
t(\logH) \concat k  = (t\cdot h)(\logH)$. 
Similarly,  $(s\cdot h)(\logH) = s(\logH) \concat s(\logL) = 
t(\logH) \concat t(\logL) = (t\cdot h)(\logH)$. 

\item  
{\em Actions $a$ with $\dom(a) = L$.}  Assume $s\oceq_{L} t$. 
By the access control table, this means that  $s(\logL) = t(\logL)$. 

Action $l$ could change only $\logL$. 
We have $(s\cdot l)(\logL) = s(\logL) \concat l = 
t(\logL) \concat l  = (t\cdot l)(\logL)$. 

Action $\putL$ could change only $\logH$. Assume that 
$s(\logH) = t(\logH)$. Then 
$(s\cdot \putL)(\logH) = s(\logH) \concat s(\logL) = 
t(\logH) \concat  t(\logL) = (t\cdot \putL)(\logH)$. 
\end{itemize} 
This completes the proof that RM2 is satisfied. 

For RM3, we need to show that if $x\not \in \alter(\dom(a))$, then $(s\cdot a)(x) = s(x)$. 
This can be done by inspection of the code for each action. 
In the case of $\dom(a) = S$, we have $\Obj \setminus \alter(\dom(a))= \emptyset$, so the 
claim is trivial. For $\dom(a) = H$, we have $\Obj \setminus \alter(\dom(a))= \{\logL, \togpos\}$,
but the only object assigned to by $H$ actions is $\logH$, so the values of objects in this set are 
not changed by $H$ actions. For $\dom(a) = L$, we have $\Obj \setminus \alter(\dom(a))= \{\togpos\}$,
and $L$ actions do not assign to $\togpos$, so again there is no change to the value of this object. 

To show AOI$'$, we consider the possible values for $x \in \alter(u) \cap \observe(v)$, and check that there
exists an edge $\flowstoin[f]{u}{v}$ in each case. 
\begin{itemize} 
\item For $x = \logH$, the only case of $x \in \observe(v)$ is $v =H$. Since there is an edge from every 
domain $v$ to $H$ (in the case of $v=H$ there is the implicit reflexive edge labelled $\top$), 
AOI$'$ is immediate in this case. 

\item For $x = \logL$, we have that $x\in \alter(u)$ just when $u\in \{L,S\}$, 
and $x \in \observe(v)$ just when $v\in \{L,H\}$. For each of these possible values of $u,v$
there exists an edge from $u$ to $v$. 

\item for $x = \togpos$, we have $x \in \alter(u) \cap \observe(v)$ just when $u=v=S$, 
in which case we have the reflexive edge. 
\end{itemize} 

Before considering I1-I2, we establish the following invariant on the states of $M$: 
\begin{equation} 
\text{For all $\alpha \in A^*$, the number $\numof_\strltToggle(\alpha)$ is even iff $(s_0\cdot \alpha)(\togpos) = H$.} \label{sltinvt}
\end{equation}  
It is worth remarking that the proof of this invariant itself benefits from the 
property RM3 that we have just established. Note that 
$\togpos \in \alter(u)$ iff $u = S$. Thus, using RM3, we have that 
$(s_0\cdot \alpha)(\togpos) = s_0\cdot (\alpha\restrict \{S\}) (\togpos)$. 
Moreover, inspection of the code shows that the actions $k$ of domain $S$ 
do not change the value of $\togpos$. Thus, in fact 
$(s_0\cdot \alpha)(\togpos) = s_0\cdot (\alpha\restrict \{\strltToggle\}) (\togpos)$, where we 
interpret $\alpha\restrict B$ where $B$ is a set of actions to mean the subsequence of 
$\alpha$ of actions in $B$. A straightforward induction on sequences $\beta \in \{\strltToggle\}^*$ shows that 
$(s_0\cdot \beta)(\togpos) = H$ if $\numof_\strltToggle(\beta)$ is even and 
$(s_0\cdot \beta)(\togpos) = L$ otherwise. Since  
$\numof_\strltToggle(\alpha) = \numof_\strltToggle(\alpha\restrict \{\strltToggle\})$, we have 
$\numof_\strltToggle(\alpha)$ even iff $\numof_\strltToggle(\alpha\restrict \{\strltToggle\})$ is even 
iff $(s_0\cdot \alpha\restrict \{\strltToggle\})( \togpos) = H$ 
iff $(s_0\cdot \alpha)( \togpos) = H$, as required.

We now prove I1. 
The only case where $\flowstoin[f]{\dom(a)}{u}$ for $f\neq \top$ is $\dom(a) = S$ and $u = L$ and $f= \strltFltr$. 
 Suppose that $\intn(\strltFltr)(\alpha, a) = \epsilon$ and 
$x\in \observe(L)\cap \alter(S)$.  We need to show that
$(s_0\cdot \alpha a)(x) = (s_0\cdot \alpha)(x)$. 
Now $\observe(L)\cap \alter(S) = \{\logL\}$, so we have $x = \logL$. 
Further, if $\dom(a) = S$ then $\intn(\strltFltr)(\alpha,a) = \epsilon$ just when $a \neq \strltToggle$
and $\numof_\strltToggle(\alpha)$ is even. It therefore suffices to consider the case of $a = k$. 
By the invariant~(\ref{sltinvt}), since  $\numof_\strltToggle(\alpha)$ is even, 
we have $(s_0\cdot \alpha)(\togpos) = H$, so the code for $a=k$ takes the else branch, 
which does not change $x=\logL$. Hence $(s_0\cdot \alpha a)(x) = (s_0\cdot \alpha)(x)$, 
as required. 

For I2, the only case where $\flowstoin[f]{\dom(a)}{u}$ with $f\neq \top$ 
and  $\flowstoin[g]{\dom(b)}{u}$ with  $g \neq \top$ is 
$\dom(a) = \dom(b) = S$ and $u = L$ and $f =g= \strltFltr$.
Suppose that  $\intn(f)(\alpha, a) = \intn(g)(\beta,b)\neq \epsilon$ and 
$x\in \observe(u)\cap (\alter(\dom(a))\cup \alter(\dom(b)))$
and $(s_0\cdot \alpha) (x) = (s_0\cdot\beta)(x)$. We need to show that 
 $(s_0\cdot \alpha a)(x) = (s_0\cdot \beta b)(x)$. 
In the relevant case,  $\observe(u)\cap (\alter(\dom(a))\cup \alter(\dom(b))) = \observe(L) \cap \alter(S) = 
\{\logL\}$, so we have $x = \logL$. 
Now,  if $\intn(f)(\alpha, a) = \intn(g)(\beta,b)\neq \epsilon$ 
then $a = b = \strltToggle$ or both $\numof_\strltToggle(\alpha)$ and $\numof_\strltToggle(\beta)$ 
are odd and $a=b \in A_S \setminus \{\strltToggle\}$. In the former case, the 
code for $a = b = \strltToggle$ does not change the value of $x= \logL$, 
so it is immediate from $(s_0\cdot \alpha) (x) = (s_0\cdot\beta)(x)$ that 
$(s_0\cdot \alpha a)(x) = (s_0\cdot \beta b)(x)$. In the latter case, we have by the invariant~(\ref{sltinvt})
that $(s_0\cdot \alpha)(\togpos) = (s_0\cdot \beta)(\togpos) = L$, so the code for $a = k$ 
takes the else branch from both states and we have 
$(s_0\cdot \alpha a )(\logL) = 
(s_0\cdot \alpha )(\logL) \concat a =  
(s_0\cdot \beta )(\logL) \concat a =  
(s_0\cdot \beta a)(\logL)$, as required. 
\end{proof} 

Intuitively, the proof of RM2 can be summarized to be 
a consequence of the fact that for each action $a$, the code for action $a$ ``reads"  only 
objects in $\observe(\dom(a))$, and ``writes" to objects in $\alter(\dom(a))$ in 
ways that depend only on the values read and the values of the object written to. 
Note that the action $\putL$ assigns a value to the $H$ object $\logH$ that depends on the 
value of the same object $\logH$, which $L$ is 
not permitted to observe/read. However, this is not a violation of RM2. 
We expect that it is possible to develop a static analysis that would be sound for 
checking RM1-3, but we do not pursue this here.

\subsection{Downgrader Architecture}

\newcommand{\buff}{\mathit{buff}}
\newcommand{\logP}{\mathit{logP}}

Next, we give an example that illustrates the role that code structure can play in the
enforcement of an architecture. 
We sketch a class of  implementations of the refined downgrader architectural specification 
$(\downarchr, \archspec_{\downarchr})$ of Section~\ref{sec:downgrader-refine}, 
in the form of systems with structured state. 
Let $\archint = (A, \dom,\intn)$ be an architectural interpretation in $\archspec_{\downarchr}$. 
The set of actions $A_u$ for each domain $u$ is some arbitrary set, 
and $\intn(\fT)$ is
defined by $\intn(\fT)(\alpha,a ) = a$ for all $\alpha\in A^*$ and $a\in A$ with $\dom(a) = T$,
and $\intn(\fT)(\alpha,a ) = \epsilon$ otherwise. 

We describe the system $M$ abstractly: for each  domain $u$ 
 we assume that some set of objects $\Obj_u$ is given. 
For the domain $H_P$, we take $\Obj_{H_P} = \{\buff,\logP\}$. 
Intuitively, $\logP$ is a log of events that are observable to the domain $H_P$, 
and $\buff$ is a buffer storing declassification requests from $L$. 
The other domains $C,T,U, H_C,D$ and $L$ are associated with an arbitrary set of objects. 
The set of all objects $\Obj$ is the union of these sets $\Obj_u$.

The access control table on these objects is given in the following ``capability" formatted table. 
\begin{table}[h] 
\begin{tabular}{|c|c|c|c|c|c|c|c|} 
\hline
                    &    $C$       &  $T$          &  $U$          &  $H_C$            &  $H_P$     &  $D$                &  $L$ \\ 
                    \hline 
$\observe$  &  $\Obj_C$ &  $\Obj_T$ &  $\Obj_U$  & $\Obj_{H_C}$  &  $\logP$    & $\buff,\Obj_D$ & $\Obj_L$ \\ 
\hline 
$\alter$        &  $\Obj_C$     & $\Obj_T$    &  $\Obj_U$  & $\Obj_{H_C}$ & $\Obj_{T}$,$\Obj_{U}$ & $\Obj_{D}$,$\Obj_{L}$ & $\Obj_{L}$\\ 
                    &  $\Obj_{H_C}$ & $\Obj_{H_C}$&   $\logP$    &  $\Obj_{T}$   & $\Obj_{D}, \logP$           & $\buff, \logP$                 & $\buff$ \\ 
                    &                         &  $\logP$         &                   & $\Obj_{C}$     &                                       &                                       &  \\ 
\hline 
\end{tabular} 
\end{table} 
\\
For example, this table says that an object $x \in \Obj_U$ can be observed by domain $U$ only, and can be altered by domains 
$U$ and $H_P$. 

 States are defined to be assignments of a value in some set for each $x\in \Obj$. 
In case of $\buff$, the value $s(\buff)$ is assumed to be in the set $A_L^*$ 
of $L$ action sequences.  In case of $\log_P$, the value is assumed to be 
in the set $(A_T\cup A_U \cup A_{H_P})^*$, i.e., this log records actions in 
the  domains $T,U$ and $H_P$. We assume each domain observes all its observable objects, i.e., 
$\obs_u(s) = s\restrict \observe(u)$.

In order to focus on the semantic conditions associated to the 
edge $\flowstoin[f]{T}{H_P}$, we assume that the actions for all domains 
$T$ are given semantics in such a way as to satisfy the reference monitor
conditions RM2-RM3. For actions $a\in A_T$, we assume that the 
code for these actions  (which determines the state transition when the 
action is performed) has the following structure: 
$$ \logP :=\logP \concat a ; P_a$$
where $P_a$ is a program that reads only objects in 
$\Obj_T$ and 
writes only objects in $\Obj_T\cup \Obj_{H_C}$. 

We now argue that a system within this class of systems satisfies conditions
RM1-RM3, AOI$'$ and I1-I2. Condition RM1 is trivial from the definition of
observations. Conditions RM2-RM3 have been assumed for all actions
except $a \in A_T$. We argue that the latter also satisfy RM2. 
Given the allowed code pattern for such actions $a$, 
the only objects $x$ that are altered by $a$ 
are in $\{\logP \} \cup \Obj_T\cup \Obj_{H_C}$. The object $\logP$ is altered
by appending $a$. This gives a result that depends only on the previous value of 
$\logP$ (and the action $a$ being performed), which satisfies RM2. 
Other objects $x\in \Obj_T\cup \Obj_{H_C}$ are altered by the code 
in ways that depend only on the action $a$ and the value of objects in $\Obj_T = \observe(T)$, 
which again satisfies RM2. 

For RM3 in the case of actions $a\in A_T$, note that the 
complement of $\alter(T)$ is $\Obj_C\cup \Obj_U\cup   \Obj_D\cup \Obj_L \cup \{\buff\}$. 
Since none of these objects are written by the code pattern for $a$, 
RM3 is satisfied in this case. 

Satisfaction of AOI$'$ can be checked by inspection, considering
all possibilities for $x\in \alter(u) \cap \observe(v)$ in turn and 
verifying that there is an edge from $u$ to $v$ in the policy in each case. 
For example, we have $\logP \in \alter(T)\cap \observe(H_C)$, 
but we also have the edge $\flowstoin[f]{T}{H_C}$. 

To check conditions I1-I2, we need only consider the edge $\flowstoin[f]{T}{H_C}$, i.e.,
we consider actions $a$ with $\dom(a) = T$ and let $u = H_C$. 
Condition I1 has precondition $\intn(f)(\alpha,a) = \epsilon$. 
Since always $\intn(f)(\alpha,a) = a \neq \epsilon$, the precondition for I1
is always false and this condition is trivially true. 

For condition I2, suppose that $\dom(a) = \dom(b) = T$
and $\intn(f)(\alpha,a) = \intn(f)(\beta,b) \neq \epsilon$
and $x\in \observe(H_P) \cap \alter(T) $ and $(s_0\cdot\alpha)(x) = (s_0\cdot\beta)(x)$. 
By definition of $\intn(f)$ we have $a=b$. 
Since $\observe(H_P) \cap \alter(T) = \{\logP\}$, we have $x = \logP$
and $(s_0\cdot\alpha)(\logP) = (s_0\cdot\beta)(\logP)$. 
Since the code $P_a = P_b$ does not alter $\logP$, 
we obtain that  $(s_0\cdot\alpha a)(\logP) = 
(s_0\cdot\alpha)(\logP)\concat a = (s_0\cdot\beta)(\logP)\concat b = 
(s_0\cdot\beta b)(\logP)$, as required. 

This completes the argument that a system constructed as described
\TFF-complies with the refined downgrader architectural specification
$(\downarchr, \archspec_{\downarchr})$. We note that by Theorem~\ref{thm:downr-refines-down}
and Theorem~\ref{thm:sref}, it follows that such a system also \TFF-complies
with the architectural specification $(\downarch, \archspec_{\downarch})$. 

\subsection{Electronic Election}

\newcommand{\vyes}{\mathtt{yes}}
\newcommand{\vno}{\mathtt{no}}
\newcommand{\tally}{\mathtt{tally}}
\newcommand{\bb}{\mathit{bb}}

We describe a system that implements the  architectural specification $(\elecarch,
\archspec_{\elecarch})$ for an electronic election of Section~\ref{sec:election}. 

We first select a particular architectural interpretation 
$\archint = (A, \dom, \intn)$  that satisfies this specification. 
We suppose that the election is a referendum with the voters voting either
``yes" or ``no", and the decision determined by a majority of the voters.  
We take the set of domains to be $V \cup \{\elecauth\}$, where 
$V = \{V_1, \ldots,V_n\}$ represents the set of voters, and  $\elecauth$
is the election authority. The set of actions $A = \{\vyes^v~|~v\in V\} \cup 
\{\vno^v~|~v\in V\} \cup \{\tally\}$, with associated domains
given by $\dom(\vyes^v) = \dom(\vno^v) = v$ for all $v\in V$ and 
$\dom(\tally) =  \elecauth$. Note that these actions are voter-homogeneous, as required
by $\archspec_{\elecarch}$. 

The architecture has just one edge label $\announce$. 
Thus, for the interpretation $\intn$, we need to define $\intn(\announce)$. 
Given a sequence $\alpha$, define the {\em latest action of voter $v$}
to be the action $a\in \{\vyes^v, \vno^v\}$ such that 
$\alpha = \alpha_0 a \alpha_1$ and $\alpha_1$ contains no action $b$ with 
$\dom(b) = v$, if such a decomposition exists, or $\bot$ otherwise. 
We now define $\intn(\announce)(\alpha, \tally)$
for  $\alpha\in A^*$ to be the number of voters $v$ whose latest action is $\vyes^v$. 
(Since the election authority has only the one action $\tally$, this is all that is required to
specify $\intn$.) 
That is, in this interpretation, the information that the election authority is permitted
to reveal is the number of voters who have voted yes in the latest round of voting. 
(We assume that a round consists of the events between two consecutive $\tally$ actions,
but that if a voter does not vote in round, their vote defaults to their
vote in a previous round, if any.)  It is easily seen that this interpretation is 
identity-oblivious. Thus $\archint \in \archspec_{\elecarch}$.  

Next, we describe a system $M$ with structured state and argue that it \TFF-complies with 
interpreted architecture $(\elecarch,\archint)$. 
We take the set of objects of $M$ to be 
the set $\Obj = \{v_1, \ldots, v_n, \bb\}$, where $v_i$ represents the election authority's 
record of the vote of voter $i$ and $b$ represents a bulletin board where the results of the 
election are broadcast to the voters. The objects $v_i$ and $\bb$ take a value in $\{\bot,Y,N\}$, with $\bot$ indicating 
that no vote has yet been made by the voter, or in the case of $\bb$, that the election authority has not yet 
announced a result. In the initial state, all objects take value $\bot$.

Access control on these objects is captured by the following ``capability" formatted table:
\begin{table}[h] 
\centerline{ 
\begin{tabular}{|c|c|c|} 
\hline
                    &    $V_i$       &  $\elecauth$       \\ 
                    \hline 
$\observe$  &  $\bb$ &  $v_1, \ldots, v_n, \bb$  \\ 
\hline 
$\alter$        &  $v_i$       & $\bb$    \\ 
\hline 
\end{tabular} }
\end{table} 
\\
It is straightforward to verify that this table satisfies condition AOI$'$ for 
policy $\elecarch$. We let each domain's observation consist of the values 
of its observable objects, i.e., $\obs_u(s) = s\restrict \observe(u)$, 
so that RM1 is satisfied trivially. 

The effect of the actions on the state is given by the following code: 
\begin{itemize} 
\item[] $\vyes^{V_i}$ : 
 $v_i := Y$ 
\item[] $\vno^{V_i}$ : 
$v_i := N$
\item[] $\tally$ : if $| \{ i~|~ v_i = Y\} | \geq n/2$ then $\bb:=Y$ else $\bb:= N$
\end{itemize} 
These definitions satisfy conditions RM2 and RM3. 
For RM2, note that the actions $\vyes^{V_i}$ and $\vno^{V_i}$ 
change $v_i$ in a way that depends only on the action. 
The action $\tally$ changes $b$ in a way that depends on $v_1\ldots, v_n$, 
but all of these objects are observable to $\elecauth$. 
For RM3, we have that the complement of $\alter(V_i)$ is $\Obj\setminus\{v_i\}$, 
but $V_i$'s actions change only $v_i$, and  the complement of $\alter(\elecauth)$ is 
$\Obj\setminus\{\bb\}$, and none of the $v_i$ are changed by $\elecauth$'s action $\tally$. 

To show that this system \TFF-complies with the interpreted architecture
$(\elecarch,\archint)$, we use conditions E1-E2 and I1$'$-I2$'$.
Let the functions $F$ be given by
$F(\announce)(s,\tally) =  (|\{ i~|~s(v_i) = Y\}| \geq n/2)$, i.e. , the output of $F$ is a 
boolean value that indicates whether the majority of the $v_i$ have value $Y$.  

We first show that these satisfy E1 and E2. Condition E2 is trivial, 
since it is never the case that $F(\announce)(s,a) = \epsilon$. 
For E1, we first claim that for all sequences $\alpha \in A^*$, 
we have that $(s_0\cdot \alpha)(v_i) =Y$ iff the latest action of 
$V_i$ in $\alpha$ is $\vyes$. The proof of this benefits from RM3: 
since $v_i \in \alter(\dom(a))$ iff $\dom(a) = V_i$, 
we have by a straightforward induction that 
$(s_0\cdot \alpha)(v_i) = (s_0\cdot \alpha\restrict \{V_i\})(v_i)$. 
The latest action of domain $V_i$ in $\alpha$ is the 
final action (if any) in  $\alpha\restrict \{V_i\}$, and $(s_0\cdot \alpha\restrict \{V_i\})(v_i) = Y$ just in case this 
action exists and equals $\vyes$. The claim now follows.

We now verify the conditions I1$'$ and I2$'$. 
Condition I1$'$ is trivial since we never have $F(\announce)(s,\tally) = \epsilon$. 
For condition  I2$'$, suppose that 
$\flowstoin[f]{\dom(a)}{u}$ with $f\neq \top$  and 
$\flowstoin[g]{\dom(b)}{u}$ with  $g\neq \top$. 
Then we must have $f=g = \announce$ and  $\dom(a) = \elecauth$, 
hence $a = \tally$ and $u= V_i$ for some $i$. 
Suppose additionally that 
$F(f)(s, a) = F(g)(t,b)\neq \epsilon$ and 
$x\in \observe(u)\cap (\alter(\dom(a))\cup \alter(\dom(b)))$ 
and $s(x) = t(x)$. 
Then $x\in \observe(V_i) \cap \alter(\elecauth)$, so $x = \bb$. 
By definition of $F$, we have $F(\announce)(s,\tally) = 
(|\{ i~|~s(v_i) = Y\}| \geq n/2)
= ( |\{ i~|~t(v_i) = Y\}| \geq n/2)
= F(\announce)(t,\tally)$. 
In case this boolean value is true, 
we have   $(s\cdot  a)(x) = Y = (t\cdot a)(x)$, 
otherwise $(s\cdot  a)(x) = N  = (t\cdot a)(x)$.  
In either case, $(s\cdot  a)(x)  = (t\cdot a)(x)$, as required. 

We remark that in this argument, we have used a function $F(\announce)$ that is not  an equivalent 
state-based encoding of the interpretation $\intn(\announce)$, but which is weaker than this interpretation. 
Correspondingly, in the implementation, the election authority reveals less information to the 
voters than the architectural interpretation permits. The architectural interpretation 
permits the election authority to reveal the  {\em number} of voters who have voted ``yes", but 
in the implementation the election authority only reveals whether this number is at least $n/2$. 
The notion of \TFF-compliance with an architecture allows this kind of weakening of information flows in the 
implementation.

\section{Enforcing Architectures on Concrete Platforms}\label{sec:enforce}

The access control model presented in Section~\ref{sec:access-control}
gives an abstract view of how an extended architecture might be enforced. 
It identifies a set of conditions whose satisfaction suffices
to ensure that a concrete system is compliant with the architecture. 
The process for verification of these conditions in specific settings
is likely to be dependent on the particulars of the implementation platform(s) being used.  
We briefly discuss a few of the possible platforms and the techniques that might be 
used to show that the access control conditions hold.

\Citet{BDRF08} survey
techniques to achieve \emph{separation} of components, that is, to
ensure that communication between components is in accordance with the
architecture.
In the case of unextended architectures, a very common technique to ensure 
compliance with the architecture is to map information flow edges to 
physical causality and use physical separation where there is no edge. 
Thus, the architecture $\hlarch$ is commonly enforced in military settings 
by mapping $H$ and $L$ to distinct processors and/or networks and 
using trusted devices (data diodes) to ensure a one-way information flow from 
$L$ to $H$. The Starlight Interactive Link \citep{starlight-interactive} is a trusted device that 
can be added to such an implementation to extend it to an implementation of the architecture 
 {\strltArch}. 

One of the longstanding objectives of research on military-grade security 
has been 
to avoid the redundancy and consequent expense of such physical implementations, 
through the use of implementations that enable different security levels to share
resources such as memory, processors and networks. 
An implementation technique that forms the basis for much work in MILS security is
the use of separation kernels, which are highly simplified operating systems with the sole 
functionality of enforcing an information flow policy. 

The key mechanism used to achieve this is typically enforcement of an 
access control policy by careful management of hardware access control 
settings and processor modes to ensure that when a process runs, it may read and write 
only memory regions authorized by the access control policy. 
Attempts to read or write memory regions that violate the policy
are denied by the hardware access control measures. This 
ensures that the reference monitor conditions RM2 and RM3 are satisfied. 
RM1 could be guaranteed by ensuring that the hardware access control 
setting ensures that all peripheral devices with which a user may interact 
are mapped to memory regions that are associated to the domain of that user.  
Use of separations kernels introduces the risk that there are covert channels, but 
 much progress has been made in recent years towards
formal proofs that  separation kernels enforce 
an information flow policy (e.g., \citet{GreveWV03,HeitmeyerALM06,MurrayMBGBSLGK13}). 

An alternative to the use of hardware access control to ensure satisfaction of the 
reference monitor conditions is to verify, e.g., using static analysis methods, 
that the code in each domain (e.g., the code describing how the actions of the domain
affect the state) reads and writes only locations that are permitted by the policy. 
Once this has been done, this code can safely run free 
of hardware access control. This approach is taken in the  Singularity system \citep{HuntL07}. 

The specifics of the static analysis techniques 
to be used to enforce the reference monitor conditions will be very language dependent, but can benefit from
programming language techniques including type safety and encapsulation constructs
including objects, object ownership, and aspects \citep{Kiczales96}.
Some of the abstract reasoning in the examples above is already suggestive of such techniques. 
For example, in the election example, the Election Authority can  be viewed as owning the 
objects $v_i$ and voter $V_i$ can be viewed as having a capability to call a method on object $v_i$. 
The code pattern used in the downgrader example for actions $a\in A_T$ could be enforced
using aspect-oriented techniques.

Once a basic (unextended) architecture has been shown to be enforced by the 
implementation, it remains to demonstrate that the trusted components in the 
architecture satisfy their local constraints  (e.g., constraints I1 and I2). 
As these local constraints are
application specific, and implementation dependent, it seems unlikely
that a single methodology will suffice. 
We expect that theorem proving, model checking, and 
language-based  information
flow techniques \citep{sm-jsac}
may all be used to
provide assurance of 
satisfaction of the 
local constraints  and  the filtering requirements introduced by our extended
architectures. 

We note that our framework is highly expressive, and it 
may not always be possible to show compliance with a filtering requirement 
using only local information. 
For example, 
directly implementing the filtering requirements  of the downgrader architectural
specification $(\downarch,\archspec_\downarch)$ by means of controls at the  downgrader component 
$D$ would seem to require the cooperation 
of High level components  (e.g., provision of secure provenance information) 
to ensure that the downgrader does not release information concerning  $ C$. 
(The refined architectural specification $(\downarchr,\archspec_{\downarchr})$ avoids the need for this 
cooperation by shifting the trust boundary.)
Identifying sufficient conditions for local verification of compliance is an interesting topic for future research.

\section{Related Work} 
\label{sec:related}

The most closely related work is that of \citet{wiini,arch-refinement}, who defines TA-security and considers refinement of architectures on the basis of this semantics.  Our contribution is to show that the definition of TA-security supports derivation of global information security properties, to extend TA-security to architectures that include filter functions, and to develop an account of architectural refinement for these extended architectures. The extension provides a way to specify the behavior of trusted components in a system: intuitively, if a component is the source of an edge in the architecture labeled by a filter function, then the component is trusted to ensure that the interpretation of the filter function limits the information that may flow along the edge. We have presented a number of examples that show that interesting information-theoretic global security properties can be derived from the very abstract statement that a system complies with such an extended architecture and a set of additional local properties. We note that these global properties are more general and application specific than the very particular property ``Low does not know any High secrets'' that is most often considered in the literature.

Other work has sought to formally describe system architecture (e.g., 
\citet{aadlv2} and Acme \citep{Garlan00AcmeChapter}), and
reason about the properties of systems conforming to a given
architecture. There are
many software engineering concerns that can be reasoned about in
architectural design, such as maintainability, 
and reliability.
This work focuses on reasoning about the information security of systems, and, as such, our architectures specify local constraints on information that may be communicated between components. The local constraints allow the inference of global information security properties. This work is complementary to work on other aspects of system design.

Relatively little theoretical work takes an architectural perspective on information security. 
One interesting line of work \citep{HanssonFM08} that takes a similar perspective to ours is conducted
in the context the architectural modeling framework AADL.
This work is based on the Bell La Padula model~\citep{BLP}. 
In a similar spirit are works on Model Driven Security, 
which extend UML  with security modeling notations. 
\Citet{BasinDL06} focus on a UML extension for role-based 
access control policies and model transformations to implementation infrastructures such 
as Enterprise Java Beans or .NET.  
\Citet{Jurjens05} extends UML with a focus on
reasoning about secrecy in distributed applications employing 
security protocols.
None of these approaches use the 
application-specific abstract non-interference-style semantics that underpins our 
contribution, nor do they target the type of reasoning envisaged in the MILS community 
for development of high-assurance systems built on infrastructure such as separation kernels.

Standard notions of refinement reduce the possible behaviors of a system. However, arbitrary behavioral refinements of a system will, in general, not preserve information security properties \citep{Jacob89}. Previous work has investigated restricting behavioral refinements to preserve information security properties \citep{GCS,OHalloran,rosco_95,BFPR03,
Morgan06}. Architectural refinement as used in this paper allows only refinements that reduce the information communicated between system components. As shown in Section~\ref{sec:refine}, this notion of refinement preserves certain information security properties. Information security can also be preserved under refinement by modifying the refined system \citep{Jacob89,Mantel01}. Other work distinguishes nondeterminism of a system specification from nondeterminism inherent in the system, and allows refinement only of specification nondeterminism \citep{SS06,Jurjens05,Bibighaus}. These works typically consider only the  simple policy $L \flowsto H$, rather than the more general intransitive policies considered here.

Most work on architectural refinement is not directly concerned with information security.  \Citet{ZA}  propose architecture refinement patterns for Multi-Level Secure systems development, but do not provide formal semantics. A series of papers \citep{moriconi2,moriconi1,moriconi3} express architectural designs as logical theories and refinement as a mapping from an abstract theory to a concrete theory. This approach is used to establish security properties in variants of the X/Open Distributed Transaction Processing architecture, using the Bell La Padula model \citep{BLP}, which lacks the kind of information flow semantics that we have studied here. It is not clear if this approach is sufficiently expressive to represent architectural refinement as used in this paper, and reason about the preservation of information security under such refinement.

Preservation of information flow properties under system composition has, like refinement, been considered problematic. In general, the composition of two secure systems is not guaranteed to be secure. The reason is essentially the same as for refinement: composition reduces the set of possible behaviors of a system, enabling an observer to make additional deductions.

A
number of approaches have been developed that 
allow security of a composed system to be derived from security of its
components.  These include use of a stronger 
definitions of security such as restrictiveness \citep{McCullough90} or bisimulation-based nondeducibility on compositions \citep{FocardiG94}.
\citet{McLean96} proposes a framework for specifying and reasoning
about  system composition, and the preservation of possiblistic
security properties.
In this work, we are not concerned with showing that security properties that hold of components also hold of a composite system. Instead we are concerned with proving global security properties, and identifying local constraints that components must satisfy.
The literature on preservation of information flow security under composition has 
also largely limited itself to the simple policy $L \flowsto H$.

To some extent, process algebraic operations can be viewed as 
expressing architectural structure. For example, one could 
take the view that a process constructed as the parallel composition of two  processes $P$ and $Q$, 
with actions $A$ of the composition then hidden, corresponds to 
an architecture in which $P$ and $Q$ are permitted to 
interact, but 
which the environment is not permitted to influence through the set of 
actions $A$. However, the semantics of these operators usually do 
not preserve this view: typically this composition  is understood 
in terms of its possible behaviors with respect to the actions 
that have not been hidden, and the fact that the system has been 
composed out of two components permitted to 
interact in a particular
way is lost in the meaning of the composition.

Downgrading has historically been one of the motivations for 
generalizing the notion of noninterference to 
intransitive policies. \Citet{RoscoeG99} argued against the 
ipurge-based semantics for noninterference on the grounds that the meaning it gives to the  
downgrader policy $H \flowsto D\flowsto L$ is too permissive. According to this semantics, 
any action  by the downgrader $D$ ``opens the floodgates,'' in the sense that it allows all information about the 
High security domain $H$ to flow to the low security domain $L$.  
\Citeauthor{RoscoeG99} proposed to deal with this issue by making the semantics of intransitive 
noninterference significantly more restrictive, in effect reverting to the purge-based definition of \citet{gm82,gm84}. Our approach to downgrading, using a filter function, provides an alternative approach 
that enables explicit specification of the information permitted to be released by the downgrader.

More recent work on downgrading has concentrated on downgrading in the setting of language-based security. \Citet{ss05} briefly survey recent work on downgrading in lang\-uage-based settings, and propose several dimensions of downgrading, and prudent principles for downgrading. They regard intransitive noninterference as specifying \emph{where} (in the security levels) downgrading may occur. 
Since we interpret security levels as system components, our architectures specify where in the system downgrading may occur. The filter functions that we propose in this work specify \emph{what} information can be downgraded and \emph{when} this may occur. Thus, our work combines the \emph{what}, \emph{where},  and \emph{when} dimensions of downgrading.
Recent work also considers multiple dimensions of downgrading, including \citet{BartheCR08}, \citet{bnr08}, \citet{MantelR07} and \citet{AskarovS07b}.

Recent work \citep{AskarovS07,bnr08} considers ``knowledge-based'' approaches to downgrading in language-based settings. However, they do not reason about security properties as general and application specific as used in this paper.  \citet{oneill-thesis} uses epistemic logic to specify many information security properties, but does not directly consider downgrading.

\section{Conclusion}\label{sec:conclusion}

Through the examination of 
a number of 
examples, we have shown that
strong information security properties can be proven about a system
from a high-level architectural description of the system.  Any system
that complies with the architecture will satisfy the information
security properties that can be proven about the architecture.

We extended the notion of 
system 
architecture to allow finer-grain
specification of what information may be sent between components.
This enables the proof of stronger security properties, while
continuing to providing the benefits of using a high-level
architectural description.
We generalized the notion of architectural refinement
\citep{arch-refinement} for the extended architectures. Certain
security properties are preserved by architectural refinement.

The MILS vision is to build high-assurance systems with well-understood
security properties by composition of COTS infrastructure and trusted
components. This work brings us closer to that goal by demonstrating
that it is possible to compositionally derive strong,
application-specific,
information-flow
security properties from high-level system specifications.

\bibliography{bib}

\end{document}